\documentclass[aps, prx, superscriptaddress, reprint, twocolumn]{revtex4-2}
\usepackage{graphicx}
\usepackage{ulem}
\usepackage{amsthm,amsfonts,amstext,amssymb,amsmath}
\usepackage{latexsym}
\usepackage[dvipsnames]{xcolor}
\usepackage{braket}
\usepackage{bbold}
\usepackage{bm}
\usepackage{placeins}
\usepackage{enumerate}
\usepackage{standalone}
\usepackage{float}
\usepackage{xcolor}
\usepackage{mathtools}
\usepackage[colorlinks=true,
  linkcolor=blue,
  citecolor=blue,
  urlcolor =blue]{hyperref}
\newcommand{\be}{\begin{equation}}
\newcommand{\ee}{\end{equation}}
\newcommand{\f}{\frac}
\newcommand{\D}{D}
\newcommand{\DA}{D_A}
\newcommand{\DB}{{{D}_B}}
\newcommand{\diff}{\text{d}}
\graphicspath{{Figures/}} %Setting the graphicspath

\DeclareMathOperator{\sinc}{sinc}
\DeclareMathOperator{\sgn}{sgn}

\newcommand{\expect}[2][]{\ensuremath{{\mathbb{E}}_{#1}\left[ #2 \right]}}
\newcommand{\var}[2][]{\ensuremath{{\text{Var}}_{#1}\left[ #2 \right]}}

\newcommand*{\ketbra}[2]{\ensuremath{\ket{#1}\bra{#2}}}

\newcommand*{\tr}[2][]{\ensuremath{\textrm{Tr}_{#1}\left[ #2 \right]}}

\newcommand{\SFF}[1][]{\ensuremath{K_{#1}}}

\begin{document}
\title{Probing many-body quantum chaos with quantum simulators}

\author{Lata Kh Joshi}
\thanks{These authors contributed equally.} \affiliation{Center for Quantum Physics, University of Innsbruck, Innsbruck A-6020, Austria}
\affiliation{Institute for Quantum Optics and Quantum Information of the Austrian Academy of Sciences,  Innsbruck A-6020, Austria}
\author{Andreas Elben}
\thanks{These authors contributed equally.}
\affiliation{Center for Quantum Physics, University of Innsbruck, Innsbruck A-6020, Austria}
\affiliation{Institute for Quantum Optics and Quantum Information of the Austrian Academy of Sciences,  Innsbruck A-6020, Austria}
\affiliation{Institute for Quantum Information and Matter and Walter Burke Institute for
Theoretical Physics, California Institute of Technology, Pasadena, CA 91125, USA}
\author{Amit Vikram} 
\affiliation{Joint Quantum Institute, University of Maryland, College Park, MD 20742, USA}
\affiliation{Condensed Matter Theory Center, Department of Physics, University of Maryland, College Park, MD 20742, USA}
\author{Beno\^it Vermersch}
\affiliation{Center for Quantum Physics, University of Innsbruck, Innsbruck A-6020, Austria}	
\affiliation{Institute for Quantum Optics and Quantum Information of the Austrian Academy of Sciences,  Innsbruck A-6020, Austria}
\affiliation{Univ.\  Grenoble Alpes, CNRS, LPMMC, 38000 Grenoble, France}
\author{Victor Galitski}
\affiliation{Joint Quantum Institute, University of Maryland, College Park, MD 20742, USA}
\affiliation{Condensed Matter Theory Center, Department of Physics, University of Maryland, College Park, MD 20742, USA}
\author{Peter Zoller}
\affiliation{Center for Quantum Physics, University of Innsbruck, Innsbruck A-6020, Austria}
\affiliation{Institute for Quantum Optics and Quantum Information of the Austrian Academy of Sciences,  Innsbruck A-6020, Austria}
	
\begin{abstract}
The spectral form factor (SFF), characterizing statistics of energy eigenvalues, is a key diagnostic of many-body quantum chaos. In addition,  partial spectral form factors (PSFFs) can be defined which refer to subsystems of the many-body system.  They  provide unique insights into energy eigenstate statistics of many-body systems, as we show in an analysis on the basis of random matrix theory and of the eigenstate thermalization hypothesis. We propose a protocol that allows the measurement of the SFF and PSFFs in quantum many-body spin models, within the framework of randomized measurements. Aimed to probe dynamical properties of quantum many-body systems, our scheme employs statistical correlations of local random operations which are applied at different times in a single experiment. Our protocol provides  a unified testbed to probe many-body quantum chaotic behavior, thermalization and  many-body localization in closed quantum systems which we illustrate with numerical simulations for Hamiltonian and Floquet many-body spin-systems.
\end{abstract}
\maketitle 
\section{Synopsis}
%\\(or Introduction and overview?)}
\label{sec:overview}
The ongoing development of quantum simulators provides us with unique opportunities to study quantum chaos in many-body systems, and its connections to random matrix theory (RMT) \cite{HaakeBook}  and Eigenstate Thermalization Hypothesis (ETH) \cite{deutsch1991eth, srednicki1994eth} in  highly  controlled laboratory settings. This refers to not only the experimental realization of engineered Hamiltonian dynamics of isolated quantum  systems, which can be tuned from integrable to non-integrable, but also the ability to measure  novel observables beyond standard low-order correlation functions \cite{Blatt2012,LS2019,Monroe2021,Browaeys2020,Kjaergaard2020}. It includes recent measurements of the growth of entanglement entropies in quantum many-body systems \cite{Islam2015,Kaufman2016,Brydges2019,Vovrosh2021} as well as of the decay of out-of-time-ordered correlation functions \cite{Garttner2017,Li2017,Wei2018,Landsman2019,Nie2019,Joshi2020,Mi2021}. In this work, our interests lie in developing experimentally  feasible probes of universal RMT  predictions  for the \textit{statistics of  energy eigenvalues} \cite{Wigner1955, Dyson1962, Casati1980, BGS1984, Mehta2004,HaakeBook} and  predictions of the ETH  for the \textit{statistics of energy eigenstates} \cite{deutsch1991eth, srednicki1994eth,srednicki1999eth,rigol2008eth,DAlessio2016,deutsch2018eth,subETH} of  quantum chaotic many-body systems. 
Using these probes, we are further interested in distinguishing  many-body localized (MBL) systems \cite{nandkishore2015mbl,Abanin2019} from the chaotic ones, where in the former the eigenvalue statistics are described by the Poisson distribution \cite{BerryInt, ponte2015floqueteth, Prakash2020} and the ETH is violated. 

In this paper, we identify the spectral form factor (SFF), and its generalization to partial SFF (PSFF), as observables of interest to reveal energy level and eigenstate statistics. The SFF is defined in terms of the time evolution operator of the quantum many-body system of interest and provides us with statistics of energy levels \cite{HaakeBook}. The PSFF will be defined in terms of the time evolution operator restricted to a subsystem of the many-body system, and contains  information on both, the statistics of energy eigenvalues and energy eigenstates. We derive analytic expressions for the PSFF in Wigner-Dyson random matrix ensembles. More generally, in chaotic quantum many-body systems, the ETH imposes constraints on the statistics of eigenstates, which are however typically violated in localized systems. Therefore, the PSFF provides a direct probe of eigenstate thermalization and localization.

The goal of the present work is to develop measurement protocols for the SFF and PSFF in quantum spin models of arbitrary dimension, as realized for instance with trapped ions \cite{Blatt2012,Monroe2021}, Rydberg atoms \cite{Browaeys2020} and superconducting qubits \cite{Kjaergaard2020}.   We extend the \textit{randomized measurement} toolbox \cite{VanEnk2012,Elben2018,Vermersch2018,Vermersch2019,Ketterer2019,Elben2020_xPlatform,Elben2020_SPT,Huang2020,Elben2020_Mixed,Cian2020,Zhou2020,Vitale2021,Garcia2021,Rath2021,Yu2021,Neven2021,Knips2020,Imai_2021,Rath2021a} to infer the SFF and PSFF from statistical correlations of local random operations applied at different times in a single experiment. In contrast to the previous works utilizing randomized measurements to infer properties of many-body quantum states \cite{VanEnk2012,Elben2018, Vermersch2018,Brydges2019,Ketterer2019,Elben2020_xPlatform,Elben2020_SPT,Cian2020, Huang2020,Elben2020_Mixed,Knips2020, Zhou2020,Imai_2021, Yu2021,Rath2021,Vitale2021,Neven2021,Tran2015,Tran2016,Satzinger2021,Mi2020,Rath2021a} and  (out-of-time-ordered) correlation functions of Heisenberg operators \cite{Vermersch2019,Joshi2020}, the present protocol yields, with the SFF and PSFF, genuine properties of the time evolution operator. We emphasize that the present protocol is ancilla-free. This is in contrast to Ref.~\cite{Vasilyev2020} where a measurement scheme for the SFF was proposed requiring time evolution of an extended system comprising of the quantum simulator and an auxiliary spin. 

Our protocol to measure the PSFF and SFF in a quantum simulation experiment can be readily implemented in existing experimental platforms. It requires only to implement local (single-spin) random unitaries and projective measurements, which have been previously demonstrated with high fidelity \cite{Brydges2019,Nie2019, Joshi2020,Satzinger2021}. Interestingly, in our protocol we obtain the SFF and PSFF from the same experimental dataset. This enables an efficient scheme to test universal RMT predictions for the energy eigenvalue spectrum and, at the same time, to probe properties of the energy eigenstates  and thermalization via ETH.

We now turn to an overview of the main results of the paper. We start by recalling  the standard definition of the SFF, define the PSFF and describe their estimation using randomized measurement protocol. We then illustrate the key features of the (P)SFF and demonstrate our measurement protocol using an example of a chaotic, periodically kicked spin$-1/2$ model. We will argue on the basis of this example and show in later sections with detailed analytical and numerical calculations that the SFF and PSFF provide unique insights into the  eigenvalue and eigenstate statistics of quantum many-body systems. 

\subsection{Spectral form factor}
\label{sec:sff-example}
The SFF in  a many-body quantum system with time-independent Hamiltonian $H$  and energy spectrum $\{E_j\}$ is defined as the Fourier transform of the two-point correlator of the energy level density  \cite{HaakeBook}. It can be expressed as
 \begin{align}
\!K(t)&\equiv \frac{1}{D^2}\,  \overline { \sum_{i,j}  e^{i (E_i-E_j) t} }  = \frac{1}{D^2}\, \overline{ \tr{T(t)}\tr{T^\dagger(t)}}~. \!\! \!
\label{eq:sff}
\end{align}
Here, we normalize $K(t)$ such that $K(0)=D^{-2} \tr{\mathbb{1}}^2=1$, with $D$ the Hilbert space dimension and have defined the unitary time-evolution operator $T(t)\equiv \exp(-iHt)$. The overline denotes a possible disorder or ensemble average over an ensemble of $T(t)$, which is needed due to non-self-averaging behavior of the SFF \cite{prange1997sff}. Replacing the energies $E_i$ with quasi-energies, this definition carries  over to Floquet models with time-periodic evolution operator $T(t=n\tau )=V^n$ ($n \in \mathbb{N}$) and $V$ the  Floquet time evolution operator for a single period $\tau$ \footnote{Denoting the set of eigenvalues of the Floquet operator $
V$ with $\{ \exp(-iE_i \tau) \}$, the quasi-energy eigenvalues $\{ E_i \}$ are only defined up to multiples of the driving frequencies $\omega=2\pi\tau^{-1}$. We fix them to lie in the interval $[ 0, \omega]$. }. 

The SFF is a probe of the universal properties of the statistics of energy eigenvalues in chaotic and localized systems. Lately, it  has played a key role in a variety of different fields, interconnecting quantum chaos \cite{HaakeBook}, quantum dynamics of black holes  \cite{Cotler2017, Cotler2017a, Saad2018, Gharibyan2018}, condensed matter systems \cite{Kos2018, chan2018a, chan2018b, Bertini2018, Suntajs2019, Parameswaran2019, shackleton2020, Liao2020, winer2020,  Sierant2020,Sierant2020b}, and the dynamics of thermalization \cite{Reimann2016}.
In Fig.~\ref{fig:figure1}(a), we illustrate its behavior in  the  context of a  periodically kicked spin-$1/2$ system.  The time evolution operator $T$ at integer multiples $n\in \mathbb{N}$ of driving period $\tau$ is  given by $T(t=n\tau)=V_3^n$ with,
\begin{equation}
V_3= e^{-i H^{(x)}\tau/3} e^{-i H^{(y)}\tau/3}e^{-i H^{(z)}\tau/3} ~. 
\label{eq:Floquet-illustration-V3}
\end{equation}

Here, the  Hamiltonians $H^{(x, y, z)}$ contain nearest-neighbor interactions with strength $J=3\tau^{-1}$ and disordered transverse fields with strength $h_i^{(x, y, z)}\in [-J, J]$,
\begin{align*}
H^{(x,y,z)}= J\sum_{i=1}^{N-1} \sigma^{(x,y,z)}_i \sigma^{(x,y,z)}_{i+1} + \sum_{i=1}^N h_i^{(y,z,x)} \sigma^{(y,z,x)}_i~,
\end{align*} 
and $\sigma^{a}$ [$a\in (x, y, z)$] denote the Pauli matrices. We have denoted the number of spins with $N$ such that $D=2^N$.
 An ensemble average is naturally performed by averaging over many instances of $T(t=n\tau)=V_3^n$, each with local disorder potentials $h_i^{(a)}$ sampled independently from the uniform distribution on $[-J, J]$. 
\begin{figure}[t]
\includegraphics[width=\linewidth]{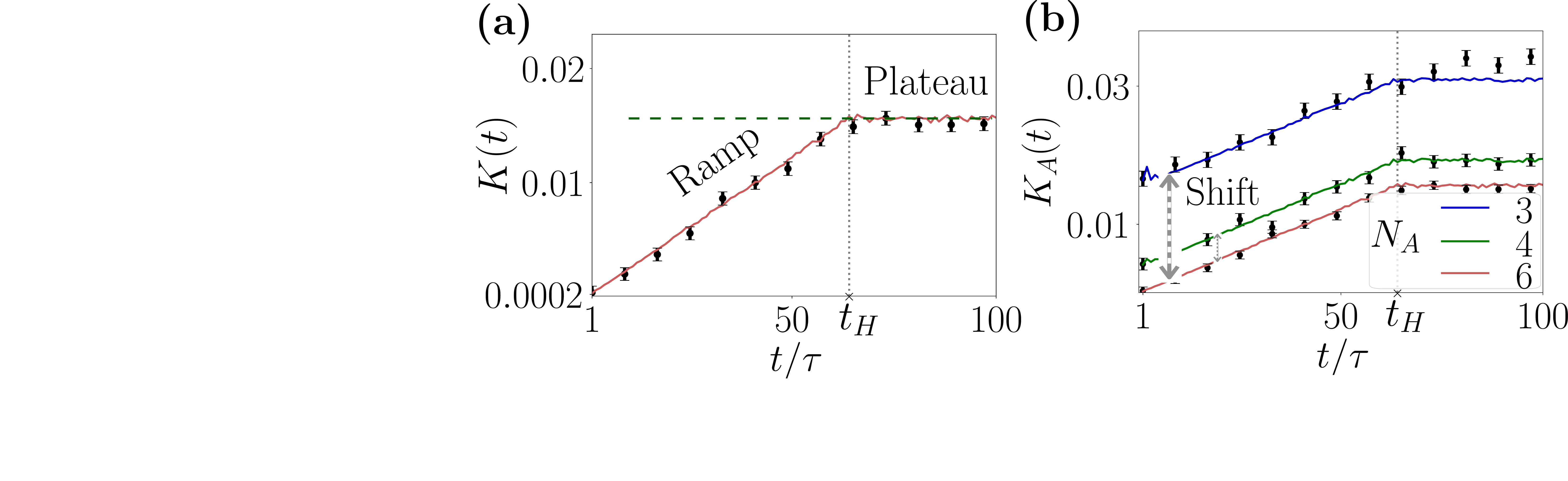}
  \caption{  \textit{Illustration of the characteristic properties of the SFF and PSFF using {the} chaotic spin-$1/2$ Floquet model $V_3$.} (a)  We display the SFF $K(t)$  for the Floquet model $V_3$ with $N=6$ qubits as a function of time $t$. We observe characteristic features such as the ramp between $t\sim \tau$ to $t=t_H= 2^N \tau$ and a plateau for $t>t_H$.  (b) For the PSFF $K_A(t)$ we observe ramp, plateau and, in particular, a constant, additive shift of the PSFF compared to the SFF, which depends on the subsystem size $N_A$ of the subsystem $A$. We have chosen subsystems A from the middle of the total system. In both, the colored lines show the numerically calculated SFF and PSFFs, averaged over $8000$ disorder realizations.  In addition, we illustrate our measurement protocol (see Sec.~\ref{sec:protocol}) by simulating $M=2\times 10^5$ experimental runs (single-shot randomized measurements) at each time and display the estimated SFF and PSFF as black dots with associated error bars.  The dashed green line in panel (a) sketches the form of the SFF generically expected in a many-body localized model.}
    \label{fig:figure1}
    \end{figure}

As shown in Fig.~\ref{fig:figure1}(a), the SFF $K(t)$ for this model and choice of parameters exhibits a period of linear growth, before transitioning to a constant at time $t/\tau \approx  \D=2^N$. 
This \textit{ramp-plateau} structure of the SFF is a characteristic feature of quantum chaotic systems \cite{Leviandier1986,Guhr1998,HaakeBook}, originating from  (quasi-)energy level repulsion and  spectral rigidity \cite{Leviandier1986}, and is predicted by RMT \cite{HaakeBook,Mehta2004}. In particular, as we briefly review in App.~\ref{app:SFF-RMT}, RMT for time evolution operators $T(t=\tau n) = V^n$, with $V$ from the circular unitary ensemble (CUE), yields 
\begin{align}
    K(t) = \frac{1}{D} \begin{cases}
    {t}/{t_H},  & 0 <  t \leq t_H \\
    1,  &  t>t_H~.
    \end{cases}
        \label{eq:SFF_cue}
\end{align}
Here, the slope of the ramp and the onset of the plateau is determined by the Heisenberg (or plateau) time $t_\text{H}$ which is connected to the mean inverse   spacing of adjacent (quasi-) energies. It typically scales with the Hilbert space dimension $t_\text{H}/\tau \sim D$   --- for $V$ from CUE, $t_H/\tau=\D$ \cite{HaakeBook,Gharibyan2018}. Thus, the SFF is expected to drop with increasing Hilbert space dimension $D=2^N$, as $D^{-2}$ at times $1\lesssim t/\tau \ll D$ and as $D^{-1}$ at times $t/\tau \gtrsim D$.
Fig.~\ref{fig:figure1}(a) shows that the SFF $K(t)$ for the $V_3$ model closely follows the CUE  prediction after the initial few time steps. This time after which the many-body model shows the same SFF as the one in RMT is known as the Thouless time $t_{\rm{Th}}$ \cite{Kos2018}. For the model $V_3$ we note that $t_{\rm{Th}}\approx 5\tau$ (see also Sec.~\ref{sec:psffNumerical}). Therefore, the quasi-energy eigenvalues of the Floquet operator $V_3$  exhibit Wigner-Dyson  statistics (see also Ref.~\cite{Vasilyev2020}).

In contrast to the example of a chaotic system 
$V_3$ presented above, the energy eigenvalues of integrable and localized models are known to exhibit Poissonian statistics \cite{BerryInt, nandkishore2015mbl,ponte2015floqueteth, Prakash2020}. This corresponds to a flat SFF  without a ramp which is, after an initial transient regime, constant in time \cite{HaakeBook},  
$K(t\gg 0)=1/D$. This is schematically shown in Fig.~\ref{fig:figure1}(a) with green dashes. These distinct features of the SFF have been pivotal in characterizing many-body chaotic and MBL phases \cite{Suntajs2019, Parameswaran2019, Vasilyev2020}.

\subsection{Partial Spectral Form Factor} 
\label{sec:psff-example}
The SFF reveals information on the statistics of \mbox{(quasi-)} energy eigenvalues. It is however by definition insensitive to properties of the (quasi-) energy eigenstates. In this subsection, we define the PSFF and illustrate its essential properties connected to properties of eigenvalues and eigenstates.

For a fixed subsystem $A\subseteq \mathcal{S}$  of the total system $\mathcal{S}$ with complement $B$ ($A\cup B = \mathcal{S}$) and Hilbert space dimensions $D_A$ and $D_B$ respectively ($D=D_AD_B$), we define the PSFF as
\begin{align}
\label{eq:psff}
K_{\mathrm A}(t)
&\equiv\frac{1}{D D_A} \, \overline { \sum_{i,j}  e^{i (E_i-E_j) t}  \tr[B]{\rho_B(E_i) \rho_B(E_j)} } \nonumber \\ &
=\frac{1}{DD_A}\, \overline{\tr[B]{ \tr[A]{T(t)} \tr[A]{T^\dagger(t)}}}~,
\end{align}
where  $\rho_B(E_i)= \tr[A]{\ketbra{E_i}{E_i}}$  denotes the reduced density matrix obtained after partial trace of the eigenstate $\ket{E_i}$ of the Hamiltonian $H$ (the Floquet time evolution operator $V$)  with energy (quasi-energy)  $E_i$.
Here, the normalization of $K_A(t)$ is chosen such that $K_A(0)= \tr[B]{ \tr[A]{\mathbb{1}}^2}/(DD_A)=1$.
Hence, the SFF and PSFF coincide when $A=\mathcal{S}$, i.e.\ $K_{\mathrm A = \mathcal{S}}(t)=K(t)$. We emphasize that for $A\subset \mathcal S$, the PSFF $K_A(t)$ contains non-trivial contributions from the eigenstates $\ket{E_i}$: We obtain terms of the form $\mathrm{Tr}(\rho_B(E_i)^2)$ and $\mathrm{Tr}(\rho_B(E_i)\rho_B(E_j))$ ($i \neq j$) which correspond to  the purity and   overlap of   reduced eigenstates. As shown below, a measurement of the PSFF allows  to extract these purities and overlaps, averaged over spectrum and ensemble,  i.e.\ allows to characterize (second-order moments of) the statistics of eigenstates.

We  remark that  $K_A(t)$ has been previously discussed as a topological invariant in the classification of symmetry-protected  matrix product unitaries  in Ref.~\cite{Gong2020a}. Its limiting cases for special subsystems ($A$ or $B$ consisting of a single site, in the limit of a large local Hilbert space dimension) have been used to study matrix elements of local operators in the energy eigenbasis in 1D Floquet circuits, with comparisons to random matrix predictions for eigenstate statistics in these subsystems (as a special case of ETH) \cite{GarrattChalker}. 

In this work, we identify a general \textit{shift-ramp-plateau} structure of the PSFF, which reveals a direct connection to ETH contained in the subsystem dependence of the PSFF. In Fig.~\ref{fig:figure1}(b),  we display the PSFF for the Floquet model \eqref{eq:Floquet-illustration-V3} for various subsystems $A$, where $N_A$ denotes number of qubits in the subsystem such that $D_A=2^{N_A}$. We first note that the PSFF also has a ramp and plateau, similar to the full SFF. 
The slope of the ramp is nearly identical for the displayed subsystem sizes $N_A\gtrsim N/2$ $=3$, which holds more generally for $D_A \gg 1$ in the CUE model, and the onset of the plateau in the PSFF takes place at the Heisenberg time $t_H$. Crucially, we find that, at late times comparable to the onset of the ramp, there is a subsystem dependent additive shift of the PSFF $K_A(t)$ compared to the full SFF $K(t)$. 

Similar to the case of the full SFF, we can compare the behavior of the PSFF to predictions of RMT. As detailed in Sec.~\ref{sec:psffAnalytical}, we find that RMT yields for time evolution operators $T(t=\tau n) = V^n$, with $V$ from the  CUE, and sufficiently large subsystems $A,B$, ($\DA,\DB\gg 1$),
\begin{align}
    K_A(t) = \frac{1}{D_A^2} + \frac{1}{D} \begin{cases}
    {t}/{t_H},  & 0 <  t \leq t_H \\
    1,  &  t>t_H~.
    \end{cases}
    \label{eq:pSFF_cue}
\end{align}
As shown in Fig.~\ref{fig:figure1}(b), and analyzed in  detail by further numerical studies in Sec.~\ref{sec:psffNumerical}, the PSFF (and SFF) for the $V_3$ model follows closely the RMT predictions. This indicates that both \textit{(quasi-) energy eigenvalues and eigenstates} of $V_3$ exhibits the Wigner-Dyson statistics of the CUE. We remark that this is consistent with previous works demonstrating   that (sub-)systems of  chaotic Floquet systems   thermalize to infinite temperature states as per RMT \cite{Regnault2016,PhysRevX.4.041048,lazarides2014floqueteth,ponte2015floqueteth,kim2014floqueteth, GarrattChalker}.

\textit{Partial spectral form factor and eigenstate thermalization hypothesis} –  Using the example of a chaotic Floquet model, we have illustrated above the essential features of the  PSFF in chaotic quantum systems. In Sec.~\ref{sec:psffAnalytical}, we analyze its behavior in  detail invoking subsystem ETH \cite{subETH} for the reduced eigenstates,  which is  a conjecture regarding the \textit{distribution of eigenstates} responsible for the thermal behavior (in the standard sense of ETH) of few-body observables in chaotic systems.

By separating out the components of the reduced density matrix into maximally mixed, smooth and fluctuating parts as a function of energy,  a generic late time expression for PSFF can be obtained. From here, we later conclude that  the features of the ramp, plateau and shift are  generic features of the PSFF in chaotic quantum many-body systems. These features are  directly connected to the spectrum and ensemble averages of the  subsystem purities $\mathrm{Tr}_B(\rho_B(E)^2)$ and of the overlaps of reduced eigenstates $\mathrm{Tr}_B(\rho_B(E_i)\rho_B(E_j))$. Furthermore,  the magnitudes of these features in the chaotic systems follow specific constraints when the eigenstates satisfy subsystem ETH, see Sec.~\ref{sec:ethConstraints}. In particular, we show that this shift, connected to the average overlaps, enables the detection of thermalization of eigenstates in the framework of subsystem ETH.

Let us take for instance the shift seen in the Fig.~\ref{fig:figure1}, defined precisely in terms of the fluctuating part of the density matrix later in Sec~\ref{sec:pSFFchaotic}. For chaotic models, the shift can be identified as the time independent constant during the linear ramp phase, and for $D_A \ll D$ it is approximated by $K_A(t_0)-K(t_0)$ where $t_{Th} < t_0 \ll t_H$.
If the eigenstates follow ETH,  it is expected that, 
\begin{equation}
K_A(t_0)-K(t_0)\approx O\left(\frac{1}{D_A^2}\right)~.
\label{eq:numerical-shift}
\end{equation}
This can be noted for the CUE in the Eqs. \eqref{eq:SFF_cue} and \eqref{eq:pSFF_cue} as well as for the $V_3$ model in Fig.~\ref{fig:figure1}, where the shift above SFF is seen to be increasing as the $N_A$ decreases and is found to follow Eq.~\eqref{eq:numerical-shift} (see Sec.~\ref{sec:psffNumerical} for more numerical details). On the other hand, for eigenstates which do not thermalize, the time independent shift above SFF is generically much larger than $ O(1/D_A^2)$.

As illustrated above, the SFF and PSFF of a quantum many-body system provide crucial insights into the statistics of energy eigenvalues and eigenstates, which results in  a joint observation of chaos  and validity of ETH. The question arises of how to probe the SFF and PSFF in today’s quantum devices. In the next subsection, we present our measurement protocol which can be directly implemented in state-of-the-art quantum simulation platforms realizing lattice spin models. It builds on the toolbox of randomized measurements. 

\subsection{Randomized measurements of spectral form factors}
\label{sec:protocol}

Initially, randomized measurements have been proposed and experimentally implemented to characterize many-body quantum states \cite{VanEnk2012,Elben2018, Vermersch2018,Brydges2019,Ketterer2019,Elben2020_xPlatform,Elben2020_SPT,Cian2020, Huang2020,Elben2020_Mixed,Knips2020, Zhou2020,Imai_2021, Yu2021,Rath2021,Vitale2021,Neven2021,Tran2015,Tran2016,Satzinger2021,Mi2020,Rath2021a} and  (out-of-time-ordered) correlation functions of Heisenberg operators \cite{Vermersch2019,Joshi2020}.
Randomized measurements on quantum states exploit statistical correlations obtained between measurements obtained from different random bases. 
However, for measuring an object like the SFF, we need to access the full trace of the time evolution operator $T(t)$, summing contributions from all its eigenstates. Therefore, we need to devise a  protocol that can measure how various initial states are propagated via $T(t)$, in a way that allows to extract the SFF from standard  projective measurements. This subsection provides this protocol and the estimation formulas to achieve this. We also comment on statistical errors arising from a finite number of experimental runs which are elaborated in detail in Sec.~\ref{sec:errors}.

\subsubsection{Description of the protocol}

Before describing the experimental sequence in detail, we first outline the key idea of our protocol: As visualized in Fig.~\ref{fig:figure2}, we consider a system $\mathcal S$ of $N$ qubits.  The first step of our protocol is to  prepare a random product state of these qubits. Next, this state is evolved with $T(t)$. Finally, a local measurement in the conjugate random product basis is performed, in order to probe how the time-evolved state compares to the initial random product state.
This is repeated for many random product states in order to sample the complete trace $\tr[]{T(t)}$ of the time evolution operator and its adjoint uniformly.
For instance, in the trivial case $T(t=0)=\mathbb{1}$, we obtain that the `time-evolved' state always matches to the initial random state corresponding to $D^{-1}\tr[]{T(0)}=1$. At later times $t$, we obtain in general a more complex statistics of measurement results from which we can extract the SFF and PSFF. 

In our protocol, we note that the ensemble average  over time evolution operators in the definition of SFF and PSFF can be favorably combined with the averaging over random product states and measurement bases. As detailed in the prescription of the protocol in the next paragraph,  each time evolution operator can thus in practice be applied only to a single random initial product state and measured only once in the corresponding randomized basis, i.e., only a single-shot measurement for each time evolution operator is sufficient in our protocol.

\begin{figure}[t]
\includegraphics[width=0.85\linewidth]{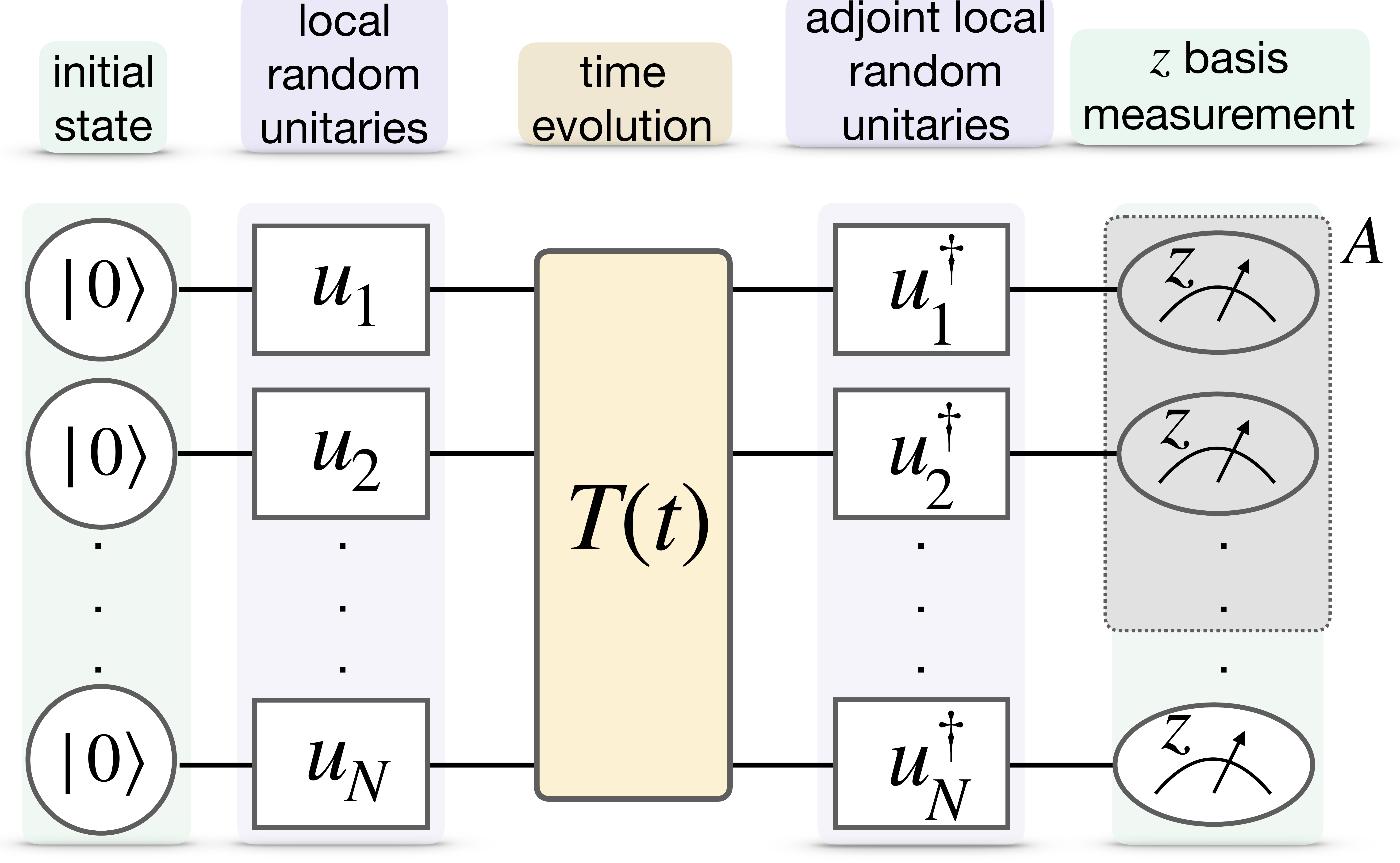}
  \caption{\textit{Probing SFF and PSFF using randomized measurements.} We present our protocol for the measurement of the SFF and PSFF using statistical correlations of local random unitaries  applied at different times in a single experiment. We begin with a product state $\rho_0=\ketbra{0}{0}^{\otimes N}$. Before and after the time evolution  $T(t)$, we apply random local rotations $U=\bigotimes_i u_i$ and $U^\dagger$, respectively, where local unitaries $u_i$ are sampled from  a unitary $2-$design. Here,   $T(t)$  can be generated as   Hamiltonian evolution, $T(t)={\rm{exp}}(-i Ht)$,   or Floquet dynamics, $T(t=n\tau)=V^n,~n\in\mathbb N$, where $V$ is Floquet evolution operator for time period $\tau$. In the last step, a single-shot  measurement is performed in the $z-$basis to collect a bitstring of the form $\mathbf{s}=(s_1, s_2, ..., s_{N})$ with $s_i\in\{0,1\}$. This procedure is repeated $M$ times and $M$ bitstrings are collected to estimate the SFF and PSFF using Eqs.\ \eqref{eq:sffmeas} and \eqref{eq:psffmeas}. The gray shaded region shows one possible choice of the subsystem $A$.}
    \label{fig:figure2}
    \end{figure}

In detail, the experimental recipe reads as follows:
 (i) We begin with a product state $\rho_0=\ketbra{\mathbf{0}}{\mathbf{0}}$ with $\ket{\mathbf{0}} \equiv  \ket{ 0}^{\otimes N}$. (ii) On this initial state, we apply local random unitaries $U=\bigotimes_{i=1}^N u_i$ where  $u_i$ are the local unitaries independently sampled from a unitary 2-design \cite{Dankert2009,Gross2007} on the local Hilbert space $\mathbb C^2$. Here, unitary 2-designs are ensembles of random unitaries whose first and second moments match the moments of the Haar measure on the unitary group (defining the CUE) \cite{Dankert2009,Gross2007}. Examples of unitary 2-designs on $\mathbb C^2$ include the (discrete) single-qubit Clifford group as well as uniformly distributed  unitary $2\times 2$ matrices which can be sampled for instance via the algorithm presented in Ref.~\cite{Mezzadri2006}. (iii) We evolve the system in time, i.e.\ apply  a time evolution operator $T(t)$, which is  generated   by a Hamiltonian $H$ (or Floquet operator $V$) with randomly sampled disorder potentials. (iv) We apply the adjoint local random unitary $U^\dagger$ resulting in  the final state $\rho_f(t)=U^\dagger T(t) U \rho_0 U^\dagger T^\dagger(t) U$. (v) Lastly, we perform a single-shot measurement in the computational basis  with outcome bitstring $\mathbf{s}=(s_1, \dots, s_N)$ with $s_i \in \{0,1\}$ for $i=1,\dots, N$. This concludes a single experimental run of our protocol. Steps (i)-(v) are now repeated $M$ times with new disorder realizations and new local random unitaries such that a set of outcome bitstrings $\mathbf{s}^{(r)}$ with $r=1, \dots M$ is collected. 

\subsubsection{Estimation formulas and illustrations}

The statistics of the measured  bitstrings $\mathbf{s}^{(r)}$, $r=1, \dots M$, depends  on the applied time evolution operators $T(t)$. Using the theory of unitary $2$- designs, we can express the SFF as a function of this data.  We define \be
\widehat{K(t)}= \frac{1}{M}  \sum_{r=1}^{M} \;  (-2)^{ - |\mathbf{s}^{(r)}|}~,
\label{eq:sffmeas}
\ee
where $|\mathbf{s}|\equiv \sum_i s_i$.
As we show in Sec.\ \ref{sec:proof_mt},  $\widehat{K(t)}$ yields an (unbiased)  estimate of the SFF for a finite number $M$ of experimental runs and converges  to $K(t)$  when $M\to \infty$. 

Remarkably, from the \textit{same} measurement data $\mathbf{s}^{(r)}$, we have also access to the PSFF $K_A(t)$ for arbitrary subsystems $A\subseteq \mathcal{S}$ via post-processing. To  this end, we simply project the measured bitrings on the subsystem $A$ of interest, i.e., define  $\mathbf{s}_A=(s_i)_{i\in A}$,
and use 
\be
\widehat{K_A(t)}= \frac{1}{M}  \sum_{r=1}^{M} \;  (-2)^{ - |\mathbf{s}_A^{(r)}|}~,
\label{eq:psffmeas}
\ee
which gives an (unbiased) estimate for $K_A(t)$ for finite $M$ and converges to $K_A(t)$ when $M\to \infty$ (see Sec.~\ref{sec:proof_mt}).

In Fig.~\ref{fig:figure1}(a-b), we  illustrate our measurement protocol in the  context of the periodically kicked spin-$1/2$ model $V_3$, Eq.\ \eqref{eq:Floquet-illustration-V3}. We consider a total system size of $N=6$ qubits and present the simulated experimental results (black dots and error bars) for $K(t)$ and $K_A(t)$ using $M=2 \times 10^5$ experimental runs for the single-shot sequence shown in Fig.~\ref{fig:figure2} at each time $t$. We observe that the simulated experiment agrees with the exact numerical calculations at all times $t$ within error bars. Here, error bars, indicating the standard error of the mean, quantify statistical errors arising from the finite measurement budget (i.e.\ the finite number $M$ of simulated single-shot measurements), see next subsection.

\subsubsection{Statistical errors and remarks}

The SFF and PSFF can  be accessed from the same set of measurement data via the estimators defined in Eqs.~\eqref{eq:sffmeas} and \eqref{eq:psffmeas}. Statistical errors arise in practice from a finite number $M$ of experimental runs, and are governed by the variance of these estimators. We discuss statistical errors in detail via numerical and analytical calculations in \ref{sec:errors}, and find a typical scaling of $M\sim 10^{N_A}\approx 2^{3.32 N_A}$ to access the (P)SFF of a (sub-)system of size $N_A$ up to a fixed relative error. Such exponential scaling of the measurement effort  reflects the exponential decrease of the SFF  with system size [see remarks below Eq.~\eqref{eq:SFF_cue}].  We emphasize however that this scaling of the experimental effort is  substantially better than  for  quantum process tomography which requires at least $\sim 2^{5N_A}$ experiments to reconstruct the full time evolution operator $T(t)$ \cite{Torlai2020}. Importantly, and in contrast to quantum process tomography, the initial state and the measurement basis coincide in our protocol. 

 As detailed in Sec.~\ref{sec:errors}, we can further decrease the required number of experimental runs to observe the ramp and plateau of the (P)SFF, by considering an averaged PSFF. Here, an average over PSFFs of all subsystems with a fixed size is performed. This results in a further improved signal-to-noise ratio.

Lastly, we remark that our protocol shares some similarities with randomized benchmarking \cite{Emerson2005,Emerson2007,Knill2008,Magesan2012,erhard2019characterizing}, where however \textit{global} random unitaries and their inverses are applied sequentially. In the case of randomized benchmarking the goal is to characterize noise and decoherence acting during the implementation of these global random unitaries. In contrast, with our protocol, the aim is to characterize a unitary time evolution operator $T(t)$  using \textit{local} random unitaries $U=\bigotimes_i u_i$ applied before and after  $T(t)$, which can be prepared with high fidelity~\cite{Brydges2019,Elben2020_xPlatform}.

\paragraph*{Organization of the paper:}In the remainder of the manuscript, we elaborate on the contents of the above synopsis with technical details,  derivations, and examples. In Sec.~\ref{sec:psffAnalytical},  we provide an in-depth theoretical analysis of the PSFF in RMT and in generic many-body models in relation to ETH. The analytic results are compared with numerics in Sec.~\ref{sec:psffNumerical} where we consider many-body models undergoing Floquet and Hamiltonian evolution. For the latter, we discuss both, chaotic and MBL phases. Sec.~\ref{sec:proof_mt} contains the necessary background and proof of our protocol to measure  the SFF. In Sec.~\ref{sec:errors}, we discuss statistical errors, arising in our measurement protocol from a finite number of experimental runs, and the influence of experimental imperfections. Lastly, we summarize in Sec.~\ref{sec:conclusion} with some concluding remarks and future directions.

\section{Partial Spectral Form Factor: Analytic Results}
\label{sec:psffAnalytical}
In this section, we analyze the origin of the main features observed in the PSFF, namely the ramp, plateau and shift, based on analytical calculations. We provide arguments to show that the PSFF generically  is  a reliable probe of eigenvalue correlations characterizing chaotic and localized phases, signified by the presence and absence of a late time ramp-plateau structure respectively. In addition, we show that the specific features observed in the PSFF are related to the ensemble and spectrum averaged second-moments of reduced density matrices of eigenstates at different energies, and therefore provide a useful measure of eigenstate properties. 

This section is organized as follows. In Sec.~\ref{sec:pSFF_RMT}, we analyze the PSFF in standard Wigner-Dyson random matrix ensembles (see App.~\ref{app:SFF-RMT} for a brief discussion), which are mathematically idealized models of quantum chaotic systems in which the PSFF can be obtained exactly. These ensembles display the essential features of the PSFF and present a clear example of the roles of eigenvalue and eigenstate statistics in these features. This is followed by a discussion of more general chaotic systems in Sec.~\ref{sec:pSFFchaotic}, where we show that the PSFF detects thermalization in the sense of ETH ~\cite{deutsch1991eth, srednicki1994eth, srednicki1999eth,rigol2008eth,DAlessio2016,deutsch2018eth,subETH} in addition to level statistics (see also Ref.~\cite{GarrattChalker}, that compares ETH for Floquet circuits to random matrix ensembles using the PSFF for specific subsystem sizes). We then discuss the PSFF in localized systems in Sec.~\ref{sec:pSFFlocalized}, and summarize our main conclusions for all cases in Sec.~\ref{sec:psff_summary}. 

Common to all these cases is the fact that the time-independent part of the PSFF in Eq.\ \eqref{eq:psff} is given by the plateau value, which depends only on the eigenstate purities (assuming no degeneracies) i.e. $K_{\mathrm A}(t\rightarrow \infty) =  {\mathcal{P}}_B/D_A$, where
\begin{equation}
    \mathcal{{P}}_B= \frac{1}{D} \overline{\sum_{i}\tr[B]{\rho_B^2(E_i)}}
    \label{eq:pb}
\end{equation}
is the (spectrum- and ensemble-)averaged purity of the reduced energy eigenstates. For later reference, we separate out this time-independent plateau value,
\begin{align}
K_{\mathrm A}(t)=&   
\frac{\mathcal{P}_B}{D_A} +\frac{1}{D \DA} \sum_{i\neq j}\overline { e^{i (E_i-E_j) t}\tr[B]{\rho_B(E_i)\rho_B(E_j)} } ,
\label{eq:psff-PQ}
\end{align}
and note that the time-dependent second term only involves overlaps of distinct energy levels.

\subsection{Random matrix ensembles}
\label{sec:pSFF_RMT}
To understand the essential features of the PSFF we first analyze it in RMT, allowing for an exact determination of the PSFF.
We choose Hamiltonians $H$ (Floquet operators $V$)   from the canonical Wigner-Dyson random matrix ensembles \cite{Wigner1955, Dyson1962, Mehta2004, HaakeBook}, yielding time evolution operators  $T(t)=\exp(-iHt)$ [$T(t=\tau n)=V^n$]. To evaluate the ensemble average in Eq.~\eqref{eq:psff-PQ}, we can utilize that for these RMT ensembles the  eigenvalues and eigenstates of $H$ ($V$) are uncorrelated. Thus, their ensemble average factorizes and can be performed independently. We find
\begin{align}
 K_A(t) = \frac{\mathcal{P}_B-\mathcal{Q}_B } {\DA}+ \DB \mathcal{Q}_B K(t)~,
\label{eq:psffrmt}
\end{align}
where $\mathcal{Q}_B=(D(D-1))^{-1}\overline{\sum_{i\neq j}\tr[B]{\rho_B(E_i)\rho_B(E_j)}}$ and $\mathcal{P}_B$ are the averaged overlap and purities of the reduced eigenstates, respectively. We note that here the PSFF is the  full SFF with a scaling factor $\DB \mathcal{Q}_B$ and a constant subsystem dependent shift $(\mathcal{P}_B-\mathcal{Q}_B)/\DA$ such that the entire time dependence of the PSFF is captured in the SFF. Therefore,  the PSFF in these models preserves the characteristic ramp-plateau structure and the relevant time scales of the SFF.

As shown in App.~\ref{app:pSFF-RMT}, we can evaluate $\mathcal{P}_B$ and $\mathcal{Q}_B$ explicitly using Wigner-Dyson RMT for the eigenstates of $H$ ($V)$. They are functions  of only the Hilbert space dimensions of subsystems $A$ and $B$, i.e. $\mathcal{P}_B\equiv\mathcal{P}_B (\DA, \DB)$ and $\mathcal{Q}_B\equiv\mathcal{Q}_B (\DA, \DB)$. The precise functional form of $\mathcal{P}_B$ and $\mathcal{Q}_B$  depends on the symmetry class of the Hamiltonian $H$ (Floquet operator $V$). For the case of the  unitary Wigner-Dyson ensembles, for example $H$ from the Gaussian unitary ensemble or $V$ from CUE, we find 
\begin{align}
    \mathcal{P}_B = \frac{\DA+\DB}{\DA\DB+1} \quad ; \quad
    \mathcal{Q}_B =\frac{\DB \left({\DA}^2-1\right) }{{\DA}^2 {\DB}^2-1}~.
    \label{eq:RMT-PandQ}
\end{align}
The analogous expressions for  orthogonal Wigner-Dyson ensembles  can be found in App.~\ref{app:pSFF-RMT}. In both symmetry classes at $D_A,D_B \gg 1$, we find that, $\mathcal{P}_B-\mathcal{Q}_B\approx 1/\DA$ and $\mathcal{Q}_B\approx 1/\DB$. Thus, in this limit, the PSFF has a constant shift of $1/D_A^2$ added to the SFF and the slope of the ramp  is the same as the slope of the ramp in the SFF, i.e. $K_A(t)\approx K(t)+1/\DA^2$ [see also Eq.\ \eqref{eq:pSFF_cue}].

\subsection{General chaotic systems} 

\label{sec:pSFFchaotic}
In the case of more general chaotic systems, we begin by separating out the reduced density matrices of the energy eigenstates into smooth and fluctuating functions of energy,
\begin{equation}
    \rho_B(E) = \frac{\mathbb{1}}{D_B}+\Delta\rho_B(E)+\delta\rho_B(E)~.
    \label{eq:rho_separation}
\end{equation}
Here, the first term is a constant corresponding to a maximally mixed reduced density matrix; $\Delta\rho_B(E)$ is traceless and a smooth function of $E$, while $\delta\rho_B(E)$ is again traceless but required to fluctuate rapidly with $E$. For our present purposes, it is useful to define the smooth and fluctuating parts in terms of their Fourier transforms with respect to a continuous energy variable as follows: for some cutoff time $t_\rho \ll O(D)$, we take their respective Fourier transforms to satisfy $(\Delta\tilde{\rho}_B(t))_{jk} = 0$ for $\lvert t\rvert > t_\rho$, and $ (\delta\tilde{\rho}_B(t))_{jk} = 0$ for $\lvert t\rvert < t_\rho$ (with some additional details in App.~\ref{app:pSFF-ETH}). The essence of the definition is that as a function of energy, the smooth part varies only over scales much larger than some energy window of size $t_\rho^{-1}$ containing several levels, while the fluctuating part varies only over scales much smaller than $t_\rho^{-1}$.

We will further assume that $\delta\rho_B(E)$ behaves as if it is `randomized' within these energy windows over the ensemble i.e.\ it is uncorrelated with the smooth part and satisfies $\overline{\tr[B]{\delta\rho_B(E_i)\delta\rho_B(E_j)}} = \delta_{ij}\overline{\tr[B]{\delta\rho_B^2(E_i)}}$ for $E_i,E_j$ closer than $\sim t_\rho^{-1}$, fluctuating around an average of zero (we do not require this behavior to persist over larger energy scales $\lvert E_i-E_j\rvert \gtrsim t_\rho^{-1})$. We note that this assumption is consistent with the general picture of random behavior over small energy windows in chaotic systems \cite{DAlessio2016}, and we can justify it more generally (irrespective of whether the system/ensemble is chaotic) as follows. In evaluating the SFF $K(t)$, the ensemble is usually chosen to have sufficiently large disorder so that the energy levels are randomly distributed over some large energy window, across different ensemble realizations. This is necessary to eliminate the erratic fluctuations of the SFF at large $t$ that depend on the precise positions of levels, and obtain a smooth ensemble-averaged behavior (see e.g.\ Refs.~\cite{prange1997sff,Bertini2018} for further discussion of this point). Our assumption is essentially that, this random redistribution of levels over different ensemble realizations extends to an energy window of $\sim t_\rho^{-1}$, effectively randomizing the fluctuations $\delta\rho_B(E)$ faster than this scale, while $\Delta\rho_B(E)$ which varies over scales larger than this energy window is not randomized in this manner. We also note that the eigenstates of a given ensemble realization themselves may additionally be random superpositions of those of a different realization, e.g.\ generally randomly mixing all eigenstates of the latter within the energy window in fully chaotic systems (i.e.\ systems with no `physical' conserved quantities other than energy)~\cite{deutsch1991eth,deutsch2010eigenstates,lu2019renyi, murthy2019eigenstates}, which gives further weight to this assumption.

\subsubsection{Shift-ramp-plateau structure of the PSFF}
\label{sec:pSFF_shiftrampplateau}

Using the form in Eq.~\eqref{eq:rho_separation}, the overlaps occurring in the definition of the PSFF in Eq.~\eqref{eq:psff} separate out into independent contributions from each part of the reduced density matrix - the cross terms vanish, due to tracelessness for terms involving overlaps with the maximally mixed part, or due to the randomization of $\delta\rho_B(E)$ for terms involving overlaps of the smooth and fluctuating part for $t \gg t_\rho$. We can write this as,
\begin{equation}
    K_A(t \gg t_\rho) = K(t) + \Delta K_A(t) + \delta K_A(t)~,
    \label{eq:KA_separation}
\end{equation}
where $\Delta K_A(t)$ involves only overlaps of the form $\tr[B]{\Delta\rho_B(E_i)\Delta\rho_B(E_j)}$ and similarly, $\delta K_A(t)$ involves only those of the form $\tr[B]{\delta\rho_B(E_i)\delta\rho_B(E_j)}$.
On decomposing $\delta K_A(t)$ in a manner analogous to Eq.~\eqref{eq:psff-PQ}, it follows that its time dependent part for $t\gg t_\rho$ (which sees contributions only from variations of the overlaps of fluctuating parts within energy windows smaller than $\sim t_\rho^{-1}$) vanishes on ensemble averaging, an important consequence of the randomization of $\delta\rho_B(E)$. This leaves only a constant contribution from the purity of the fluctuating part, $\delta K_A(t \gg t_\rho) = \delta\mathcal{P}_B/\DA$, where $\delta\mathcal{P}_B\equiv D^{-1}\overline{\sum_i \tr[B]{\delta\rho_B^2(E_i)}}$ (here we use `purity' to generally mean $\tr{x^2}$ for a Hermitian operator $x$). We see that this constant late-time shift is a generic feature of the PSFF, independent of the specific form of the full SFF $K(t)$. It merges into the plateau of the PSFF when $K(t)$ and $\Delta K_A(t)$ show only a plateau behavior - and therefore, the shift is an independent observable only if the other two terms show non-trivial time dependence at late times $t \gg t_\rho$.

We note that $\Delta K_A(t)$ is modulated only by a smooth function of two energy variables varying over scales larger than $t_\rho^{-1}$. For $t \gg t_\rho$, it should then essentially see the contribution to $K(t)$ from each part of the spectrum but modulated by the value of the function for nearly equal energies in that part. In App.~\ref{app:pSFF-ETH}, we show this by direct calculation for a fully chaotic system with Wigner-Dyson level statistics, obtaining a modulated linear ramp and plateau in addition to the late-time shift, for $t \gg t_{\rm Th},t_\rho$,
\begin{align}
    &K_A(t\gg t_{\rm Th},t_\rho) = \frac{\delta\mathcal{P}_B}{D_A}\qquad\qquad\qquad\qquad\qquad  \nonumber\\
    +&\frac{1}{D}\begin{cases}
(\beta\pi D)^{-1}\gamma t\left(1+D_B\widetilde{\Delta\mathcal{P}}_B\right) & \text{for } t \ll t_H~,\\ 1+D_B\Delta\mathcal{P}_B & \text{for } t \gg t_H~.
\end{cases}
\label{eq:pSFF_ETH}
\end{align}
Here, $\beta=1,2$ respectively for the orthogonal and unitary classes, while $\gamma = \sum_i \Omega^{-1}(E_i)$ is the range of energies in the spectrum with $\Omega(E)$ representing the (smoothened) local density of states, in agreement with known results for the full SFF (see e.g.~Refs.~\cite{Gharibyan2018,Liu2018}). To keep the expressions simple, we are ignoring corrections that are prominent near $t\sim t_H$ [see, for instance, the exact form of the GOE SFF in Eq.~\eqref{eq:sffgoe}]; we focus instead on the $t \ll t_H$ regime where the ramp appears linear for all values of $\beta$ and profiles of $\Omega(E)$, and the $t \gg t_H$ regime with a constant plateau. However, both expressions are exact throughout the range of times when $\beta = 2$ with constant density of states $\Omega(E) = t_H/(2\pi)$. We have also defined two ensemble-averaged quantities corresponding to slightly different spectrum averages of the purity of the smooth part, $\Delta \mathcal{{P}}_B = D^{-1}  \overline{\sum_{i}\tr[B]{\Delta\rho_B^2(E_i)}}$ and $\widetilde{\Delta\mathcal{{P}}_B} = \gamma^{-1} \overline{\sum_{i}\Omega^{-1}(E_i)\tr[B]{\Delta\rho_B^2(E_i)}}$, the latter including the contribution to the coefficient of the linear ramp from each part of the spectrum. We note that the purities of the smooth and fluctuating parts are (exactly) related to the overall average purity by $\mathcal{P}_B = D_B^{-1}+\Delta\mathcal{P}_B+\delta\mathcal{P}_B$, giving the expected plateau value of $\mathcal{P}_B/D_A$ in Eq.~\eqref{eq:pSFF_ETH}. There are also two competing time scales for the onset of the ramp, $t_{\rm Th}$ and $t_{\rho}$ - the former entirely determines the behavior of $K(t)$ but the latter appears in $\Delta K_A(t)$ and $\delta K_A(t)$.

For direct comparison with numerics, it is useful to define the ensemble averaged overlap of adjacent states, $Q_B = (D-1)^{-1}\sum_i \overline{\tr[B]{\rho_B(E_i)\rho_B(E_{i+1})}}$.
Using Eq.~\eqref{eq:rho_separation},  we note that,
\begin{align}
\mathcal{Q}_B  &= \frac{1}{\DB} + \Delta \mathcal P_B~, \nonumber \\
{\mathcal{P}_B} - {\mathcal{Q}_B}  &= \delta \mathcal P_B~,
\label{eq:PQ-numerics}
\end{align}
which follow from the assumption of uncorrelated $\delta\rho_B(E)$ in the ensemble, and taking $\Delta\rho_B(E_i) \simeq \Delta\rho_B(E_{i+1})$. We note that this definition of $Q_B$ is equivalent to that in Sec.~\ref{sec:pSFF_RMT} for random matrix ensembles, where the ensemble averaged overlaps between distinct states are independent of their energies. Sec.~\ref{sec:psffNumerical} will directly use $\mathcal{P}_B$ and $\mathcal{Q}_B$, with the implicit assumption that $\widetilde{\Delta P}_B$ is of similar order of magnitude to $\Delta\mathcal{P}_B$ (due to $\Omega(E)$ being of a similar order of magnitude throughout the spectrum) and is therefore similarly well represented by $\mathcal{Q}_B$.

\subsubsection{Constraints from eigenstate thermalization}
\label{sec:ethConstraints}

We have seen that at late times, the PSFF preserves the characteristic features of the SFF, such as the ramp and the Heisenberg time (as in Eq.~\eqref{eq:pSFF_ETH} for fully chaotic systems). However, there are non-negative subsystem-dependent parameters $\mathcal{P}_B$, $\delta\mathcal{P}_B$ and $\Delta\mathcal{P}_B$ ($\sim \widetilde{\Delta\mathcal{P}}_B$) that respectively influence the plateau value, the magnitude of the shift and the magnitude i.e., slope of the ramp. The purity $\mathcal{P}_B$ measures the extent of delocalization of eigenstates in a physical basis (e.g.\ a product basis of qubits), while we will see that $\delta \mathcal{P}_B$ and $\Delta\mathcal{P}_B$ are complementary probes of thermalization of these eigenstates. Specifically, we mean thermalization in the sense of ETH - that eigenstates corresponding to sufficiently close energies show nearly identical behavior in the dynamics of few-body observables~\cite{deutsch1991eth,srednicki1994eth,srednicki1999eth,rigol2008eth,DAlessio2016,deutsch2018eth}.

For our purposes, it is convenient to use subsystem ETH~\cite{subETH}, which amounts to imposing ETH on an entire subsystem i.e.\ for all observables in the subsystem, and is directly expressed in terms of reduced density matrices. It can be interpreted as the requirement of a small fluctuating part for the reduced density matrices of thermal eigenstates, as opposed to large fluctuations for non-thermal eigenstates. We can therefore apply it directly to the decomposition of reduced density matrices in Eq.~\eqref{eq:rho_separation}. An important advantage of this version of ETH is that the dependence on subsystem size is made more explicit, whereas more conventional statements of ETH restrict themselves to few body operators, corresponding to extremely small subsystems and therefore negligible subsystem dependence. This subsystem size dependence will turn out to be the primary non-trivial indicator of the properties of eigenstates in the PSFF. 

In App.~\ref{app:subETH}, we discuss the general constraints from (an extension of) subsystem ETH for eigenstates with an arbitrary extent of delocalization in a physical basis. Here, we present the results for a system with fully delocalized eigenstates, characterized by subsystem purities that follow the volume law of entanglement~\cite{Abanin2019},
\begin{equation}
\mathcal{P}_B = D_B^{-1}+O(D_B^{-1})+O(D_A^{-1}), \label{eq:purityconstraint2}    
\end{equation}
which cannot be less than $D_B^{-1}$ as well as $D_A^{-1}$. This is the case relevant for the numerical examples of Sec.~\ref{sec:psffNumerical}. If these eigenstates are thermal, subsystem ETH requires the smooth and fluctuating parts to satisfy,
\begin{equation}
    \Delta\mathcal{P}_B = O(D_B^{-1}),\  \delta\mathcal{P}_B = O(D_A^{-1}). \label{eq:thermalconstraints2}
\end{equation}
Non-thermal eigenstates are characterized by much larger fluctuations, $\delta\mathcal{P}_B \gg O(D_A^{-1})$, with $\Delta\mathcal{P}_B$ being correspondingly smaller so as to satisfy the constraint $\mathcal{P}_B = D_B^{-1}+\Delta\mathcal{P}_B+\delta\mathcal{P}_B$.
A narrower class of such chaotic systems (e.g.\ Floquet systems) have uniformly random eigenstates that are distributed in close agreement with the standard random matrix ensembles (Sec.~\ref{sec:pSFF_RMT}); the leading forms of the corresponding exact results in Eq.~\eqref{eq:RMT-PandQ} are seen to be consistent with Eqs.~\eqref{eq:purityconstraint2},\eqref{eq:thermalconstraints2}, on relating the two using Eq.~\eqref{eq:PQ-numerics}. In this context, we note that Ref.~\cite{GarrattChalker} has observed subleading corrections to the random matrix prediction for eigenstates in 1D Floquet quantum circuits.

\subsection{Localized systems}
\label{sec:pSFFlocalized}
Now, we consider localized systems, which show Poisson level statistics (i.e.\ uncorrelated neighboring levels) with localized non-thermal eigenstates, for strong disorder~\cite{nandkishore2015mbl, Abanin2019}. Here, $K(t)$ shows only a plateau at late times, allowing us to access only the purity $\mathcal{P}_B$ through the PSFF. Fully localized states are essentially nearly pure states with $\mathcal{P}_B\sim O(1)$ (more precisely, following an area law of entanglement~\cite{Abanin2019}), and additionally have large fluctuations $\delta\mathcal{P}_B \sim O(1) \leq 1-D_B^{-1}$. In other words, fully localized states cannot thermalize, as they would have to be distributed over different physical basis states due to orthogonality. An $O(1)$ plateau value is therefore all we need to characterize the eigenstates of such systems.

On the other hand, when the eigenstates become more delocalized in the approach to a chaotic phase, thermalization becomes a possibility. The moment any non-trivial correlations between nearby energy eigenvalues emerge in the spectrum, leading to a time dependence of $K(t)$ for $t>t_\rho$, $\delta\mathcal{P}_B$ becomes a meaningful observable in the PSFF according to the discussion following Eq.~\eqref{eq:KA_separation}. Here, the PSFF can be used to study the extent of thermalization in addition to the delocalization of the eigenstates.

\subsection{Summary}
\label{sec:psff_summary}

Let us summarize the main conclusions of this section from a unified perspective, before moving on to illustrate them with numerical examples in the next section. The PSFF in a subsystem $A$ combines energy level statistics, as reflected in the SFF, with the purities and overlaps of the reduced energy eigenstates in the complementary subsystem B. The plateau value of the PSFF encodes the (spectrum and ensemble averaged) purity, which is $\sim O(1)$ in a fully localized phase, and small for fully delocalized states in accordance with the volume law of entanglement, Eq.~\eqref{eq:purityconstraint2}. Something more interesting happens at late times if the SFF has a ramp or other time-dependent feature due to the existence of local level correlations. The PSFF inherits the ramp, but the ramp couples only to the smooth, slowly varying part of the reduced energy eigenstates. The rapidly fluctuating part is left over as a nearly time-independent shift [Eq.~\eqref{eq:pSFF_ETH}].

Eigenstate thermalization is primarily encoded in the size of the fluctuating part as measured by the shift - namely, an exponential suppression of the latter with subsystem size $N_A$ is indicative of thermalization [Eq.~\eqref{eq:thermalconstraints2}], while the lack of such a suppression translates to a failure of the eigenstates to thermalize. The smooth part is correspondingly large for thermal eigenstates and small for non-thermal eigenstates, so as to preserve the overall purity (i.e.\ extent of delocalization). Finally, there are special systems for which much more precise predictions for the PSFF can be theoretically derived/motivated and tested, such as chaotic Floquet systems with their random matrix-like eigenstates [Eqs.~\eqref{eq:psffrmt} and \eqref{eq:RMT-PandQ}].

Thus, the PSFF complements the SFF in analyzing late-time quantum chaos by being able to probe if the eigenstates satisfy ETH, in addition to (and because of) capturing information about level correlations as contained in the ramp of the SFF. In particular, we expect that it could potentially be useful in studying the joint emergence or loss of Wigner-Dyson level statistics and eigenstate thermalization (which are formally independent notions of late time quantum chaos) and their interdependence, across a transition or crossover between a chaotic and non-chaotic phase. This could be done by tuning the parameters of a system (say, in a quantum simulator) between such phases, and measuring PSFFs across different choices of subsystems of different sizes - analyzing the extent of delocalization of eigenstates in the absence of a ramp via the plateau value, and additionally the extent of thermalization through the value of the shift if a ramp or other time-dependent feature is present at late times. Among the interesting possibilities that have been considered for such an intermediate regime, which could conceivably be probed with the PSFF, is the existence of so-called non-ergodic extended states~\cite{deluca2014nee,kravtsov2015nee,facoetti2016nee,altshuler2016nee,kravtsov2018nee,altland2019nee} where the eigenstates are incompletely delocalized but do not thermalize, or alternatives in which the eigenstates thermalize without being fully delocalized~\cite{ThermalNEEs}.

\section{Partial Spectral Form Factor: Numerical Results}
\label{sec:psffNumerical}
Having discussed features of the PSFF and its connection to the SFF utilizing Wigner-Dyson random matrix ensembles and the ETH, we now present our numerical results of PSFFs in locally interacting many-body models, as realized in quantum simulators.
For this purpose, we focus on two examples: the Floquet model Eq.\ \eqref{eq:Floquet-illustration-V3} and the Hamiltonian model Eq.\ \eqref{eq:Hamiltonian}. Our results are in agreement with the analysis of the previous Sec.~\ref{sec:psffAnalytical}, in particular regarding the orders predicted for the averaged purity $\mathcal{P}_B$ and the overlap $\mathcal{Q}_B$ via Eq.\ \eqref{eq:PQ-numerics}. We consider the Floquet model in the chaotic phase and the Hamiltonian model in both the chaotic and MBL phases. 

\subsubsection{Example 1: Floquet system}
\begin{figure}[!ht]
    \centering
\includegraphics[width=\linewidth]{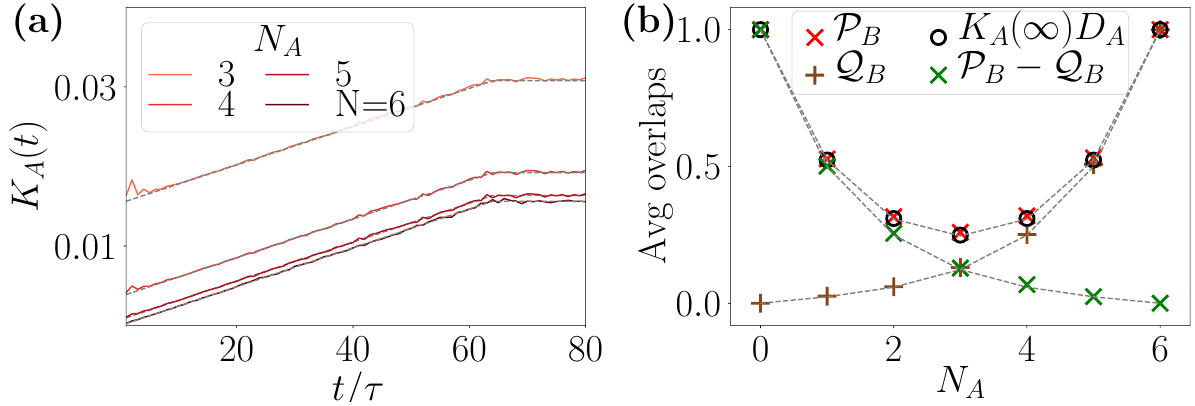}
\caption{\textit{Results for the Floquet $V_3$ model.} (a) The SFF and PSFF are presented for $N=6$, $N_A=3,~4,~5$ in red colors. In gray, we plot the same quantities in a CUE model. (b) The plateau value $K(\infty)$ multiplied with the subsystem dimension $D_A$ is plotted in black circles and matches with the averaged purity $\mathcal{P}_B$ plotted with red crosses. The average overlap $\mathcal{Q}_B$ and the difference $\mathcal{P}_B-\mathcal{Q}_B$ are presented in brown and green respectively. We observe a perfect match with the respective quantities in CUE plotted in gray, indicating the same averaged eigenvalue and eigenstate statistics in CUE and $V_3$. In the numerical computation, we have taken 8000 disorder realizations to perform ensemble averaging and the subsystems $A$ are chosen from the middle of the spin chain.}
\label{fig:ManyBodyV3}
\end{figure}

The Floquet time evolution operator $V_3$ has the same quasi-energy eigenvalue statistics as the CUE random matrix ensemble \cite{Regnault2016, Vasilyev2020}. As mentioned in Sec.~\ref{sec:overview} the Floquet models are known to thermalize to infinite temperatures as per RMT and thus we expect the eigenstate statistics to also be the same as in the corresponding RMT class. To show this,   we present in Fig.~\ref{fig:ManyBodyV3}(a) numerically obtained SFF and PSFF for a total system size of $N=6$ and subsystem sizes $N_A=3, 4$ and $5$ for the model $V_3$. We plot with gray lines the corresponding $K_A(t)$ in a CUE model where the analytic forms can be exactly calculated (see Sec.~\ref{sec:pSFF_RMT} and App.~\ref{app:pSFF-RMT}). For the PSFF $K_{A}(t)$ at $N_A=3$ and very early times, we notice that the onset of the ramp takes a few initial periods to set, but eventually the PSFF follows the CUE prediction. 

The closeness between the statistics of CUE and $V_3$ can further be seen from the average overlaps of reduced densities of eigenstates $\mathcal{P}_B$ and $\mathcal{Q}_B$.  In Fig.~\ref{fig:ManyBodyV3}(b) we present the average purity and  overlaps as functions of subsystem size $N_A$. At plateau time, $t> t_H(=D\tau)$ the PSFF becomes $K_A(t\rightarrow\infty)= \mathcal{P}_B/ \DA$, see Eq.~\eqref{eq:psff-PQ}. We plot numerically obtained $K_A(\infty)\DA$ in black circles, and the average purity $\mathcal{P}_B$ with red crosses, they confirm the analytic expectation. The average overlap $\mathcal{Q}_B$ and the difference $\mathcal{P}_B-\mathcal{Q}_B$ are plotted in brown and green circles respectively and match with the CUE data. 

To conclude, the SFF, PSFF, averaged purity and overlaps match in the CUE and $V_3$ model and thus we expect the form of the PSFF in Eq.\ \eqref{eq:pSFF_cue} to hold for the model $V_3$, after a small initial time period. We know from Eq.\ \eqref{eq:RMT-PandQ}, for large Hilbert space dimensions, that $\mathcal{Q}_B\approx 1/\DB$ and $\mathcal{P}_B-\mathcal{Q}_B\approx 1/\DA$. Therefore utilizing, Eq.\ \eqref{eq:PQ-numerics}, we find that  $\Delta\mathcal{P}_B=0$ and $\delta \mathcal{P}_B=O(1/\DA)$ for $V_3$ and the RMT models. The purity of the smooth part (of the form of $ \mathrm{Tr}[\Delta\rho_B^2(E)]$) appears  in the ramp part of the PSFF in Eq.\ \eqref{eq:pSFF_ETH} and thus we note that the ramp coefficient is $\sim 1/D^2$ for $D_A \gg 1$. On the other hand, the purity of the fluctuating part (of the form of $ \mathrm{Tr}[\delta\rho_B^2(E)]$) comes in the time-independent term added to the SFF in Eq.\ \eqref{eq:pSFF_ETH}, which is to the leading orders $~1/D_A^2$, as also in the CUE model [Eq.\ \eqref{eq:pSFF_cue}]. To further have another numerical example of the Floquet model thermalizing according to RMT, we present the example of a chaotic Floquet model with time-reversal symmetry in App.~\ref{app:coe-u2}.

\subsubsection{Example 2: Hamiltonian system}
As our second example, we consider a transverse field Ising model in presence of longitudinal local disorders,
\be
H=J\left(\sum_{\substack{i,j=1 \\ i<j}}^N \frac{1}{(i-j)^\alpha}\sigma^z_i \sigma^z_j+\sum_{i=1}^N \sigma^x_i\right) + W \sum_{i=1}^N h_i \sigma^z_i,
\label{eq:Hamiltonian}
\ee
where $h_i$ are drawn uniformly at random from $(-1,1)$. The coefficient $J$ and the exponent $\alpha$ denote the strength and  range of the interactions respectively. The disorder strength $W$ is known to specify the nature of the dynamics; $W\sim J$ depicts chaotic regime and $W \gg J$ corresponds to the localized regime (for a similar model see, \cite{Gharibyan2018}).  In the App.~\ref{app:IsingGapRatio}, we present the adjacent level gap ratio as a function of $W/J$ and $\alpha$ and find that the chaotic and localized phases exist for short ($\alpha>1$) as well as for long ($\alpha<1$) range interactions. In this work, we choose $\alpha=1.2$, and as examples of the chaotic and localized phases, we take  $W=J$  and $W=10J$ respectively. In contrast to the presence of the ramp and plateau in the SFF for chaotic models, the SFF for localized models stays flat for all times $t\gg 0$. In the numerics, we will find that the PSFF preserves this flat feature of the SFF, and has a  subsystem dependent shift added over the SFF, as predicted in Sec.~\ref{sec:pSFFlocalized}. In Fig.~\ref{fig:HamilSFFpSFF} and \ref{fig:manybodyHamil} we present numerical results for the Hamiltonian model \eqref{eq:Hamiltonian} in these two phases.  For clarity, we have used red color for the chaotic phase ($W=J$) and blue for the MBL phase ($W=10J$). We note that the Hamiltonian of Eq.~\eqref{eq:Hamiltonian} has the time-reversal symmetry of complex conjugation in the computational ($\sigma^z_i$) basis \cite{HaakeBook, avishai2002_GOE, brown2008_GOE}. A chaotic Hamiltonian with this symmetry is known to follow the eigenvalue statistics (or the SFF) of GOE after the Thouless time $t>t_{\rm Th}$ \cite{Mehta2004, DAlessio2016,HaakeBook, avishai2002_GOE, brown2008_GOE}, thus we have also put the results for GOE class in gray in Fig.~\ref{fig:HamilSFFpSFF}. 
\begin{figure}[!ht]
    \centering
\includegraphics[width=0.9\linewidth]{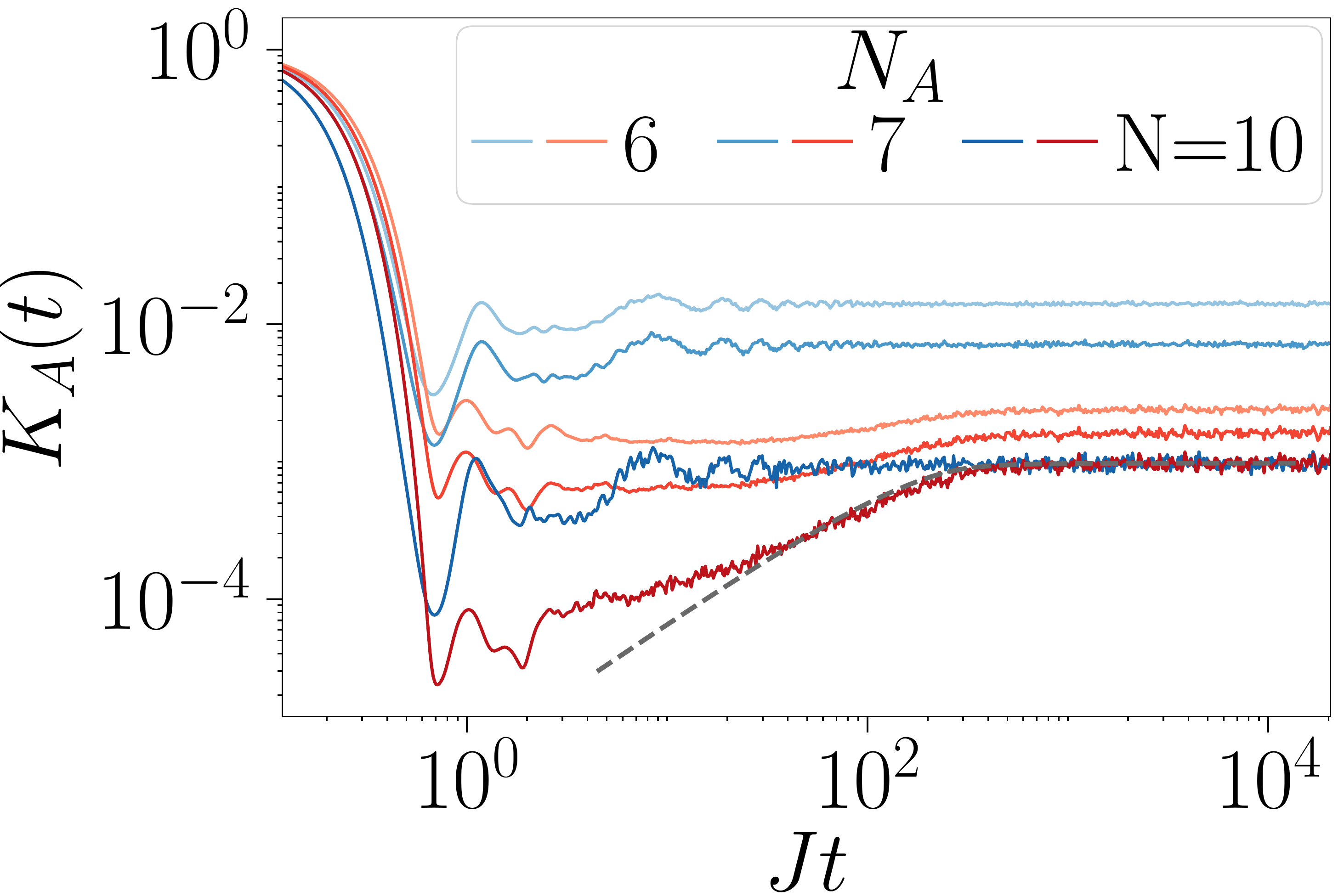}
\caption{\textit{Results for the Hamiltonian model.} In a log-log plot we present the chaotic phase ($W=J$) in red, MBL phase $(W=10J$) in blue, and the GOE in gray. In both phases the  SFF and PSFF are plotted for (sub-)system sizes  $N_A=6,~7$ and $N_A=N=10$.  The SFF for the chaotic phase has the characteristic ramp and plateau and  follows the GOE SFF at late times. The PSFF in this phase also has the shift, ramp and plateau, we plot these in a focused linear scale plot in Fig.~\ref{fig:manybodyHamil}(a). The MBL phase shows a flat SFF and PSFF for all times $t\gg 0$. The mean level spacing (i.e.\  the Heisenberg time) in the MBL phase and GOE are numerically rescaled to match to the one in the chaotic phase.}
\label{fig:HamilSFFpSFF}
\end{figure}
\begin{figure}[!ht]
    \centering
\includegraphics[width=\linewidth]{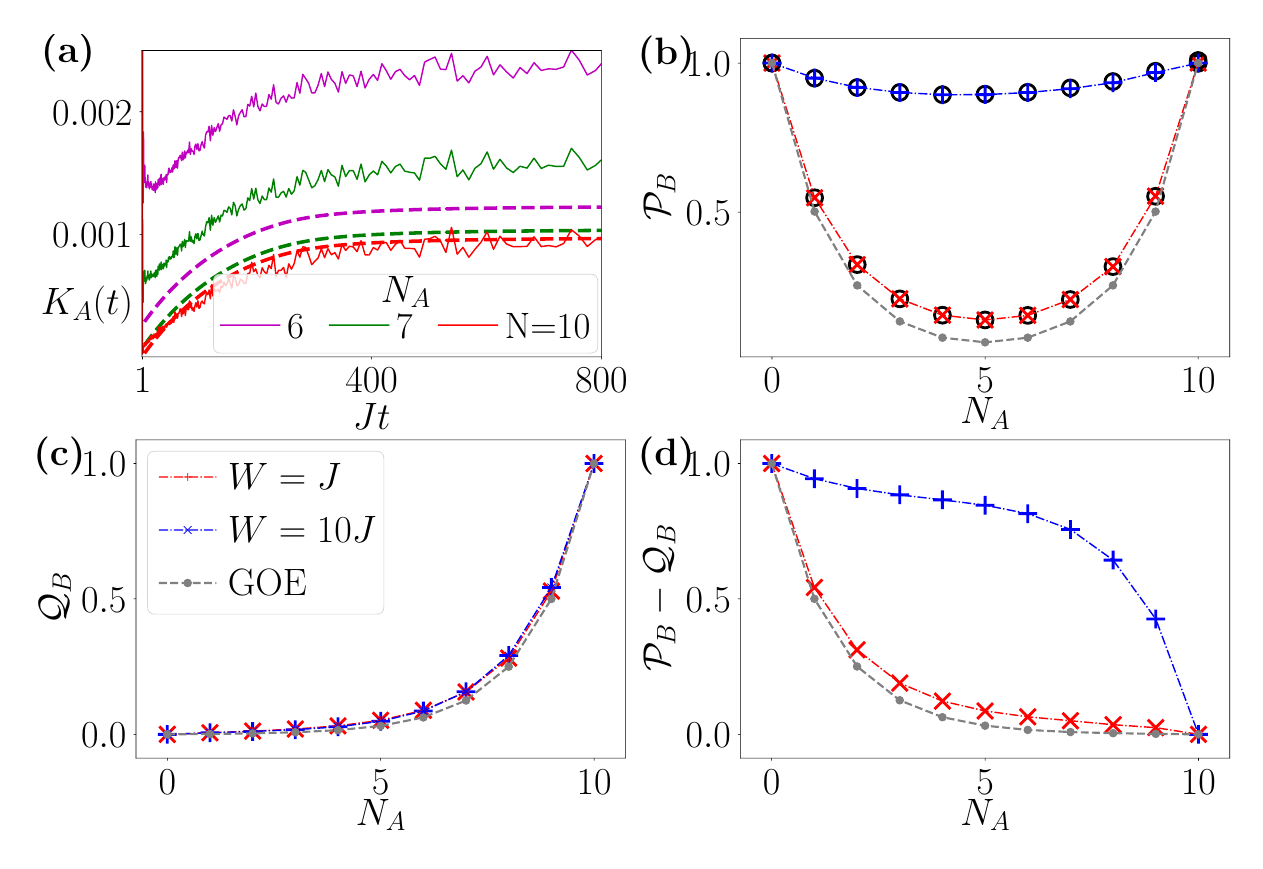}
\caption{\textit{Results for the Hamiltonian model.} (a) In linear scale we present the SFF and PSFF for the chaotic phase ($W=J$). (Sub-)system sizes $N_A=6$, $7$ and $N_A=N=10$ are plotted with magenta, green and red respectively for both the Hamiltonian model (with solid curves) and the GOE (with dashes). We observe  differences in the PSFF for chaotic Hamiltonian and GOE. These differences are investigated in (b), (c) and (d)  through $\mathcal{P}_B$ and $\mathcal{Q}_B$. We use red color for the  chaotic phase ($W=J$) and gray for the GOE. For comparison we have also plotted these quantities in the localized phase ($W=10J$) using blue color. (b) We plot  $K_A(\infty) D_A$ using  black circles which matches with the corresponding average purity $\mathcal{P}_B$ of the MBL and chaotic phase. (c) The average overlap $\mathcal{Q}_B$ for MBL, chaotic and GOE  follow closely the behavior $1/\DB$. (d) The difference $\mathcal{P}_B-\mathcal{Q}_B\approx\delta{\mathcal{P}_B}$, which encodes the shift of the PSFF, is larger for large disorders (MBL) compared to small disorders (chaotic). In the numerical computation, we have taken 200 Hamiltonians to perform ensemble averaging and the subsystems $A$ are chosen from the middle of the spin chain.
}
    \label{fig:manybodyHamil}
\end{figure}

 As a side remark, we emphasize at this point that the spectrum of the local Hamiltonian model, Eq.\ \eqref{eq:Hamiltonian}, does not have the same density of states as the GOE spectrum and thus the Hamiltonian SFF should be compared with an average of GOE SFFs, each with $t_H$ determined by different parts of the Hamiltonian spectrum. Often, this is circumvented by removing the non-universal effects arising from the edges of the local Hamiltonian spectrum by using a filter function such that only the middle part of the spectrum contributes \cite{Suntajs2019} or considering very large system sizes where the edge effects are effectively smaller. In our work, we focus on the measurement of chaotic features through the observation of the ramp, plateau and the shift which can already be observed without filtering for moderate system sizes,  which we focus on. 

In Fig.~\ref{fig:HamilSFFpSFF}, the SFF and PSFF are presented for the system size $N=10$ and subsystem sizes $N_A=6$ and $7$. In order to have the same Heisenberg time $t_H$, the eigenvalues are numerically rescaled such that the average mean level spacing for $W=10J$ match with the one for $W=J$. As a guide, we have plotted in gray the GOE SFF where the $t_H$ is determined from the full width of the chaotic Hamiltonian spectrum and observe that the SFF for the chaotic phase follows the  GOE SFF closely. The PSFF for the chaotic phase, shifted up compared to the SFF, also shows the ramp and plateau behavior which are seen better in a linear plot in Fig.~\ref{fig:manybodyHamil}(a). Here, focused to display chaotic features, we have used solid lines for the chaotic  Hamiltonian and dashes for the GOE. The different subsystem sizes are shown in different colors. We note that the PSFF for the chaotic local model and GOE are different (see the magenta and green curves). These differences arise due to the differences in eigenstate properties of the local Hamiltonian and GOE.

Further, to concretely discuss second-moments of eigenstates, in Fig.\  \ref{fig:manybodyHamil}(b)  we present the averaged purity  $\mathcal{P}_B$ using crossed markers. We have also plotted here the plateau values $K_A(\infty)\DA$ (in black circles) for both chaotic and MBL phases which agree with their respective purities following $K_A(t\rightarrow \infty)=\mathcal{P}_B/\DA$ [see Eq.~\eqref{eq:psff-PQ}]. Note that these average purities are consistent with a volume law of entanglement in the chaotic phase, and an area law in the localized phase \cite{Abanin2019}. For the remainder of this section, it is useful to discuss the two phases $W=J$ and $W=10J$ separately. 

For the chaotic phase $W=J$, the average overlaps $\mathcal{Q}_B$ and $\mathcal{P}_B-\mathcal{Q}_B$ are presented in red in the bottom panel of Fig.~\ref{fig:manybodyHamil} as functions of $N_A$. Assuming ETH for the chaotic systems, we have discussed orders of magnitude of these overlaps in Sec.~\ref{sec:pSFFchaotic}.  Utilizing Eq.\ \eqref{eq:PQ-numerics} we can  comment  on the orders of $\Delta \mathcal{P}_B$ and $\delta \mathcal{P}_B$ (see App.~\ref{app:IsingOrders} for more details on the numerical extraction of these orders).  From  $\mathcal {Q}_B$ [Fig.~\ref{fig:manybodyHamil}(c)], we find $\Delta \mathcal {P}_B=O(1/\DB)$ and from $\mathcal{P}_B-\mathcal{Q}_B=\delta\mathcal{P}_B$ [Fig.~\ref{fig:manybodyHamil}(d)], we find $\delta \mathcal{P}_B \sim O(1/\DA)$, confirming the ETH predictions for chaotic systems. We verify that the value of the shift of the PSFF in the linear ramp region is given in terms of the purity of the fluctuating part i.e., by $\delta \mathcal {P}_B/D_A$ in App.~\ref{app:IsingShift}.  For comparison,
we have plotted the same quantities in a GOE model in gray. We note a difference between the overlaps (properties  of the eigenstates) in the local chaotic Hamiltonian and GOE, which is not surprising because the statistics of eigenstates need not be the same in the two models.

Next, we look at the orders of magnitude of the overlaps in the phase $W=10J$, plotted in blue in the bottom panel of Fig.~\ref{fig:manybodyHamil}. Following Eq.\ \eqref{eq:PQ-numerics} from the $\mathcal {Q}_B$ [Fig.~\ref{fig:manybodyHamil}(c)], we find $\Delta \mathcal {P}_B=O(1/\DB)$ and from $\mathcal{P}_B-\mathcal{Q}_B=\delta\mathcal{P}_B$ [Fig.~\ref{fig:manybodyHamil}(d)], we find $\delta \mathcal{P}_B \sim O(1) \gg O(1/\DA)$.
The localized phase is not expected to satisfy ETH, and as discussed in the Sec.~\ref{sec:pSFFlocalized}, we expect such large shift in the PSFF in MBL systems.  Due to larger $\delta\mathcal{P}_B$ in the MBL phase, we notice a larger overall shift of the PSFF in the MBL phase, shown in blue in Fig.~\ref{fig:manybodyHamil}(b)-(d). 

\section{Proof of the protocol}
\label{sec:proof_mt}
In Sec.~\ref{sec:protocol}, we presented our measurement protocol and  defined estimators for the SFF and PSFF [Eqs.~\eqref{eq:sffmeas} and \eqref{eq:psffmeas}] in terms of the measured bitstrings. In this section, we prove analytically that these are unbiased estimators of the SFF and PSFF utilizing the theory of unitary $2$-designs.

\subsection{Useful results from unitary $2$-designs}

Unitary $n-$designs are ensembles of random unitary matrices, whose averages of polynomial moments of order up to $n$ coincide with ones of the Haar measure (or equivalently the CUE)~\cite{Dankert2009}.
With the help of Weingarten calculus, these moments can be expressed analytically~\cite{Collins2006}, allowing us to relate the statistics of randomized measurements to the quantity that we would like to measure.
Since the measured bitstrings from the protocol are sampled from the Born probabilities $|\bra{\mathbf{s}}U^\dag T(t) U \ket{\mathbf{0}}|^2$ which are polynomial functions of order two in $U$, we restrict ourselves to Weingarten calculus of order two. Using independent local unitaries $U=\bigotimes_i u_i$, one  finds for any operator $C$ defined on the `two-copy' Hilbert space $\mathcal{H}^{\otimes 2}$~\cite{Elben2019}
\begin{equation}
    \mathbb E_{U}\left[ 
  (U\otimes U ) \,  C \, (U^\dag \otimes U^\dag ) 
\right]
= \sum_{\sigma,\tau } w_{\sigma,\tau} \tr[]{\sigma  C}  \tau . \label{eq:twirling}
\end{equation}
Here, $\mathbb E_{U}$ denotes the average over local unitaries of the form $U=\bigotimes_i u_i$ with $u_i$ sampled for each $i$ independently from a unitary $2$-design on the local Hilbert space $\mathbb{C}^{\otimes 2}$. Further, the sum extends to all two-copy permutation operators $\sigma =\bigotimes_i \sigma_i$ and $\tau =\bigotimes_i \tau_i$ with $\sigma_i,\tau_i=\mathbb{1}_i, \mathbb{S}_i $. Here, the identity $\mathbb{1}_i$ and the \textit{swap operator $\mathbb{S}_i$}  act as  $\mathbb{1}_i \ket{s_i}\otimes \ket{s'_i}=\ket{s_i}\otimes \ket{s'_i}$ and $\mathbb{S}_i \ket{s_i}\otimes \ket{s'_i}=\ket{s'_i}\otimes \ket{s_i}$ on local basis states  $\ket{s_i}$ and $\ket{s'_i}$.
Finally, the coefficient $w_{\sigma,\tau}=\prod_i \textup{Wg}^{U(2)}(\sigma_i\tau_i^{-1})$ is determined by the Weingarten function $\textup{Wg}^{U(2)}$, with
$\textup{Wg}^{U(2)}(\mathbb{1}_i)=1/3$ and $\textup{Wg}^{U(2)}(\mathbb{S}_i)=-1/6$.
The expression above, which is valid for any  operator $C$, is the mathematical backbone of randomized measurements. In randomized measurement protocols, the goal is then to identify an operator $C$, whose expectation value can be inferred from the experimental data, such that the right hand side of the above equation reveals the quantity of interest.

In order to reconstruct the SFF, it will turn out to be particularly useful  to choose  $C=O\otimes \rho_0$ with $\rho_0= \ketbra{\mathbf{0}}{\mathbf{0}}$ and 
\begin{align}
    O&=(\ketbra{0}{0}-\frac{1}{2}\ketbra{1}{1})^{\otimes N}
= \sum_{\mathbf{s}}(-2)^{|\mathbf{s}|} \ket{\mathbf{s}}\bra{\mathbf{s}}
\label{eq:oswap}
\end{align}
where the sum extends to all bitstrings $\mathbf{s}=(s_1, \dots, s_N)$ with $s_i\in \{0,1\}$, and $|\mathbf{s}|\equiv \sum_i s_i$.
For this choice, we obtain
\begin{equation}
    \mathbb E_{U}\left[ 
  (U \otimes U )\, (O\otimes \rho_0) \, (U^\dag\otimes U^\dag) 
\right]
= 4^{-N} \mathbb{S} \label{eq:twirlingswap}
\end{equation}
with $\mathbb{S}=\bigotimes_i \mathbb{S}_i=\sum_{\mathbf{s},\mathbf{s'}}\ket{\mathbf{s'}}\bra{\mathbf{s}}\otimes \ket{\mathbf{s}} \bra{\mathbf{s'}}$. The Swap operation $\mathbb{S}$ is the key operation to extract non-trivial quantities, such as the purity, in randomized measurements~\cite{Elben2019}. Here, to access the SFF, it is convenient to take the partial transpose operation $A\otimes B\to A^T \otimes B$ in the above equation, leading to 
\begin{equation}
    \mathbb E_{U}\left[ 
  (U^* \otimes U ) (O^T \otimes \rho_0) (U^T\otimes U^\dag) 
\right]
= 2^{-N} \ket{\Phi_N^+}\bra{\Phi_N^+} \label{eq:twirlingBell},
\end{equation}
where $\ket{\Phi_N^+}=\bigotimes _i \ket{\Phi_i^+}=2^{-N/2}\sum_s \ket{\mathbf{s}}\otimes \ket{\mathbf{s}}$ is a product of Bell pairs $\ket{\Phi_i^+}=2^{-1/2}(\ket{0}\otimes \ket{0}+\ket{1}\otimes \ket{1})$.

\subsection{Rewriting the SFF in a form suitable for randomized measurements}

 For clarity, we focus on the measurement of the full SFF $K(t)$, and present the  case of the  PSFF in App.~\ref{app:MeasProt}.  We first define for a fixed time-evolution operator $T(t)$
 \begin{align}
     K_{T(t)} \equiv  4^{-N} \tr{T(t)}\tr{T^\dagger(t)}
 \end{align}
 such that the ensemble (disorder) average $K(t)=\overline{K_{T(t)}}$ yields the SFF, according to the definition Eq.~\eqref{eq:sff}.
 Secondly, we show that $K_{T(t)}$ equals the survival probability of the Bell State $\ket{\Phi_N^+}$ under the dynamics generated by $\mathbb{1}\otimes T(t)$, i.e.\ 
  \begin{align}
K_{T(t)}&=\braket{\Phi_N^+|\mathbb{1} \otimes {T}(t)|\Phi_N^+}\braket{\Phi_N^+|\mathbb{1} \otimes {T^\dagger }(t)|\Phi_N^+}.
\label{eq:survprob}
\end{align}
To this end, we use the following identity for any two operators $A,B$ on $\mathcal{H}$
 \begin{align}
     \tr{A B} = 2^N\braket{\Phi_N^+| A^T \otimes B | \Phi_N^+},
     \label{eq:bellident}
 \end{align}
 which can be proven by inserting the definition of the Bell state  $\ket{\Phi_N^+}=2^{-N/2}\sum_s \ket{\mathbf{s}}\otimes \ket{\mathbf{s}}$ in terms of computational basis states. Eq.~\eqref{eq:survprob} follows directly by choosing $A=\mathbb{1}$ and $B=T(t)$. We note that the identity Eq.~\eqref{eq:survprob} has been discussed in the context of holographic duality \cite{Sonner-survival}. 
In this case generalized finite temperature form factors can be written in terms of thermofield double-states, which take the form of Bell states in the limit of  infinite temperature.
 With the help of  Eq.~\eqref{eq:twirlingBell},  we can now replace one Bell state projector in Eq.~\eqref{eq:survprob} with $O\otimes \rho_0$ averaged over random unitaries $U$. We find   \mbox{$K_{T(t)}= \mathbb E_{U}
    \left[ K_{T(t),U} \right]$ with  $K_{T(t),U}$} defined as
\begin{align}
 K_{T(t),U}&\equiv  2^N  \bra{\Phi_N^+} U^* O^T U^T \otimes {T}(t) U \rho_0 U^\dag{T^\dagger}(t)
   \ket{\Phi_N^+}. \end{align}
   Using once more the identity~\eqref{eq:bellident}, it follows that $ K_{T(t),U}$
   equals the expectation values of the operator $O$ in the final state $\rho_f(t)$
   \begin{align}
         K_{T(t),U}&=\tr{ O  \, \right. \underbrace{ U^\dag T(t) U \rho_0 U^\dag T^\dag(t) U}_{\rho_f(t)}\left. } \nonumber \\
        &=\sum_{\mathbf{s}} (-2)^{|\mathbf{s}|} |\bra{\mathbf{s}} U^\dag T(t) U \ket{\mathbf{0}}|^2.
   \end{align}
   Here, $|\bra{\mathbf{s}} U^\dag T(t) U \ket{\mathbf{0}}|^2$ is precisely the Born probability of finding a bitstring $\mathbf{s}$, in the computational basis measurement performed at the end of our measurement sequence when the state $\rho_f(t)$ has been prepared [c.f.~Sec.~\ref{sec:protocol}]. 
   It follows thus that 
    \begin{eqnarray}
     K_{T(t),U}
    &=& \mathbb E_{QM}\left[ (-2)^{|\mathbf{s}|}\right],
 \end{eqnarray}
  where $\mathbb E_{QM}$ is the quantum mechanical average and ${\mathbf{s}}$ denotes the outcome of the  computational basis measurement at the end of the measurement sequence.

  In summary, it follows that for each measured bitstring ${\mathbf{s}}$, $(-2)^{|{\mathbf{s}}|}$ provides an estimation of the SFF, which in expectation over ensemble (disorder) average, over random unitaries and quantum mechanical averaging, yields the SFF
  \begin{align}
       K(t) = \overline{\mathbb E_{U} \mathbb E_{QM} \left[ (-2)^{|{\mathbf{s}}|}\right]}.
       \label{eq:estimator_sff}
   \end{align}
    In practice, we repeat our measurement protocol by performing $M$ independent experimental runs (with independently sampled time evolution operators and random unitaries), and calculate the  empirical average  $\widehat{K(t)}$ [Eq.~\eqref{eq:sffmeas}]. Using Eq.~\eqref{eq:estimator_sff}, it follows that $\widehat{K(t)}$ converges to $K(t)$ in the limit $M\to \infty$. For finite $M$,  statistical errors are governed by the variance of $\widehat{K(t)}$, and are discussed in the next section.
    In the App.~\ref{app:MeasProt}, we extend our derivation to the case of the PSFF, and illustrate the mapping between randomized measurements and the (P)SFF graphically.

\section{Statistical errors and imperfections}
\label{sec:errors}

We have discussed characteristic features of the SFF and PSFF, such as shift, ramp and plateau. The crucial question arises whether these can be measured in today's quantum simulators, utilizing our protocol (Sec.~\ref{sec:protocol}) with a finite measurement budget (number of experimental runs $M$) and in the presence of unavoidable experimental imperfections. In the following, we  first analyze in detail  statistical errors which arise from a finite number of experimental runs $M$. These determine the signal-to-noise ratio for a measurement of the shift of the PSFF (extracted from measurements at a single point in time) and the slope of the SFF and PSFF (extracted from differences of measurements at various points in time).
Subsequently, we discuss the influence of experimental imperfections, such as imperfect implementation of our measurement protocol or decoherence during the time evolution.

\subsection{Statistical errors}
We discuss statistical errors arising from a finite number of experimental runs $M$. We first consider the estimation of the SFF and PSFF at single point in time, and secondly the estimation of (the slope of) the ramp from measurements of the SFF and PSFF at different times.

\subsubsection{Observing PSFF and SFF}

We can bound the statistical errors of the estimator $\widehat{K_A(t)}$ [Eq.~\eqref{eq:psffmeas}]  by its variance. As shown in App.~\ref{app:staterrors}, we find that, 
\begin{align}
\text{Var}[\widehat{K_A}] &=\frac{1}{M} \left( 2^{-N_A} \sum_{B\subseteq A} K_B - K_A^2\right) \equiv \frac{\sigma_A^2}{M}~,
\label{eq:var}
\end{align}
where we have dropped the time argument for brevity. Here, $K_B$ denotes the PSFF defined in the subsystem $B$ and the sum extends over all subsystems $B\subseteq A$. The variance of $\widehat{K(t)}$ [Eq.~\eqref{eq:sffmeas}] follows by taking $A$ to be the full system.
We obtain an expected relative error 
$\mathcal E_A=\sigma_A/( K_A\sqrt{M})$ 
of an estimation $\widehat{K_A(t)}$ with $M$ experimental runs. 
As it can be rigorously shown via Chebyshev's inequality, the required number of measurements  to obtain with high probability an estimate of $K_A(t)$ with fixed relative error scales as $M\sim \sigma_A^2/K_A^2$.
\begin{figure}[!ht]
\includegraphics[width=\linewidth]{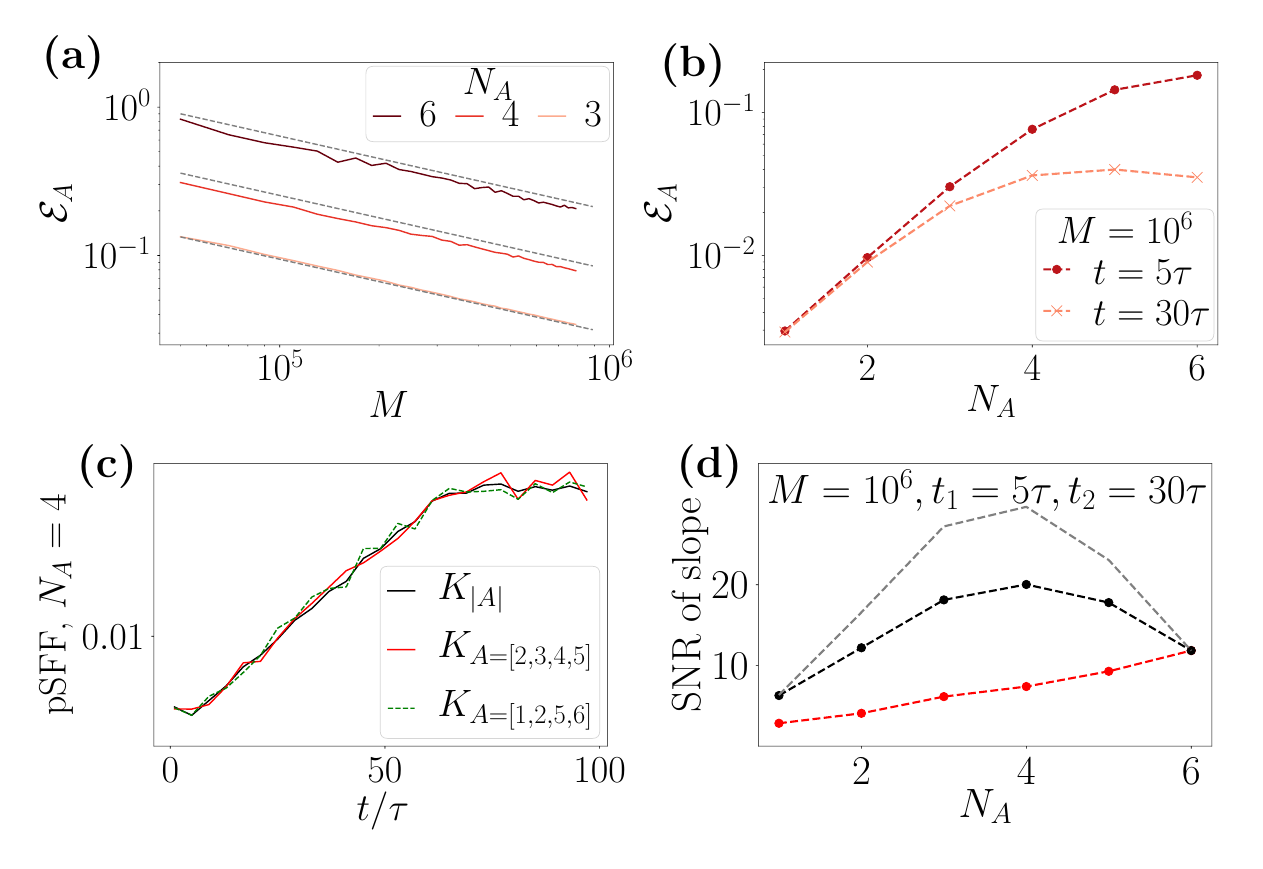}
\caption{\textit{Statistical errors in the Floquet model $V_3$.} (a) The relative error, $\mathcal{E}_A=\sigma_A/(K_A \sqrt M)$ is plotted for total system size $N=6$ and subsystem sizes $N_A=3,4, $ and $6$ as a function of number of measurements $M$ at time $t/\tau=5$. For a fixed $M$ we perform 100 numerical experiments each with $M$ single shots and present the average $\mathcal{E}_A$ using colored lines. The gray lines represent corresponding errors in a model with CUE dynamics (calculated analytically in App.~\ref{app:staterrors}). (b) The relative error $\mathcal{E}_A$ as a function of subsystem size $N_A$ at two times $t/\tau=5$ and $t/\tau=30$ is shown.
(c) Single PSFF with subsystem size $N_A=4$ for two choices $A=[2,3,4,5]$ (red), $A=[1,2,5,6]$ (green dashed), and the average PSFF $K_{|A|}$ (black) follow each other;  the numbers in the $[\cdots]$ denote qubit index.  (d) For the observation of the ramp we plot the SNR of the slope, SNR$[c_A(t_2, t_1)]$ (in red) and SNR[$c_{|A|}(t_2, t_1)$] (in black). Both  SNRs  are constructed from a single data set of $M=10^6$. As a guide to the eye, we also present in gray the SNR[$c_{|A|}(t_2, t_1)$] when all the measurements in the averaged PSFF are done independently, i.e.\ when $\sigma_{|A|}=\sigma_A \binom{N}{N_A}^{-1/2}$. This would require $M \binom{N}{N_A}$ number of independent measurements.}
\label{fig:staterror}
\end{figure}

The expected statistical error $\mathcal E_A$, and hence also the number of required experimental runs,  depends thus on the value of $K_A$ itself, as well as on the PSFF $K_B$ of all subsystems $B\subseteq A$. For Hamiltonians (Floquet-) operators from Wigner-Dyson RMT, we can explicitly evaluate $\sigma_A$ (see App.~\ref{app:staterrors}).  As the worst-case estimate, we find that at the point of weakest signal, after a single time step $t=\tau$ in Floquet dynamics  $T(t=n\tau)=V^n$ with $V$ sampled from CUE where $K_A(t=1\tau)=2^{-2N_A}$, the expected relative statistical error is given by  
$\mathcal{E}_A= \sqrt{10^{N_A}/M}$.  A  total number of measurements $M\sim 10^{N_A}/\epsilon^2 \approx 2^{3.32N_A} /\epsilon^2$ is thus required to obtain a fixed relative error $\epsilon$.
This is to be contrasted with the number of measurements required for quantum process tomography, which requires, without strong assumptions on the process of interest \cite{Torlai2020}, at least $r2^{5 N_A}/\epsilon^2$ measurements, with $r=r(N_A)\geq 1$ being the Kraus rank of the process \cite{Kliesch2019}.
In addition, we can reduce  the exponents associated with the scaling of statistical errors in randomized measurement protocols further using importance sampling~\cite{Rath2021,Hadfield2020,Huang2021,Hillmich2021}.

In Fig.~\ref{fig:staterror}(a) we plot the relative error $\mathcal{E}_A$ as a function of the number of experimental runs $M$ in the $V_3$ model \eqref{eq:Floquet-illustration-V3} at time $t/\tau=5$ with total qubits $N=6$.  The  relative error decays as $\sim 1/\sqrt{M}$ with increasing $M$, as expected from the central limit theorem. Furthermore, it decreases with decreasing subsystem size. This is also  shown in Fig.~\ref{fig:staterror}(b) where we display, for a fixed $M$, the relative errors as a function of subsystem size $N_A$  at two different times $t/\tau=5$ and $t/\tau=30$. As expected, we observe that the relative error is largest at early times where the PSFF is smallest. At early times, the relative error increases with the subsystem size, thereby requiring more  measurements as $N_A\rightarrow N$.

\subsubsection{Observing the ramp in chaotic models}

The relative error $\mathcal E_A=\sigma_A/( K_A\sqrt{M})$ determines the required number of measurements to estimate the PSFF at a {single} point in time. While this reveals important information on the overall magnitude and in particular the `shift' of the PSFF, signatures of energy level repulsion are encoded in the ramp of the SFF and PSFF (see Sec.~\ref{sec:psffAnalytical}). To detect the ramp, we aim thus to measure the difference $K_A(t_2)- K_A(t_1)$  at two points in time $t_2>t_1$, in particular, the slope of $K_A$,
\begin{align}
    c_A(t_2,t_1)=\frac{K_A(t_2)- K_A(t_1)}{t_2-t_1} .
\end{align}
To quantify the experimental effort to resolve $c_A(t_2,t_1)$, we introduce its signal-to-noise ratio $\text{SNR}[c_A(t_2,t_1)]$, which, for independent measurements of the PSFF at times $t_2$ and $t_1$, is given by
\begin{align}
    \text{SNR}[c_A(t_2,t_1)]=\sqrt{M}\frac{K_A(t_2)- K_A(t_1)}{{\sigma_{A}(t_2)+ \sigma_{A}(t_1)} }. 
\end{align}
As shown in Secs.~\ref{sec:psffAnalytical} and \ref{sec:psffNumerical}, the slope $c_A(t_2,t_1)$ of the PSFF (i.e.\ the signal), is approximately constant as a function of the subsystem size $N_A\gtrsim N/2$. At the same time, the absolute value of the noise,   here $(\sigma_{A}(t_2)+ \sigma_{A}(t_1))/\sqrt{M}$, decreases with increasing $N_A$ (as the absolute value of the PSFF decreases). Thus, as shown in Fig.~\ref{fig:staterror}(d) (red curve) for the $V_3$ model, $\text{SNR}[c_A(t_2,t_1)]$ typically increases with increasing subsystem size $N_A$,  reaching a maximum when the subsystem is the system itself i.e., $N_A=N\,(=6$ in the example here). 

In chaotic quantum systems, our protocol enables detection of the ramp with further improved SNR: First, we note that the order of magnitude of different features of the PSFF does not depend on the actual choice of the subsystem $A$, but only on its size $|A|=N_A$. Hence, as numerically shown in Fig.~\ref{fig:staterror}(c), we can replace the PSFF $K_A$ of a specific subsystem $A$ with its average
\begin{align}
    K_{|A|}(t) = \binom{N}{N_A}^{-1}\sum_{ |A|=N_A} K_{A}(t)~,
\end{align}
where we sum over all  subsystems $A$ of fixed size $N_A$ (including disconnected subsystems).

Second, we note that from a single experimental data set, taken on the full system $\mathcal{S}$, we can estimate $K_A(t)$ for all subsystems $A\subseteq \mathcal{S}$, via spatial restriction in the post-processing. Thus, we can also obtain the average PSFF $K_{|A|}(t)$  and its slope $c_{|A|}(t_2,t_1)$. Since for $N_A<N$, there are multiple subsystems $A$ of size $N_A$, we can expect an increased $\text{SNR}[c_{|A|}(t_2,t_1)]$ for these average quantities. 

In Fig.~\ref{fig:staterror}(d), we display the numerically determined signal-to-noise-ratio  $\text{SNR}[c_{|A|}(t_2,t_1)]$, for the averaged PSFF in black. Indeed, compared to the SNR for a single subsystem $A$, SNR[$c_A(t_2, t_1)$] in red, we observe an enhanced $\text{SNR}[c_{|A|}(t_2,t_1)]$ for subsystem sizes $1<N_A<N$. We remark that we do not reach an  enhancement $\binom{N}{N_A}^{1/2}$ of the SNR which would result trivially from $\binom{N}{N_A}$ separate experiments (i.e.\ $\binom{N}{N_A} \cdot M$ experimental runs in total, gray line) since  the estimations $\widehat{K_A(t)}$ for various subsystems $A$ from a single data set are not independent. Nevertheless, Fig.~\ref{fig:staterror}(d) shows that the average PSFF $K_{|A|}$, extracted  at a subsystem size $N_A\approx N/2$ has the largest SNR for determining the slope of the ramp from a given measurement dataset. Thus, as compared to the PSFFs $K_A(t)$ for fixed subsystems $A$ or the full SFF $K(t)$, the average PSFF $K_{|A|(t)}$ at half system size provides a favorable tool to observe the ramp of the (P)SFF, i.e.\ signatures of level repulsion in chaotic quantum many systems.

\subsection{Experimental imperfections}

\begin{figure}[ht]
\includegraphics[width=\linewidth]{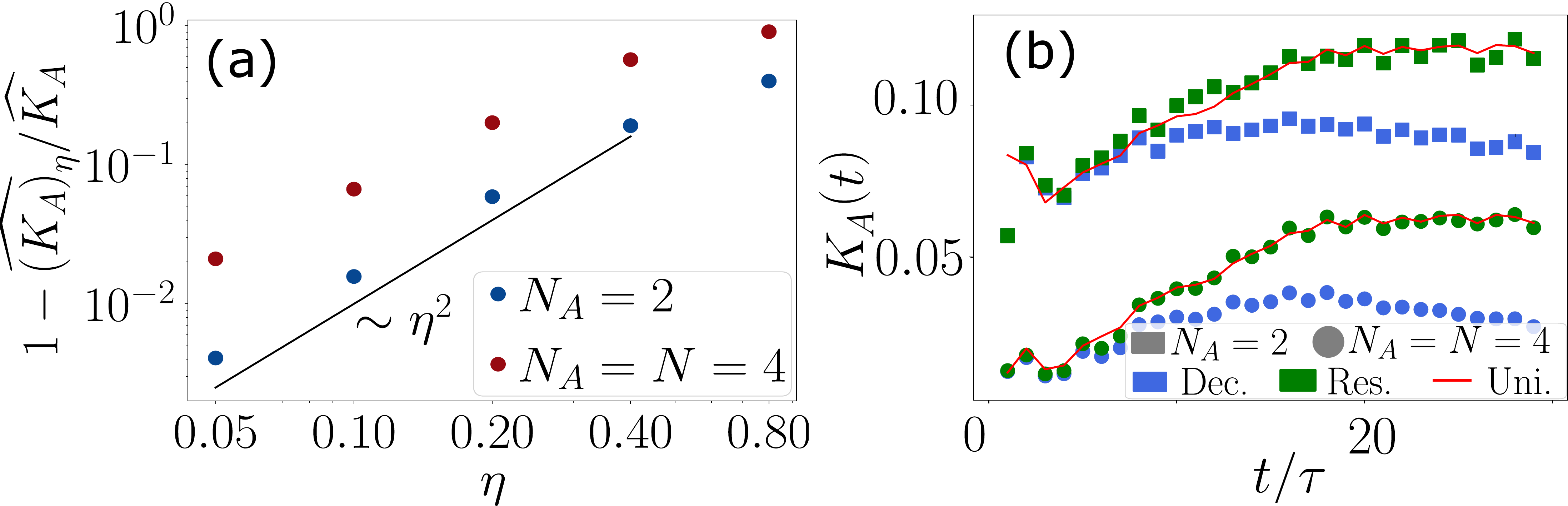}
\caption{\textit{Experimental imperfections and decoherence.} We study effects of measurement errors and decoherence on the estimated SFF and PSFF $K_A(t)$ using the example of the kicked spin  $V_3$  with total system size $N=4$.  In (a), we display the relative error  $\widehat{\epsilon_\eta}=( \widehat{K_A}-\widehat{(K_A)}_\eta)/\widehat{K_A}$  of the estimated form factors induced by a decorrelation of local random unitaries applied before and after the time evolution up to the Heisenberg time $t_H$, with strength $\eta$ (see text). In (b), we display the estimated SFF $(K)_{\text{dec}}$ (blue dots) and PSFF $(K_A)_{\text{dec}}$ (blue squares)  as function of time in a system subject to global polarization
with strength $p=0.03$ (see text). For this type of decoherence,  rescaling according to Eq.~\eqref{eq:rescale}, allows to recover the  SFF (green dots) and PSFF (green squares) for unitary dynamics (red line).}
\label{fig:imperfections}
\end{figure}

First, we consider an imperfect implementation of our measurement protocols, with errors arising from an erroneous  decorrelation of the applied initial and final local random unitaries. We model such imperfection as the effective application of a unitary $u_i$ before and a unitary $v_i=u^\dagger_i \exp(-i \eta h_i) $ after the time evolution, with $h_i$ being a local random Hermitian  matrix  sampled for each $i$ independently from the GUE \cite{HaakeBook}. While the case $\eta=0$ corresponds to the ideal case, we display in Fig.~\ref{fig:imperfections}(a) the average relative error $\widehat{\epsilon_\eta}=1-\widehat{(K_A(t))_\eta}/\widehat{K_A(t)}$ of the estimated $\widehat{(K_A(t))_\eta}$ as a function of the error strength $\eta$, obtained numerically from simulating many experimental runs. We find  $\widehat{\epsilon_\eta}$ increases approximately as $\eta^2$, indicating a decrease  of the estimated $\widehat{(K_A(t))_\eta}$.

Secondly, we consider that a measurement of the SFF and PSFF is affected by decoherence acting during the dynamical evolution of the system.
As shown in the context of other randomized measurement protocols, one can correct the effect of depolarization errors (or readout errors) based on a randomized measurement of the purity~\cite{VanEnk2012,Vermersch2018, Elben2018, Brydges2019}, which allows to extract the value of the noise strength~\cite{Vermersch2018,vovrosh2021simple}. 
Note that if the type of noise is a priori unknown, one can also mitigate errors with randomized measurements.
This is done via a calibration step that allows to convert  randomized measurements into 
 faithful `classical shadows' estimations of the quantum state~\cite{chen2020robust, berg2021modelfree, hillmich2021decision}.

Here for concreteness, we consider a Floquet system with  global depolarization, acting at each time period $\tau$ with strength $p$, i.e.\ the  final state $\rho_f(t)$ at time $t=\tau n$, defined in Sec.~\ref{sec:protocol}, is altered to \mbox{$\rho_{\text{dec}}(t)=\alpha_n \rho_f(t) + (1-\alpha_n)\,  \mathbb{1}/D$} with $\alpha_n=(1-p)^n$.

Thus, we obtain via our measurement protocol,
\begin{align}
  (K_A)_{\text{dec}}(t)=  \alpha_n K_A(t) +  \frac{1-\alpha_n}{D_A^2}\; .
  \label{eq:dec-pSFF}
\end{align}
With increasing time $t=\tau n$, decoherence leads thus to a smaller measured value $(K_A)_{\text{dec}}(t)$ than the actual spectral form factor $ K_A(t)$ (see Fig.~\ref{fig:imperfections}(b), blue dots and squares). 
However, if we know the value of $p$, we can rescale our estimator of the SFF. For this purpose, we can measure the purity  of the time evolved state. The purity is,
\begin{eqnarray}
P_n &=& {\rm{Tr}}\left[\rho_{\rm{dec}}(t)^2\right]= \alpha_n^2+\frac{1-\alpha_n^2}{D}~,
\end{eqnarray}
which gives~\cite{Vermersch2018,vovrosh2021simple},
\begin{equation}
    \alpha_n = \sqrt{\frac{D P_n-1}{D-1}}~.
    \label{eq:alpha}
\end{equation}
Thus, from a measurement of the purity $P_n$ at all times, we can find $\alpha_n$ and  rescale the erroneous PSFF \eqref{eq:dec-pSFF} to obtain,
\be
(K_A)_{\rm{res}}(t)=\frac{(K_A)_{\rm{dec}}(t)-(1-\alpha_n)/D_A^2}{\alpha_n}~.
\label{eq:rescale}
\ee
In Fig.~\ref{fig:imperfections}(b), using the green color we present this rescaled SFF (using dots) and PSFF (using squares). We note that using the  rescaled (P)SFF \eqref{eq:rescale} we recover here the  (P)SFF of the unitary dynamics (red curve). 

In summary,  while we have shown in this subsection that we can  partially  correct for decoherence effects via independent measurements of decoherence parameters, we emphasize that imperfections and decoherence discussed in this section lead  to a decay of the estimated $\widehat{K_A(t)}$. They, thus can not cause a false positive detection of the ramp.

\section{Conclusion and outlook} 
\label{sec:conclusion}

In this work, we have presented randomized measurement protocols to access the statistics of energy eigenvalues and energy eigenstates of many-body quantum systems in present day quantum simulators via (partial) spectral form factors. The spectral form factor (SFF), $K(t)$ in Eq.~\eqref{eq:sff}, is known to be a key diagnostic of many-body quantum chaos. In chaotic systems, it reveals universal properties of energy eigenvalue statistics and possesses a characteristic ramp-plateau structure (see Sec.~\ref{sec:sff-example}). In addition, we have defined partial spectral form factors (PSFFs), $K_A(t)$ in  Eq.~\eqref{eq:psff}, which contain both the statistics of energy eigenvalues and eigenstates (see Sec.~\ref{sec:psff-example}). PSFFs are natural restrictions of the SFF to subsystems $A\subseteq \mathcal{S}$ of the full system $\mathcal{S}$, such that for  $A=\mathcal  S$, PSFF and SFF coincide $K_{A=\mathcal S}(t)=K(t)$. Utilizing random matrix theory and the eigenstate thermalization hypothesis (ETH), we have shown in Sec.~\ref{sec:psffAnalytical}, that PSFFs in generic chaotic quantum many-body systems possess a characteristic shift-ramp-plateau structure [Eqs.~\eqref{eq:psffrmt} and \eqref{eq:pSFF_ETH}] and reveal crucial differences between  thermal and non-thermal eigenstates in the sense of ETH. In Sec.~\ref{sec:psffNumerical} we investigated the PSFF numerically  with examples of many-body quantum models,  discussing, in particular, differences between chaotic and localized phases.

With our protocol to measure the SFF and PSFF in quantum simulation experiments, we have  extended the toolbox of randomized measurements to access genuine properties of  dynamical quantum evolution, without any reference to the initial state or  measured observable (see Secs.~\ref{sec:protocol}, \ref{sec:proof_mt} and \ref{sec:errors}). We have shown that our protocol gives simultaneous access to the SFF and PSFF, thereby providing a unified testbed of the statistical properties of eigenvalues and eigenstates. Our protocol can be directly implemented in state-of-the-art quantum devices, based for instance on  trapped ions \cite{Blatt2012,Monroe2021}, Rydberg atoms \cite{Browaeys2020} and superconducting qubits \cite{Kjaergaard2020,Mi2021}, providing crucial experimental tools for the quantum simulation  of many-body quantum chaos and the study of thermalization in closed quantum systems. 

Our work can be generalized in various directions. First, while we have concentrated here on quantum simulators with local control realizing lattice spin models, our protocol can be also realized in collective spin systems with only global operations \cite{Sieberer_2019}. Second, while we have considered  form factors which are second-order functionals of the time evolution operators $T(t)$, partial restrictions of higher-order form factors provide possibilities to investigate thermalization of quantum many-body systems and emergent randomness beyond second-order \cite{choi2021emergent,cotler2021emergent}.  To access such higher-order (partial) form factors,  our randomized measurement  protocols could be readily combined with the classical shadows framework \cite{Huang2020}.
Thirdly, we have focused on determining the properties of unitary quantum dynamics. Beyond that, our measurement protocol readily extends to the study  of noisy quantum channels. This includes applications in the field of verification and benchmarking of quantum devices \cite{Emerson2005,Emerson2007,Knill2008,Magesan2012,erhard2019characterizing,eisert2020quantum,Carrasco_2021}, as well as the investigation of noise-induced quantum many-body phenomena such as entanglement phase transitions \cite{Li2018,Skinner2019,Chan2019,Vitale2021}. In addition to the directions listed above, it will be interesting to explore the PSFF from an analytical perspective analogous to Ref.~\cite{GarrattChalker} to study the physics of thermalization and entanglement in Hamiltonian many-body systems as well as in quantum gravity, where there have recently been path integral derivations of the SFF \cite{Saad2018}.

\begin{acknowledgments}
We thank Mikhail Baranov, Amos Chan, Manoj K.\ Joshi, Barbara Kraus, Rohan Poojary, Lukas Sieberer and  Denis Vasilyev for valuable discussions. Work in Innsbruck has been supported by the European Union's Horizon 2020 research and innovation programme under Grant Agreement No.\ 817482 (Pasquans) and No.\ 731473 (QuantERA via QT-FLAG), by the Austrian Science Foundation (FWF, P 32597 N), by the Simons Collaboration on Ultra-Quantum Matter, which is a grant from the Simons Foundation (651440, P.Z.), and by LASCEM by AFOSR No.\ 64896-PH-QC.  A.E.\ acknowledges funding by the German National Academy of Sciences Leopoldina under the grant number LPDS 2021-02. BV acknowledges funding from the French National Research Agency (ANR-20-CE47-0005, JCJC project QRand). A.V.\ and  V.G.\  were  supported by US-ARO Contract No.W911NF1310172, NSF DMR-2037158,  and the Simons Foundation.
\end{acknowledgments}

\appendix

\section{Spectral form factor in Wigner-Dyson random matrix ensembles}
\label{app:SFF-RMT}
In this appendix, we review the definition and essential properties of the Wigner-Dyson random matrix ensembles. Further, we recall the expressions of the SFF for Hamiltonian and Floquet dynamics modeled with random matrices from these ensembles.

The Wigner-Dyson ensembles are standard distributions of random matrices used to model some of the properties of energy or quasi-energy eigenvalues and eigenstates of chaotic Hamiltonian and Floquet systems \cite{Wigner1955, Dyson1962, Mehta2004, HaakeBook}. We work with two classes of the Wigner-Dyson ensembles - the unitary (U) class for systems that are not time reversal invariant, and the orthogonal (O) class for some systems with time-reversal invariance (the symplectic (S) class applies to other systems with time-reversal invariance, but is not relevant for our examples). We note in particular that nonconventional time-reversal symmetries should also be considered \cite{HaakeBook} e.g.\ invariance under complex conjugation in some basis (which corresponds to the orthogonal class). Each class is characterized by a symmetry group comprised of the corresponding set of similarity transformations (i.e.\ all unitary or orthogonal transformations).

For Hamiltonian systems with time evolution operator $T(t)=\exp(-iHt)$,
it is conventional to choose the Gaussian Unitary Ensemble (GUE) of Hermitian matrices 
or the Gaussian Orthogonal Ensemble (GOE) of real symmetric matrices 
to represent the Hamiltonian $H$ of the appropriate class. In the case of periodically driven Floquet dynamics with time-evolution operator $T(t=\tau n)=V^n, n\in \mathbb N$, where $V$ is the unitary Floquet operator corresponding to a time period $\tau$, the appropriate representative ensembles for $V$ are the Circular Unitary Ensemble (CUE) of unitary matrices and the Circular Orthogonal Ensemble (COE) of symmetric unitary matrices.
These ensembles accurately model the local eigenvalue correlations of the corresponding systems (but not necessarily global eigenvalue features larger than the inverse Thouless time scale~\cite{DAlessio2016,Gharibyan2018} e.g.\ the smoothened density of states), and describe an idealization of the eigenstate distribution (which is generalized by ETH \cite{DAlessio2016, subETH}). But for the special case of chaotic Floquet systems, the eigenstate distribution is seen to be in close agreement with the Wigner-Dyson ensembles \cite{Regnault2016,PhysRevX.4.041048,lazarides2014floqueteth,ponte2015floqueteth,kim2014floqueteth, GarrattChalker}.

For these random matrix models, the spectral form factor can be calculated analytically (see for instance Ref.~\cite{Liu2018}). For completeness, we recall the well-known expressions here.
For  Hamiltonians $H$  from GUE or GOE, one finds
\begin{description}
\item[GUE model]
\begin{align}
K(t)&= r(t)^2+\frac{1}{D}
\begin{cases}
\frac{t}{t_H} & \text{ for } 0< t\le t_H ,\\
1 & \text{ for } t > t_H,
\end{cases}
\label{eq:sffgue}
\end{align}
\item[GOE model]
\begin{align}
K(t)&= r(t)^2+\!\frac{1}{D}
\begin{cases}
2\frac{t}{t_H}-\f{t}{t_H}\log{\left(1+2\f {t} {t_H}\right)}\!\!\! & \!\!\text{ for } 0< t\le t_H,\\
2-\f{t}{t_H}\log{\left(\f{2t+t_H}{2t-t_H}\right)} \!\!\!& \!\!\text{ for } t > t_H,
\end{cases}
\label{eq:sffgoe}
\end{align}
\end{description}
where  $r(t)=t_H J_1(4D t/t_H)/(2Dt)$ with $J_1$ denoting the Bessel's function of the first kind. The Heisenberg time $t_H$, connected to the inverse spacing of adjacent energy levels, depends on the width of the Gaussian distribution of the matrix elements and marks the onset time of the plateau of the SFF. 
For the results presented in Sec.~\eqref{sec:psffNumerical}, we fix it numerically, by matching plateau onset times for the Hamiltonian  Eq.~\eqref{eq:Hamiltonian} and the GOE model. 

For the Floquet operators $V$ from  CUE or COE, one finds 
\begin{description}
    \item[CUE model] \begin{align}
    K(t)= \frac{1}{\D}
\begin{cases}
   \frac{t}{t_H} , & \text{for } 0<t \le t_H,\\
    1,              & \text{for } t > t_H,
\end{cases}
\label{eq:sffcue}
\end{align}
    \item[COE model] \begin{equation}
  K(t)= \frac{1}{\D}
\begin{cases}
\frac{2t}{ t_H}-\frac{t}{t_H}\log\left(1+2\frac{t}{t_H}\right) & \text{ for } 0<t \le  t_H ,\\
2 - \frac{t}{t_H} \log\left(\f{2 t/t_H+1}{2t/t_H-1}\right) & \text{ for } t > t_H ,
\end{cases}
\label{eq:sffcoe}
\end{equation}
\end{description}
Here, $t_H=D\tau$ with $\tau$ to be identified with the period of the Floquet system to be modeled.

\section{Partial spectral form factor in Wigner-Dyson random matrix ensembles}
\label{app:pSFF-RMT}
In this section, we derive the functional form of the partial spectral form factors, discussed in Sec.~\ref{sec:psffAnalytical}, for Hamiltonian dynamics (Floquet dynamics) modeled with the Wigner-Dyson random matrix ensembles GUE, GOE (CUE, COE), as introduced  in App.~\ref{app:SFF-RMT}.

Let $\mathcal{S}$ be a quantum system  with Hilbert space $\mathcal{H}$ of dimension $\D$, and $A\subseteq \mathcal{S}$ a subsystem with dimension $\DA$. Its complement is denoted with $B$ with dimension $\DB$. 
As discussed in App.~\ref{app:SFF-RMT}, we consider
\begin{itemize}
    \item Hamiltonian dynamics $T(t) = \exp(-iHt) $ with $H$ sampled from the GUE and GOE, respectively.
    \item Floquet dynamics with $T(t=\tau n ) = V^n$ for $n\in \mathbb{N}$ with $V$ sampled from the CUE and COE, respectively.
\end{itemize}
We can rewrite $T(t) = Y D(t) Y^\dagger$ with $D(t)=\text{diag}(e^{-iE_1 t}, \dots, e^{-iE_{\D} t})$, the diagonal matrix of eigenvalues  of $T(t)$ and $Y=(y_1 , \dots , y_D)$  the unitary (GUE, CUE) or orthogonal (GOE, COE) matrix of eigenvectors of $H$ or $V$. Crucially, we note that all time-dependence is contained in the diagonal matrix $D(t)$. In the following, we rely on the  fact:
\newtheorem{fact}{Fact}
\begin{fact}
For $H$ from GUE or GOE ($V$ from CUE or COE), the distribution of the eigenvectors of $H$ ($V$) is independent of the distribution of eigenvalues  of $H$  ($V$). Further, $Y=(y_1 , \dots , y_{\D})$ is distributed according to the Haar measure on the  group of unitary matrices $U(\D)$ (for GUE, CUE) and the group of orthogonal matrices $O(\D)$ (for GOE, COE).
\end{fact}
\begin{proof}
This fact relies only on the invariance  of the random matrix ensembles under unitary (GUE, CUE) and orthogonal transformations (GOE, COE). 
For GUE and GOE, a proof is given in Ref.~\cite{Anderson2010}, Corollary 2.5.4. It generalizes directly to CUE and COE.
\end{proof}
Using this fact, we can carry out the average over eigenvectors  in Eq.~\eqref{eq:psff} explicitly (see next subsection). With the identification $\SFF(t) = D^{-2}  \overline{|\tr{D(t)}|^2}$, we find
\begin{equation}
    \SFF[A](t)={c_A^{(1)}}+c_A^{(2)}\SFF(t)~,
    \label{eq:general-RMT}
\end{equation}
where $\text{for} ~ H \in \text{GUE} \; , \;  V\in \text{CUE}$,
\begin{align}
c_A^{(1)} &= \frac{\left.{\DB}^2-1\right. }{{\DA}^2 {\DB}^2-1}~~;~~
c_A^{(2)} =  \frac{\DB^2 \left({\DA}^2-1\right) }{{\DA}^2 {\DB}^2-1} ~,
\label{eq:app_pSFF_U}
\end{align}
and $\text{for} ~ H \in \text{GOE} \; , \;  V\in \text{COE}$,
\begin{align}
  c_A^{(1)} &=   \frac{ \left(\DB^2+\DB-2\right)}{({\DA} {\DB}-1)({\DA} {\DB}+2)}~; \nonumber\\
    c_A^{(2)} =&\frac{\DB\left(\DA \DB +\DB +1\right)\left( {\DA} - 1 \right)} {({\DA} {\DB}-1)({\DA} {\DB}+2)}~. \label{eq:app_pSFF_O}
\end{align}
In particular, $\SFF[A=\mathcal{S}](t)=\SFF(t)$ for $\DA=\D, \DB=1$ and $\SFF[A=\emptyset](t)=1$ for $\DA=1,\DB=\D$ holds, as expected. 
\paragraph*{Relation to average purity and overlap:} For Hamiltonian  $T(t)=\exp(-iHt)$ or Floquet dynamics $T(t=n\tau)=V^n$, we can rewrite the PSFF in terms of the (quasi-) energy eigenvalues and (quasi-) energy eigenstates [see Eq.~\eqref{eq:psff}]. For Hamiltonians $H$ (Floquet operators $V$) from the Wigner-Dyson random matrix ensembles we can use then fact 1 to obtain  the PSFF in terms of the average purity  $\mathcal{P}_B$ of reduced eigenstates and average overlap of distinct reduced eigenstates $\mathcal{Q}_B$  [see Sec.~\ref{sec:psffAnalytical}, in particular Eq.\ \eqref{eq:psffrmt}]. Comparing Eq.\ \eqref{eq:psffrmt} with Eq.~\eqref{eq:general-RMT} we find that
\begin{align}
c_A^{(1)}=\frac{\mathcal{P}_B-\mathcal{Q}_B}{D_A}\quad \text{and}\quad
c_A^{(2)}=D_B\mathcal{Q}_B~.
\end{align}
Using this, Eqs.~\eqref{eq:app_pSFF_U} and Eqs.~\eqref{eq:app_pSFF_O}, we obtain Eq.~\eqref{eq:RMT-PandQ} (for GUE, CUE) and the corresponding expressions for the orthogonal ensembles (GOE, COE), respectively.

\subsection*{Proof of Eqs.~\eqref{eq:general-RMT}, \eqref{eq:app_pSFF_U} and \eqref{eq:app_pSFF_O}}
We denote the basis of $\mathcal{H}$ consisting of eigenvectors  of $T(t)$ with $\ket{i}$ ($i=1,\dots, D$). Furthermore, we fix an arbitrary product basis of  $\mathcal{H}=\mathcal{H}_A \otimes \mathcal{H}_B$ as $\ket{a,b}$ with $a=1,\dots, \DA$ and $b=1,\dots, \DB$. With $T(t) = Y D(t) Y^\dagger$, we rewrite Eq.~\eqref{eq:psff} in these bases. Using the independence of eigenvalues and eigenvectors (Fact 1), we find 
\begin{align}
   DD_A  \SFF[A](t) &= \overline{ \tr[B]{\tr[A]{T(t)  } \tr[A]{(T(t)^\dagger}} } \\
    =  &  \overline{ Y_{(a_1,b_1),i_1} (Y^\dagger)_{i_1,(a_1,b_2)} Y_{(a_2,b_2),i_2}  (Y^\dagger)_{i_2,(a_2,b_1)}  }  \times \nonumber\\ & \overline{D(t)_{i_1, i_1} (D(t)^\dagger)_{i_2, i_2} }~,
    \label{eq:evdecomp}
\end{align}
where summation over repeated indices is understood. The ensemble average over the matrix elements of $Y$  can be carried out using the Weingarten calculus on the unitary group (GUE and CUE) and orthogonal group (GOE and COE), respectively.

\begin{figure}[t]
    \centering
    \includegraphics[width=0.9 \linewidth]{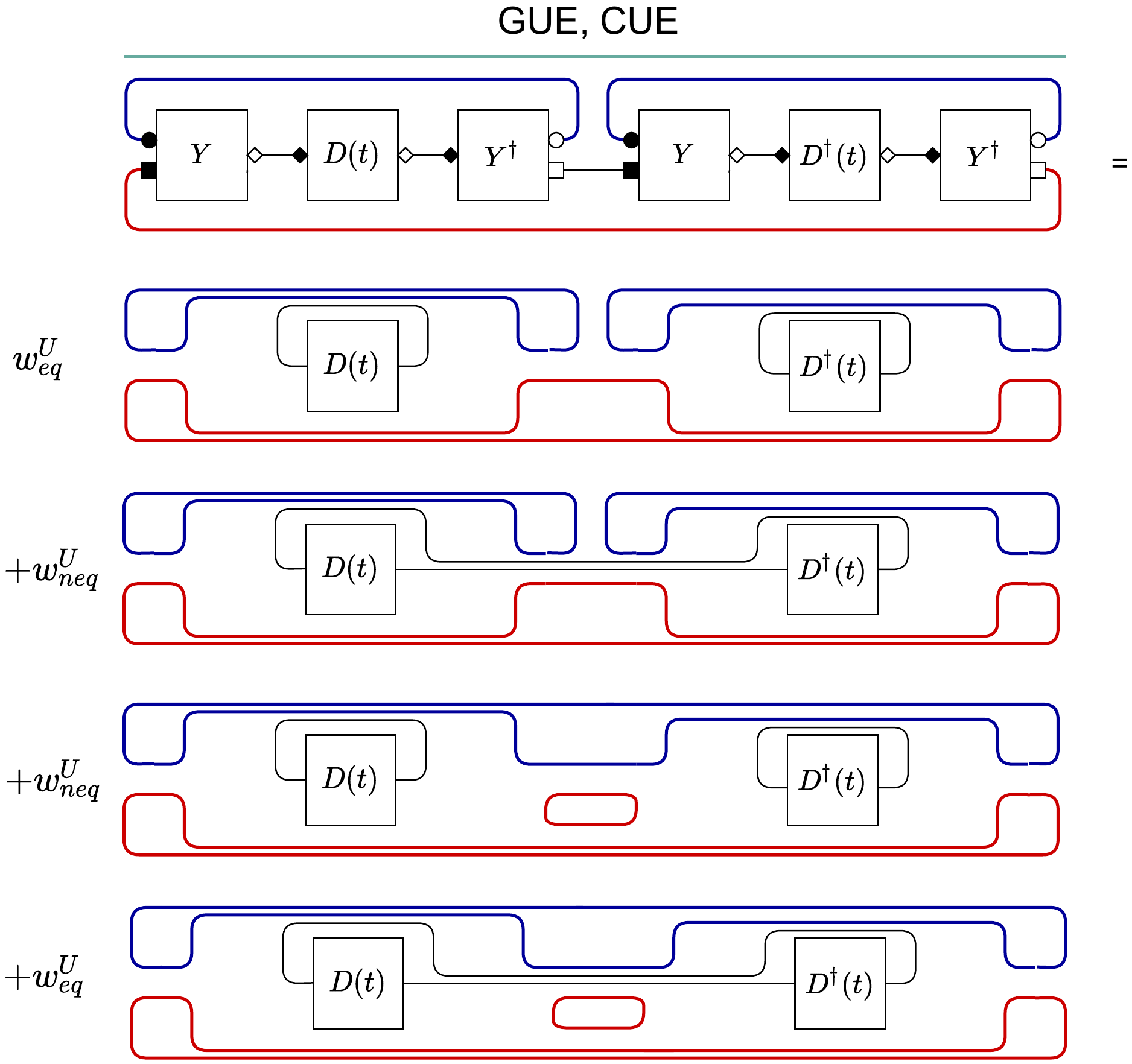}
   \caption{ \textit{Diagrammatic evalution of Eq.~\eqref{eq:evdecomp} for $Y \in U(D)$ (GUE, CUE case).} To perform the average over eigenvectors (green line), we remove the boxes $Y$ and connect white decorations of $Y$ (rhombi) with white decorations of $Y^*$ (rhombi) and black decorations of $Y$  (circles, squares) with black decorations of $Y^*$ (circles, squares) in all possible ways, corresponding to the pair partitions  $\mathbf{m}, \mathbf{n} \in \mathcal{M}^U(4)$ \cite{Collins2010,Elben2019}. Summing over the resulting diagrams, weighted with corresponding value of the Weingarten function, yields Eq.~\eqref{eq:app_pSFF_U}.
    In all diagrams, each blue loop contributes a factor  $\DA$, each red loop a factor $\DB$.}
\label{fig:diag_pSFF-U}
\end{figure}

\begin{figure}[t]
    \centering
    \includegraphics[width=0.9\linewidth]{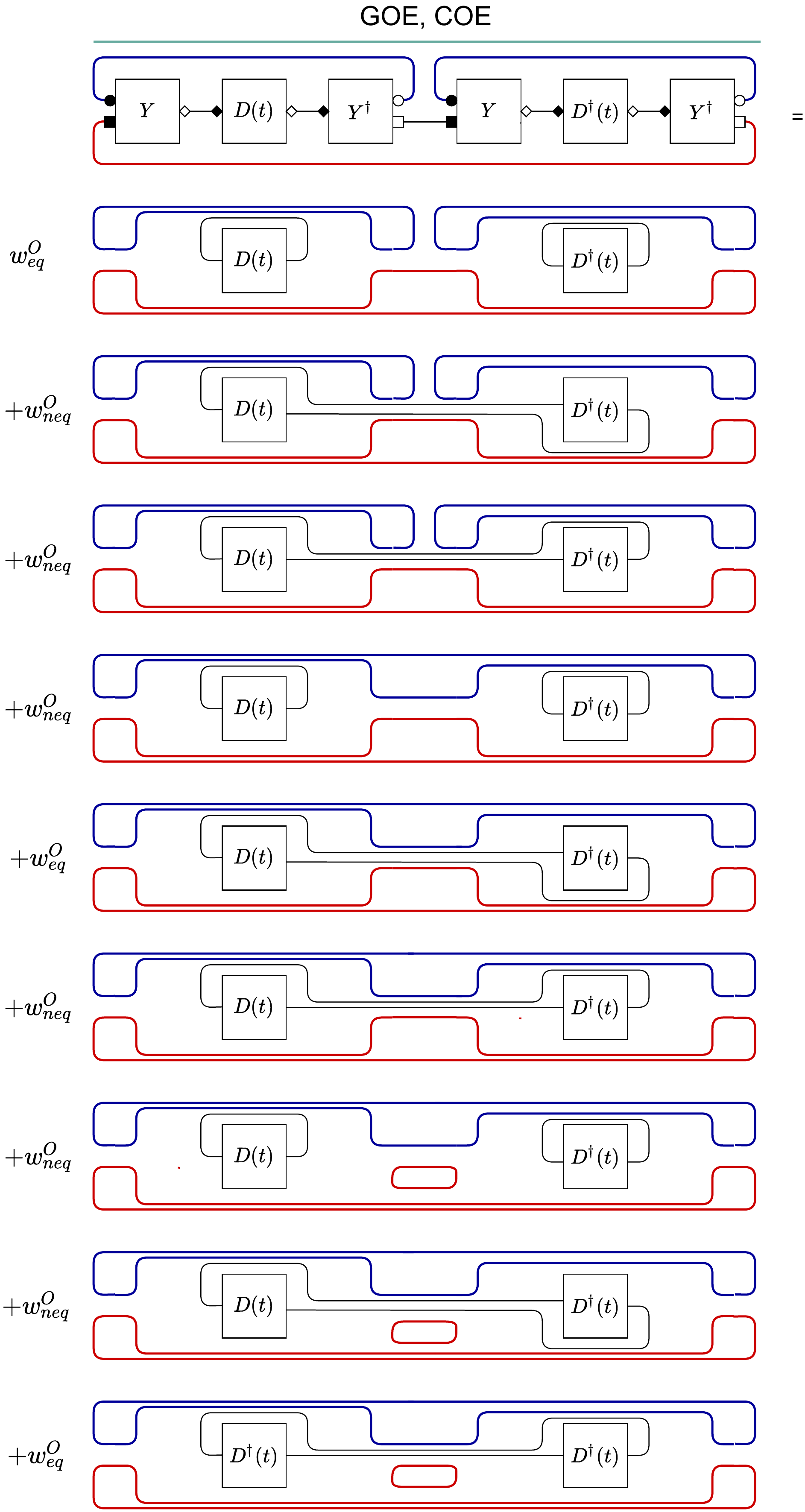}
\caption{\textit{Diagrammatic evalution of Eq.~\eqref{eq:evdecomp} for $Y \in O(D)$ (GOE, COE case).} To perform the average over eigenvectors (green line), we remove the boxes $Y$ and connect white decorations (rhombi)  with white decorations (rhombi) and black decorations (circles, squares)  with black decorations of same type (circles, squares) in all possible ways, corresponding to all pair partitions  $\mathbf{m}, \mathbf{n} \in \mathcal{M}^O(4)$ \cite{Collins2010,Elben2019}. Summing over the resulting diagrams, weighted with corresponding value of the Weingarten function, yields Eq.~\eqref{eq:app_pSFF_O}. In all diagrams, each blue loop contributes a factor  $\DA$, each red loop a factor $\DB$.}
\label{fig:diag_pSFF-O}
\end{figure}

The Weingarten calculus for the unitary group  and for orthogonal group  can be formulated in terms of pair partitions, defined as follows.
\newtheorem{definition}{Definition}
\begin{definition}[Pair partitions]
For $n\in \mathbb{N}$, 
a) we denote with $\mathcal{M}^{O}(2n)$ the set of all pair partitions of $\{1,\dots,2n\}$, partitioning $\{1,\dots,2n\}$  into $n$ distinct pairs. Then, each pair partition $\mathbf{m}\in \mathcal{M}^{O}(2n)$ can be uniquely expressed as 
\begin{align}
\left\{ \{ \mathbf{m}(1),\mathbf{m}(2)\} ,\dots \{\mathbf{m}(2n-1),\mathbf{m}(2n)\} \right\}
\end{align}
with $\mathbf{m}(1)<\mathbf{m}(3)<\dots<\mathbf{m}(2n-1)$ and $\mathbf{m}(2i-1)<\mathbf{m}(2i)$ for all $i\in \{1,\dots n\} $.\\
b) we denote with $\mathcal{M}^{U}(2n) \subseteq \mathcal{M}^{O}(2n)$ the set of all pair partitions of $\{1,\dots,2n\}$ which pair elements  in $\left\{1,\dots,n\right\}$ with elements $\left\{n+1,\dots,2n\right\}$. Then, each partition $\mathbf{m}\in \mathcal{M}^{U}(2n)$ can be uniquely expressed as 
\begin{align}
\left\{ \{\mathbf{m}(1),\mathbf{m}(2)\} ,\dots \{\mathbf{m}(2n-1),\mathbf{m}(2n)\} \right\}
\end{align}
with $\mathbf{m}(1)<\mathbf{m}(3)<\dots<\mathbf{m}(2n-1)$ and $\mathbf{m}(2i-1)\in \left\{1,\dots,n\right\}$ and $\mathbf{m}(2i)\in \left\{n+1,\dots,2n\right\}$ for all $i\in \{1,\dots n\} $.
\end{definition}

The following fact is shown in Ref.~\cite{Collins2009}.
\begin{fact}[Weingarten calculus]
(i) Let $Y$ be distributed according to the Haar measure on the orthogonal group $O(D)$. With indices $i_1, \dots, i_{2n}$ and $j_1,\dots, j_{2n}$ in $\left\{1,\dots, D \right\}$ it holds
\begin{align}
&\int_{Y\in O(D)} Y_{i_1,j_1} \cdots  Y_{i_{2n},j_{2n}} \text{d}Y =\nonumber\\
& \smashoperator{\sum_{\mathbf{m},\mathbf{n} \in \mathcal{M}^{O}(2n)}} \; \textup{Wg}^{O(D)}(\mathbf{m},\mathbf{n} ) \prod_{k=1}^{n} 
\delta_{i_{\mathbf{m}(2k-1)},i_{\mathbf{m}(2k)}} \delta_{j_{\mathbf{m}(2k-1)},j_{\mathbf{m}(2k)}}
\end{align}
with $\mathcal{M}^{O}(2n)$ the set of all pair partitions on $\left\{1,2, \dots, 2n\right\}$ and $\textup{Wg}^{O(D)}$ the Weingarten function on the orthogonal group $O(D)$.\\
(ii) Let $Y$ be distributed according to the Haar measure on the unitary group $U(D)$. With indices $i_1, \dots, i_{2n}$ and $j_1,\dots, j_{2n}$ in $\left\{1,\dots, D \right\}$ it holds
\begin{align}
&\int_{Y\in U(D)} Y_{i_1,j_1} \cdots Y_{i_n,j_n}   \, Y^{*}_{i_{n+1},j_{n+1}} Y^{*}_{i_{2n},j_{2n}}\text{d}Y  =\nonumber\\ &\smashoperator{\sum_{\mathbf{m},\mathbf{n} \in \mathcal{M}^U(2n)}}\; \textup{Wg}^{U(D)}(\mathbf{m},\mathbf{n} ) \prod_{k=1}^{n} 
\delta_{i_{\mathbf{m}(2k-1)},i_{\mathbf{m}(2k)}} \delta_{j_{\mathbf{m}(2k-1)},j_{\mathbf{m}(2k)}}
\label{eq:haar_unitary}
\end{align}
with $\mathcal{M}^{U}(2n) \subsetneq \mathcal{M}^{O}(2n)$ the set of all pair partitions on $\left\{1,2, \dots, 2n\right\}$ which pair elements  in $\left\{1,\dots,n\right\}$ with elements $\left\{n+1,\dots,2n\right\}$    and $\textup{Wg}^{U(D)}$ the Weingarten function on the unitary group $U(D)$.
\end{fact}
In our case, we are only interested in the case $n=2$.
As shown in Ref.~\cite{Collins2009},  when $\mathbf{m},\mathbf{n} \in  \mathcal{M}^{O}(4)$ and  $D\geq2$, 
\begin{align}
w^{O}_{\text{eq}} & \equiv \text{Wg}^{O(D)}(\mathbf{m},\mathbf{n} ) = \frac{D+1}{D (D+2)(D-1)} \; \text{for} \; \mathbf{m}=\mathbf{n}\nonumber\\
w^{O}_{\text{neq}} &\equiv \text{Wg}^{O(D)}(\mathbf{m},\mathbf{n} ) = \frac{-1}{D (D+2)(D-1)} \; \text{for} \; \mathbf{m}\neq\mathbf{n} ~ .
\end{align}
Furthermore, it holds  for $\mathbf{m},\mathbf{n} \in  \mathcal{M}^{U}(4)$ and  $D\geq2$
\begin{align}
w^{U}_{\text{eq}} &\equiv \text{Wg}^{U(D)}(\mathbf{m},\mathbf{n} ) = \frac{D}{D (D^2-1)} \quad \text{for} \quad \mathbf{m}=\mathbf{n} \nonumber\\
w^{U}_{\text{neq}} &\equiv \text{Wg}^{U(D)}(\mathbf{m},\mathbf{n} ) = \frac{-1}{D (D^2-1)} \quad \text{for} \quad \mathbf{m}\neq\mathbf{n}~ .
\end{align}
Using Fact 2 and these expressions, we can perform the average over eigenvector elements in Eq.~\eqref{eq:evdecomp} explicitly. This is most easily performed diagrammatically and shown in Figs.~\ref{fig:diag_pSFF-U} and \ref{fig:diag_pSFF-O}.

\section{Partial spectral form factor in general chaotic systems}
\label{app:pSFF-ETH}
Here, we derive the typical behavior of the PSFF for ensembles of chaotic systems, more general than random matrix ensembles, as considered in Sec.~\ref{sec:pSFFchaotic} of the main text. As in Eq.~\eqref{eq:rho_separation}, we decompose the reduced density matrix into a pure trace, a traceless smooth part and a traceless fluctuating part,  $\rho_B(E) = D_B^{-1}\mathbb{1}+\Delta\rho_B(E)+\delta\rho_B(E)$. For the smooth part, we assume that there exists an extrapolation of each matrix element to a continuous energy variable such that for some (as yet unspecified) time $t_\rho \ll O(D)$,
\begin{equation}
    (\Delta\tilde{\rho}_B(t))_{jk} \equiv \int\diff E\ e^{-i E t} (\Delta{\rho_B}(E))_{jk} = 0, \ \forall\ \lvert t \rvert > t_\rho~.
    \label{eq:rho_def_smooth}
\end{equation}
The remaining energy dependent part of $\rho_B(E)$ i.e.\ the part that oscillates rapidly and has no low frequency Fourier component (on extrapolation to continuous energy) will be taken to be the fluctuating part,
\begin{equation}
    (\delta\tilde{\rho}_B(t))_{jk} \equiv \int\diff E\ e^{-i E t} (\delta{\rho_B}(E))_{jk} = 0, \ \forall\ \lvert t \rvert \leq t_\rho~.
    \label{eq:rho_def_fluctuating}
\end{equation}
Up to this point, such a decomposition is always possible. We will additionally take $t_\rho$ to be set by the scale of randomization in the ensemble discussed in Sec.~\ref{sec:pSFFchaotic}, so that the fluctuating part can be identified as the part that is completely randomized in the ensemble. We note that the smooth part may fluctuate between different ensemble realizations, but can not be randomized in the same sense as the fluctuating part as it is roughly constant within an energy window of size $t_\rho^{-1}$. Similarly, we will not require randomization of the correlators of $\delta\rho_B(E)$ between energies further apart than $\sim t_\rho^{-1}$, for which the correlator may have to be nonvanishing to maintain zero Fourier component of the fluctuating part at $t \leq t_\rho$.

To understand the effect of this decomposition in the PSFF, we will first perform a prototype calculation with simpler notation. Consider two functions $f(E)$ and $g(E)$ of a continuous variable $E$, with respective Fourier transforms $\tilde{f}(t)$ and $\tilde{g}(t)$, both of which potentially vary over different realizations of the ensemble. We will eventually associate these functions with (components of) the different parts of the reduced density matrices of the energy eigenstates. Define the quantity,
\begin{align}
    &F(t) = \frac{1}{D^2}\overline{\sum_{j,k} e^{i (E_j-E_k) t} f(E_j) g^\ast(E_k)} \nonumber \\
     &= \frac{1}{D^2}\int\frac{\diff t_l}{2\pi}\int \frac{\diff t_r}{2\pi}\ \overline{\tilde{f}(t_l)\tilde{g}^\ast(t_r) \sum_{j,k} e^{iE_j(t+t_l)-iE_k(t+t_r)}}~.   
%    &= \frac{1}{D^2}\int\frac{\diff t_l}{2\pi}\int \frac{\diff t_r}{2\pi}\ \overline{\lbrace \sum_{j,k} e^{iE_j(t+t_l)-iE_k(t+t_r)}\rbrace\tilde{f}(t_l)\tilde{g}^\ast(t_r)}~.
    \label{eq:fgSFF}
\end{align}
Now, it is convenient to define an ensemble-averaged unequal time SFF $K(t_1,t_2) = D^{-2}\overline{\sum_{j,k} e^{iE_j t_1-iE_k t_2}} $, which reduces to $K(t)$ at equal times $t_1=t_2=t$. The sum of phases $D^{-2}\sum_{j,k} e^{iE_j(t+t_l)-iE_k(t+t_r)}$ in Eq.~\eqref{eq:fgSFF} would fluctuate strongly over different ensemble realizations at large $t_1,t_2$ corresponding to fluctuations of the positions of energy levels, much like the SFF without ensemble averaging~\cite{prange1997sff}; if we assume the ensemble is such that these fluctuations are not correlated with those of $f$ and $g$ (i.e.\ the reduced energy eigenstates), we can perform the ensemble average over the sum of phases independently, allowing us to formally replace it with $K(t+t_l,t+t_r)$,
\begin{equation}
    F(t) = \int\frac{\diff t_l}{2\pi}\int \frac{\diff t_r}{2\pi}\ K(t+t_l,t+t_r)\overline{\tilde{f}(t_l)\tilde{g}^\ast(t_r)}~.
    \label{eq:fgSFF2}
\end{equation}
For instance, in a fully chaotic system as we will soon specialize to, this assumption can be justified by considering the energy eigenstates in an ensemble realization as sufficiently random superpositions of those of another ensemble realization (in the spirit of Refs.~\cite{deutsch1991eth,deutsch2010eigenstates,lu2019renyi, murthy2019eigenstates}), which should then be uncorrelated with the precise positions of the energy levels.

To simplify Eq.~\eqref{eq:fgSFF2} further, we need to know the form of $K(t_1,t_2)$. For mathematical simplicity, we assume (fully chaotic) level statistics in the unitary Wigner-Dyson class. The ensemble-averaged two level correlation function for nearby energy levels $E_j$,$E_k$ (closer than $\sim t_{\rm Th}^{-1}$) in this class takes the universal form \cite{Mehta2004, HaakeBook, Liu2018},
\begin{align}
    &\overline{\delta\left(E+\frac{\omega}{2}-E_j\right) \delta\left(E-\frac{\omega}{2}-E_k\right)} \nonumber \\ &= \Omega^2(E)\left\lbrace1+\frac{\delta(\omega)}{\Omega(E)}-\sinc^2\left[\omega\pi\Omega(E)\right]\right\rbrace~,
    \label{eq:2ptRMT}
\end{align}
where $\Omega(E)$ is the smoothened (continuous and ensemble-averaged) density of states, whose Fourier transform satisfies $\tilde{\Omega}(t\gg t_{\rm Th}) \approx 0$. The ensemble averaged sum over $E_j,E_k$ in the definition of $K(t)$ can then be replaced by an integral weighted by the two level correlation in Eq.~\eqref{eq:2ptRMT}. Using methods analogous to the calculation of $K(t)$ for this correlation function in Ref.~\cite{Liu2018}, we obtain the following late time behavior for $t_1,t_2 \gg t_{\rm Th}$,
\begin{align}
    K(t_1,t_2) = \frac{1}{D^2}\begin{dcases}
   \tilde{\Omega}(\tau_{12}), &T_{12} > 2 \pi \Omega(E)\ \forall E, \\
  \frac{\lvert T_{12}\rvert}{\beta\pi}\tilde{\Theta}_\Omega(\tau_{12}), &T_{12} < 2 \pi \Omega(E)\ \forall E,
  \end{dcases}
  \label{eq:unequalSFF}
\end{align}
where $\beta=2$ for the unitary Wigner-Dyson class, and we have introduced the shorthand symbols $T_{12} = (t_1+t_2)/2$, $\tau_{12} = t_2-t_1$. $\tilde{\Theta}_\Omega(t)$ is the Fourier transform of the unit step function $\Theta(\Omega(E))$, the latter being $1$ where $\Omega(E)>0$ and zero elsewhere. Essentially, the unequal time SFF is generally negligible for (large) unequal times, with a small spread around $t_1=t_2$ determined by the variation of the density of states; as noted earlier, it reduces to the SFF at precisely equal times. We also identify $2\pi\Omega(E)$ with the Heisenberg time $t_H$, assuming that $\Omega(E)$ is at least of the same order of magnitude throughout the spectrum. In the orthogonal and symplectic Wigner-Dyson classes, there are significant corrections (relative to the unitary class) to the form of the equal time SFF $K(t)$ near $t\sim t_H$. But for $t\ll t_H$, virtually the same results hold with $\beta = 1$ for the orthogonal class and $\beta = 4$ for the symplectic class~\cite{Liu2018} (of course, the plateau behavior for $t\gg t_H$ is generally independent of such specifics). Analogously, we expect similar replacements (the appropriate value of $\beta$, and focusing on the $T_{12}\gg t_H$ and $T_{12} \ll t_H$ regimes) to work for the unequal time SFF in Eq.~\eqref{eq:unequalSFF} as well. With this expectation, we write
\begin{equation}
    K(t_1,t_2) = \frac{1}{D^2}\begin{dcases}
   \tilde{\Omega}(\tau_{12}), &T_{12} \gg t_H, \\
  \frac{\lvert T_{12}\rvert}{\beta\pi}\tilde{\Theta}_\Omega(\tau_{12}), &T_{12} \ll t_H,
  \end{dcases}
  \label{eq:unequalSFF2}
\end{equation}
for $t_1,t_2 \gg t_{\rm Th}$ in any Wigner-Dyson symmetry class.

Using the decomposition of $\rho_B(E)$ with these definitions then gives several terms for $K_A(t)$ of the form of Eq.~\eqref{eq:fgSFF}, where $f$ and $g$ independently go over each of $D_B^{-1}$, $\Delta\rho_B$ and $\delta\rho_B$, with an additional trace of the product over the $B$ subspace. Now, we will argue that all cross terms with $f \neq g$ may be taken to vanish. When $f=D_B^{-1}$, the overlap becomes $\tr[B]{fg} = D_B^{-1}\tr[B]{g}$, which is zero when $g=\Delta\rho_B,\delta\rho_B$, which are both traceless. When say, $f$ is $\Delta\rho_B$ and $g$ is $\delta\rho_B$, the cross term vanishes due to the assumption that ensemble averaging randomizes $\delta\rho_B$.

Dropping the cross terms for the above reasons gives the form of Eq.~\eqref{eq:KA_separation} in the main text, $K_A(t) = K(t) + \Delta K_A(t) + \delta K_A(t)$, where $K(t)$ is the full SFF, and
\begin{align}
    \Delta K_A(t) &= \sum_{j,k}\frac{\overline{e^{i (E_j-E_k) t}  \tr[B]{\Delta \rho_B(E_j) \Delta \rho_B(E_k)} }}{DD_A}, \label{eq:DeltaKAdef}\\
    \delta K_A(t) &= \sum_{j,k}\frac{\overline{e^{i (E_j-E_k) t}  \tr[B]{\delta \rho_B(E_j) \delta \rho_B(E_k)} }}{DD_A}.
\end{align}
In the main text, it is argued that $\delta K_A(t\gg t_\rho)$ amounts to a constant shift after ensemble averaging due to the randomization of $\delta\rho_B(E)$. Here, we will complete the evaluation of $\Delta K_A(t)$ using the prototype Eq.~\eqref{eq:fgSFF2} with $f=g=(\Delta\rho_B)_{ab}$ and the expression in Eq.~\eqref{eq:unequalSFF2} with $t_1 = t+t_l$, $t_2 = t+t_r$. As the definition of $\Delta\rho_B$ sets $t_l,t_r < t_\rho$, we have $\lvert T_{12}\rvert = \lvert t\rvert +\sgn(t)(t_l+t_r)/2$ at large times (i.e. $t \gg t_{\rm Th},t_\rho$). For $t\ll t_H$ in this regime, this gives,
\begin{align}
    \Delta K(t : t_{\rm Th},& t_\rho \ll  t \ll t_H) 
    = \frac{1}{D D_A} \int\frac{\diff t_l}{2\pi}\int \frac{\diff t_r}{2\pi}\nonumber\\
   & \left[  \frac{1}{\beta\pi} \left(\lvert t\rvert+\sgn(t)\frac{t_l+t_r}{2}\right) \tilde{\Theta}_\Omega(t_l-t_r) \right. \nonumber \\
    &\left. \overline{\left(\sum_{a,b}\left(\Delta\tilde{\rho}_B(t_l)\right)_{ab}\left(\Delta\tilde{\rho}_B^\ast(t_r)\right)_{ab}\right) }\right].
    \label{eq:DeltaK_intermediate}
\end{align}
The Hermiticity of $\Delta\rho_B$ implies that $\left(\Delta\tilde{\rho}_B(-t)\right)_{ab} = \left(\Delta\tilde{\rho}_B^\ast(t)\right)_{ba}$. Consequently, making the integration variable transformation $t_l \to -t_r$, $t_r \to -t_l$ in Eq.~\eqref{eq:DeltaK_intermediate}, we see that inside the parentheses in the second line the $\lvert t\rvert$ term is unaltered but the $\sgn(t)$ term transforms to its negative, while all factors outside the parentheses remain unaltered. It follows that the contribution from the $\sgn(t)$ term actually evaluates to zero, leaving only a linear ramp term from $\lvert t\rvert$. For $t\gg t_H$, we directly obtain only a plateau contribution. Now, it is straightforward to Fourier transform back to the energy variable $E$, 
\begin{align}
    &\Delta K_A(t\gg t_{\rm Th},t_\rho) \nonumber \\ &= \frac{1}{D D_A}\int\diff E\ \begin{dcases}
    \overline{\Omega(E)\tr[B]{\Delta\rho_B^2(E)}}~,&\ t \gg t_H, \\
    \frac{t}{\beta\pi}\overline{\Theta(\Omega(E))\tr[B]{\Delta\rho_B^2(E)}}~,&\ t \ll t_H .
    \end{dcases}
\end{align}
For ease of interpretation, we can convert $E$ back to a discrete energy variable from its present continuous form via the following correspondence relations for sums over energy levels: $\sum_i \leftrightarrow \int\diff E\ \Omega(E)$ and $\sum_i \Omega^{-1}(E_i) \leftrightarrow \int\diff E\ \Theta(\Omega(E))$, which become equalities on ensemble averaging. Then we get the expression,
\begin{align}
    &\Delta K_A(t\gg t_{\rm Th},t_\rho) \nonumber \\ &= \frac{1}{D D_A}\ \begin{dcases}
    \overline{\sum_i\tr[B]{\Delta\rho_B^2(E_i)}}~,&\ t \gg t_H, \\
    \frac{t}{\beta\pi}\overline{\sum_i\Omega^{-1}(E_i)\tr[B]{\Delta\rho_B^2(E_i)}}~,&\ t \ll t_H .
    \end{dcases}
\end{align}
Together with the expression for the full SFF [$t_1=t_2$ in Eq.~\eqref{eq:unequalSFF2}] and the constant contribution from the fluctuating part, this directly leads to Eq.~\eqref{eq:pSFF_ETH} in the main text.

\section{Constraints from eigenstate thermalization}
\label{app:subETH}

In this Appendix, we discuss the constraints on the spectrum and ensemble averaged PSFF parameters, $\mathcal{P}_B$ (purity of reduced density matrices), $\delta\mathcal{P}_B$ (fluctuating part) and $\Delta\mathcal{P}_B$ (smooth part), as measures of the extent of delocalization and thermalization of energy eigenstates. In App.~\ref{app:subETHconstraints}, we discuss these constraints based on a qualitative picture of subsystem ETH, paying particular attention to thermalization as a distinct phenomenon from delocalization. We justify this qualitative picture in the subsequent section, first in terms of a version of the original conjecture of subsystem ETH~\cite{subETH} for fully delocalized states in App.~\ref{app:subETH.1}, and argue for its extension to eigenstates of arbitrary delocalization in App.~\ref{app:subETH.2}.

\subsection{PSFF as a probe of thermalization and delocalization}
\label{app:subETHconstraints}
We begin with a qualitative discussion of thermalization (in the sense of subsystem ETH) and delocalization. We work in a `physical basis' - one whose basis vectors are close to pure states in most physically accessible (e.g.\ local \cite{nandkishore2015mbl}) subsystems, such as a product basis of qubits. \textit{Thermalization} then corresponds to a significant overlap of the macroscopic features of eigenstates of nearby energies whose individual components are sufficiently random (and therefore, macroscopically similar), whereas non-thermal behavior is seen when nearby eigenstates do not have a large overlap. This is to be distinguished from the extent of \textit{delocalization} of an eigenstate, which is the number of bases states $\ell \leq D$ that it has a significant probability of being found in.

It is useful to introduce an effective dimension \mbox{$D_A^{\rm{eff}} \leq D_A, \ell$} of the Hilbert space of subsystem $A$, corresponding to the typical number of degrees of freedom of subsystem $A$ over which the eigenstate is delocalized within its support in the physical basis. In particular, $D_A^{\rm{eff}} = D_A$ if the eigenstates appear completely delocalized over subsystem $A$, and more generally $D_A^{\rm{eff}}$ is typically larger for larger $D_A$ (up to $\ell$). For instance, $D_A^{\rm{eff}}$ is a monotonically increasing function of $D_A$ when the latter is varied by successively choosing larger subsystems $A$ containing the previous one; additionally, it increases from $D_A^{\rm{eff}} = 1$ for $D_A = 1$, to $D_A^{\rm{eff}} = \ell$ for $D_A = D$. We also use the notation $O(x)$ to mean a non-negative number whose magnitude is at most of the order of magnitude of $x$, to leading order when $x \gg 1$. In particular, we will take $D \gg D_A, D_B \gg 1$.

Assuming that $D_A^{\rm{eff}}$ is typical for $A$ throughout the spectrum, the purity in subsystem $B$ satisfies,
\begin{equation}
\mathcal{P}_B = (D_A^{\rm{eff}}/\ell) + O(D_A^{\rm{eff}}/\ell)+O(1/D_A^{\rm{eff}}),
\label{eq:purityconstraint1}
\end{equation}
subject to $\mathcal{P}_B \gtrsim (D_A^{\rm{eff}}/\ell), (1/D_A^{\rm{eff}})$.
The first two terms are due to the eigenstate being delocalized in subsystem $B$ with effective dimension $(\ell/D_A^{\rm{eff}})$, with the second term containing larger scale variations of its components. We will call this, the `macroscopic' contribution, which grows with $D_A^{\rm{eff}}$. The last term is due to the randomness of the eigenstate components i.e.\ the `microscopic' contribution, which decays with $D_A^{\rm{eff}}$ (and is also typically bounded from below by $(1/D_A^{\rm{eff}})$). Being a linear combination of the macroscopic and microscopic contribution, the purity shows an initial decay with $D_A^{\rm{eff}}$ for small values of the latter, and eventually a growth for larger values of $D_A^{\rm{eff}} \gtrsim \sqrt{\ell}$. Both  $D_A^{\rm{eff}} = 1,\ell$ correspond to pure states with $\mathcal{P}_B = 1$.

The parameters $\delta\mathcal{P}_B$, $\Delta\mathcal{P}_B$ satisfy the following order-of-magnitude inequalities,
\begin{align}
    \delta\mathcal{P}_B &\gtrsim O(1/D_A^{\rm{eff}}), \label{eq:fluctuationconstraint1}\\
    D_B^{-1}+\Delta\mathcal{P}_B &\lesssim (D_A^{\rm{eff}}/\ell)+O(D_A^{\rm{eff}}/\ell) \label{eq:overlapconstraint1}.
\end{align}
The first inequality is the statement that the fluctuating part must include at least the randomness of eigenstate components; the second says that the smooth part or overlap of such eigenstates can at most contain all their macroscopic features. They are also subject to the constraint  $\mathcal{P}_B = D_B^{-1}+\Delta\mathcal{P}_B+\delta\mathcal{P}_B$, which can be interpreted in the present context as follows: the macroscopic contribution to the purity must be distributed in some manner between the smooth and fluctuating parts (with the exception of the maximally mixed part $D_B^{-1}$); the microscopic contribution is however completely contained in the fluctuating part.

According to ETH, the only difference between \textit{thermal eigenstates} of nearby energies is in their microscopic random fluctuations, with all their macroscopic features completely contained in their overlap. This means that the inequalities in Eqs.~\eqref{eq:fluctuationconstraint1} and \eqref{eq:overlapconstraint1} are satisfied as equalities for thermal eigenstates. In particular, $\delta\mathcal{P}_B$ can only decay with increasing $D_A^{\rm{eff}}$ - a fact that is responsible for the nearly identical dynamics of observables in subsystem $B$ (for large $D_A$) in such eigenstates. In contrast, \textit{non-thermal eigenstates} have at least some of the macroscopic contribution included in the fluctuating part, and therefore satisfy Eqs.~\eqref{eq:fluctuationconstraint1} and \eqref{eq:overlapconstraint1} much further from equality. In this case, the macroscopic contribution to the fluctuating part may even show up as a growth of $\delta\mathcal{P}_B$ with $D_A^{\rm{eff}}$ if the latter is sufficiently large (analogous to the behavior of the purity), for choices of subsystems where the incomplete overlap of neighboring eigenstates remains `visible'. At the same time, all eigenstates trivially satisfy $\delta\mathcal{P}_B = \Delta\mathcal{P}_B = 0$ for $D_A = D$.

We conclude that $\mathcal{P}_B$ is a measure of delocalization of eigenstates, while $\delta\mathcal{P}_B$ and $\Delta\mathcal{P}_B$ are probes of thermalization. Setting $\ell = D$ gives the results discussed in the main text for chaotic systems with fully delocalized eigenstates (Sec.~\ref{sec:ethConstraints}). For fully localized systems, $\ell = O(1)$ gives $D_A^{\rm{eff}} = O(1)$, with $\mathcal{P}_B = O(1)$ and $\delta\mathcal{P}_B = O(1) \lesssim (1-D_B^{-1})$, automatically implying a lack of thermalization (Sec.~\ref{sec:pSFFlocalized}). Additionally, the same results hold when the PSFF is defined only over a portion of the spectrum, where the parameters merely become averages over that portion of the spectrum. This suggests that such a filtered~\cite{Gharibyan2018} PSFF can access equivalent information about the properties of a smaller set of eigenstates of interest.

\subsection{Subsystem ETH constraints}

\subsubsection{Fully delocalized eigenstates}
\label{app:subETH.1}
Subsystem ETH~\cite{subETH} is a hypothesis concerning the behavior of energy eigenstates in a chaotic system, applying in its original version to fully delocalized eigenstates. It states that the eigenstates are of such a form as to lead to the thermal behavior of all observables on subsystem $B$, when it is a physically accessible subsystem - in the sense of diagonal and off-diagonal ETH (e.g.\ as presented in the reviews \cite{DAlessio2016,deutsch2018eth}). Denoting the eigenstates by $\lvert E\rangle$, there are two statements of the hypothesis: the diagonal statement stating that the reduced density matrix $\rho_B(E) = \tr[A]{\lvert E\rangle\langle E\rvert}$ is close to some smooth density matrix $P_B(E)$ that does not vary rapidly with energy,  and the off-diagonal statement requiring the reduced transition operators $q_B(E_1,E_2) = \tr[A]{\lvert E_1\rangle\langle E_2\rvert}$ with $E_1\neq E_2$ to be small. We will adapt these statements, in their subsystem dependent version (which doesn't need the restriction $D_B\ll D_A$ to few-body subsystems), for our present context as follows:
\begin{align}
    \tr[B]{\left(\rho_B(E)-P_B(E)\right)^2} &= O(D_A^{-1}), \label{eq:diagSubETH}\\
    \tr[B]{q_B^2(E_1,E_2)} &= O(D_A^{-1}), \label{eq:offdiagSubETH}
\end{align}
where we use the notation $x^2 = x x^\dagger$ for an operator $x$ for simplicity. Eqs.~\eqref{eq:diagSubETH} and \eqref{eq:offdiagSubETH} should be considered leading order constraints on the order of magnitude of these quantities when $D_A,D_B \gg 1$, as noted in the main text. They are also slightly different in some minor technical details from the main statements of Ref.~\cite{subETH}, which we will refer to as the `original conjecture' in this appendix, and we will now comment on these differences.

We replace the density of states $\Omega(E)$ with its $O(D)$ scaling behavior in all subsequent discussions though the original conjecture is stated in terms of $\Omega(E)$. This is justified by assuming an $O(1)$ spectral width for the $D$ energy levels and that $\Omega(E)$ is of a comparable order of magnitude throughout the spectrum (consistent with e.g.\ $t_H = O(D)$ in fully chaotic systems). As the PSFF involves averages over the entire spectrum, it is only this scaling behavior that is of interest to us rather than $\Omega(E)$-dependent variations in smaller regions of the spectrum.

The smallness of $(\rho_B-P_B)$ and $q_B$ are enforced above by requiring the trace of their squares $\tr[B]{x^2}$ (which we will generally call purity) to be $O(D_A^{-1})$. However, the original conjecture is stated in terms of the trace norm $(1/2)\tr[B]{(x^2)^{1/2}}$ restricted to be $O(\sqrt{D_B/D_A})$. As Ref.~\cite{subETH} notes, on account of the inequality $\left\lbrace\tr[B]{(x^2)^{1/2}}\right\rbrace^2 \leq D_B \tr[B]{x^2}$ the constraints in terms of purity would imply the original conjecture but are also slightly stronger, and it is in fact these stronger constraints that they verify numerically. We use the stronger statement because it is more convenient for our purposes, and also because there appears to be no compelling theoretical reason to rule out such stronger statements in general. For instance, Ref.~\cite{subETH} motivates the diagonal statement of the original conjecture in terms of the trace norm based on analogous canonical typicality \cite{gemmer2001quantum,tumulka_CT} constraints for the thermalization of Haar-random superpositions of energy eigenstates derived in Refs.~\cite{PopescuShortWinter, PopescuShortWinter2}; but in the process of the derivation in the latter, constraints in terms of purity similar to Eq.~\eqref{eq:diagSubETH} are also seen to hold. We also note that the purity constraints remain $< O(1)$ for $D_B > D_A$, whereas the corresponding constraints on the trace norm (which cannot be greater than $1$ for differences of density matrices~\cite{NielsenChuang}) are $> O(1)$ and therefore meaningless in this regime. The original conjecture had to restrict the subsystem-dependent form to $D_B<D_A$ (in our notation) for this reason. However, in Sec.~\ref{sec:psffNumerical} of the main text, we find numerical support for the validity of Eqs.~\eqref{eq:diagSubETH} and \eqref{eq:offdiagSubETH} even for $D_B > D_A$.

Finally, we note that the smooth reduced density matrix $P_B(E)$ is not precisely characterized in Ref.~\cite{subETH} - but it is also unnecessary to be too precise in specifying it as Eq.~\eqref{eq:diagSubETH} is only an order-of-magnitude constraint. Here, in analogy with Eq.~\eqref{eq:rho_def_smooth}, we will define $P_B(E)$ to be that part of $\rho_B(E)$ that varies slower than some rate $t_s$,
\begin{equation}
    P_B(E) = \int\frac{\diff\tau}{2\pi}e^{iE\tau} \Theta\left(t_s-\lvert\tau\rvert\right)\int\diff E'\ e^{-iE'\tau}\rho_B(E'),
\end{equation}
effectively amounting to a weighted average of $\rho_B(E)$ over energy windows of size $\sim t_s^{-1}$. We will assume Eq.~\eqref{eq:diagSubETH} is satisfied for any choice of $t_s$ larger than some minimum magnitude $\sim t_{\rm ETH} \ll O(D)$ (intuitively, because the more the smooth part is allowed to fluctuate, the more closely it can approximate $\rho_B(E)$). Then, if our ensemble is such that $t_\rho \gtrsim t_{\rm ETH}$, we can choose $t_s = t_\rho$. This  allows the identification $P_B(E) = D_B^{-1}+\Delta\rho_B(E)$ in the decomposition $\rho_B(E) = D_B^{-1}+\Delta\rho_B(E)+\delta\rho_B(E)$ of Eq.~\eqref{eq:rho_separation}. Eq.~\eqref{eq:diagSubETH} then gives,
\begin{equation}
    \tr[B]{\delta\rho_B^2(E)} = O(D_A^{-1}).
    \label{eq:deltarho_constraint}
\end{equation}
The constraint $\delta\mathcal{P}_B = O(D_A^{-1})$ then follows directly from here.

To similarly obtain a condition from Eq.~\eqref{eq:offdiagSubETH} that applies directly to the PSFF, we note that this equation can be rewritten in terms of  reduced density matrices of the complementary subsystem $A$ as
\begin{equation}
    \tr[A]{\rho_A(E_1)\rho_A(E_2)} = O(D_A^{-1}).
    \label{eq:overlapconstraint}
\end{equation}
On taking the ensemble average, and using the expansion of $\rho_A(E)$ in terms of its smooth and fluctuating parts, the contribution from the fluctuating part $\delta\rho_A(E)$ to the left hand side vanishes due to the randomization assumption in Sec.~\ref{sec:pSFFchaotic}. We are then left with $D_A^{-1}+\overline{\tr[A]{\Delta\rho_A(E_1)\Delta\rho_A(E_2)}}$, in which we can take $E_1-E_2 \ll t_\rho^{-1}$ (e.g.\ neighboring levels) so that the second term is approximately $\overline{\tr[A]{\Delta\rho_A^2(E_1)}}$. From this, we get the smooth purity constraint $\Delta\mathcal{P}_A = O(D_A^{-1})$ on taking the appropriate spectrum averages. In the context of $\Delta\mathcal{P}_B$ (and $\widetilde{\Delta\mathcal{P}}_B$) in the main text, these purities are evaluated in subsystem $B$ rather than $A$, and the corresponding constraints are therefore consequences of off-diagonal subsystem ETH, Eq.~\eqref{eq:offdiagSubETH}, applied to subsystem $A$ instead of $B$.

\subsubsection{Extension to partially delocalized eigenstates}
\label{app:subETH.2}

We begin with a complementary approach to that of the previous subsection, to argue that the purity based expressions of subsystem ETH should generally hold for chaotic systems with fully delocalized eigenstates. Consider requiring each matrix element of $\rho_B(E)$ to differ from the corresponding matrix element of $P_B(E)$ only by a small amount $O(\sqrt{D_A}/D)$, as a stronger diagonal statement that implies Eq.~\eqref{eq:diagSubETH} (a weaker, $D_A$-independent version of such a statement is also considered in Ref.~\cite{subETH}). To justify this constraint, we consider the following situation. Let $\lvert E_1\rangle$ and $\lvert E_2\rangle$ be two `typical' nearby eigenstates that are completely delocalized over the $D$ basis vectors (in some `physical' product basis of subsystems $A$ and $B$) with random (real or complex) phases. Their density operators $\rho(E_1) = \lvert E_1\rangle\langle E_1\rvert$ and $\rho(E_2) = \lvert E_2\rangle\langle E_2\rvert$ have matrix elements of the schematic form
\begin{equation}
    \rho_{ab}(E) \sim O(D^{-1})e^{i\phi_a-i\phi_b}.
\end{equation}
The difference $\rho(E_1)-\rho(E_2)$, after a partial trace over $A$, can be taken to represent the fluctuations of $\rho_B(E)$ around $P_B(E)$. Given our above assumptions on the eigenstates, the matrix elements of $\rho(E_1)-\rho(E_2)$ are typically $\sim O(D^{-1})$ in magnitude with random signs or phases (i.e.\ with zero $2$-point correlation, which crucially requires even large-scale non-uniformities in the magnitudes to agree up to random fluctuations).
The sum of $D_A$ such matrix elements in the partial trace over subsystem $A$ then has magnitude $O(\sqrt{D_A}/D)$, justifying the above constraint. Similarly, the operator $q(E_1,E_2) = \lvert E_1\rangle \langle E_2\rvert$ for such eigenstates has $O(D^{-1})$ matrix elements with random phases, giving $O(\sqrt{D_A}/D)$ matrix elements after the partial trace and therefore the off-diagonal statement Eq.~\eqref{eq:offdiagSubETH}. Such a picture of random energy projector matrix elements of comparable magnitudes is reminiscent of Berry's conjecture for chaotic wavefunctions \cite{Berry1977} (as well as other related statements e.g.\ Refs.~\cite{deutsch1991eth,deutsch2010eigenstates,lu2019renyi, murthy2019eigenstates}), which has been interpreted as the origin of eigenstate thermalization in chaotic systems \cite{deutsch1991eth,srednicki1994eth}.

Using an analogous argument for eigenstates that are not necessarily delocalized over all $D$ basis vectors, we can clearly highlight the difference between delocalization and thermalization, and the distinct information contained in the overall purities as opposed to the smooth and fluctuating parts of the reduced density matrices. For this purpose, consider an eigenstate $\lvert E_1\rangle$ that is randomly (but not necessarily uniformly) distributed only over a set of $\sim \ell \leq D$ `physical' basis vectors, with negligible support outside this set. Its density matrix $\rho(E_1)$ then has an $\ell \times \ell$ block (after suitably permuting rows and columns) of non-vanishing elements each of typical magnitude $O(\ell^{-1})$, and all elements outside this block may be taken to vanish. As always, all the diagonal elements are strictly non-negative and add to $1$, while the independent off-diagonal elements could have arbitrary signs or phases (which are typically random). Thus, we have the schematic form,
\begin{equation}
    \rho_{ab}(E_1) \sim \left[O(\ell^{-1}) e^{i\phi_a-i\phi_b}\right]\overline{\Theta}(1 \leq \lbrace p_a,p_b\rbrace \lesssim \ell),
\end{equation}
where $\overline{\Theta}(x) = 1$ if $x$ is true and $0$ otherwise, and $p_k$ denotes the index corresponding to $k$ after a permutation $p$ of rows/columns.

The behavior of $\rho(E_1)$ under a partial trace depends on the choice of the subsystem $A$. We will choose subsystems which can be traced out by factorizing the chosen basis (which means the basis states are pure states within the subsystem). This identifies a class of subsystems which are sensitive to the specific extent of delocalization $\ell$ of eigenstates; in a more general basis in the Hilbert space, the eigenstates may appear delocalized by an arbitrary extent, including fully localized in the energy eigenbasis and generically fully delocalized ($\ell = D$) in a Haar random basis according to canonical typicality~\cite{tumulka_CT,PopescuShortWinter}. An equivalent, more physically motivated viewpoint is that the extent of delocalization of eigenstates $\ell$ should be determined by their minimum such delocalization in bases comprised of nearly pure states (e.g.\ a product basis) in most physically accessible subsystems - so that a small subset of eigenstates may be treated as if they each have $\ell$ independent random components (neglecting the global constraint of orthonormality) under a (sufficiently small) partial trace.

For convenience, we first consider the case where the eigenstate looks fully delocalized in subsystem $A$ within its support on the physical basis - in other words, the partial trace over $A$ does not mix the zero and nonzero elements of $\rho(E_1)$. In this appendix, we will call such a subsystem $A$ an \textit{unbiased} subsystem (from the point of view of the eigenstate of interest). Then $\rho_B(E_1)$ has an $\sim (\ell/D_A)\times (\ell/D_A)$ non-vanishing block with non-negative diagonal elements of magnitude $O(D_A/\ell)$, and off-diagonal elements of typical magnitude $O(\sqrt{D_A}/\ell)$ in the case of an eigenstate with random phases (as long as the partial trace combines several basis vectors where the eigenstate has comparable magnitudes). Now, we can evaluate the purity $\tr[B]{\rho_B^2(E_1)}$, which sees a net contribution of $O(D_A/\ell)$ from the diagonal elements and $O(D_A^{-1})$ from the off-diagonal elements. Additionally, normalization requires that the diagonal elements must add up to $1$, therefore the sum of their squares is greater than or equal to $\sim \ell/D_A$ - the inverse of the number of diagonal elements. Their contribution to the purity can then be written in a more descriptive form as $[(D_A/\ell)+O(D_A/\ell)]$, giving
\begin{equation}
    \tr[B]{\rho_B^2(E_1)} = (D_A/\ell)+O(D_A/\ell)+O(D_A^{-1}).
\label{eq:purity_partlylocal}
\end{equation}
Thus, we can extract information about the extent of delocalization, $\ell$, by looking at the subsystem size dependence of the purity. We note that the purity can also be written as $\tr[A]{\rho_A^2(E_1)}$ from the viewpoint of subsystem $A$ giving an additional lower bound of $D_A^{-1}$, which is mostly contained in the $O(D_A^{-1})$ term for $D_A \gg 1$ (as the diagonal contribution to purity from $\rho_A(E_1)$ is primarily due to contributions from the off-diagonal elements of $\rho_B(E_1)$).

A nearby eigenstate $\lvert E_2\rangle$ that is also distributed only across $\ell$ basis vectors (but not necessarily the same ones or in the same way as $\lvert E_1\rangle$) again shows a subsystem purity of the form of Eq.~\eqref{eq:purity_partlylocal}. The two eigenstates thermalize if their reduced density matrices do not differ significantly, in small enough subsystems that trace out a lot of the independent eigenstate components. This would be the case if these two eigenstates are distributed across roughly the same $\ell$ basis vectors in a largely similar manner (up to random fluctuations). 
From this point of view, subsystem ETH is a qualitative identification of the thermalization of a set of otherwise random-looking eigenstates with the extent of their `overlap' within subsystems, rather than merely with entanglement as represented by their individual purities (the latter being the canonical typicality approach that is only sufficient for fully, uniformly delocalized random eigenstates as in App.~\ref{app:pSFF-RMT}).

We now consider two illustrative extreme cases of fully overlapping (thermal) and fully non-overlapping (non-thermal) eigenstates. In both cases, we will be interested in $\tr[B]{(\rho_B(E_1)-\rho_B(E_2))^2}$ as a representative of the size of the fluctuating part $[\rho_B(E)-P_B(E)]$ of reduced energy eigenstates in subsystem $B$, as well as the (real-valued) overlap $\tr[B]{\rho_B(E_1)\rho_B(E_2)}$ which is equal to the norm of off-diagonal operators $\tr[A]{q_B(E_1,E_2)q_B(E_2,E_1)}$ in subsystem $A$. These are complementary quantities, being related to the subsystem purities of the individual eigenstates via
\begin{align}
&\tr[B]{(\rho_B(E_1)-\rho_B(E_2))^2} + 2 \tr[B]{\rho_B(E_1)\rho_B(E_2)} \nonumber \\
&\hphantom{\rho_B(E_1)} = \tr[B]{\rho_B^2(E_1)}+\tr[B]{\rho_B^2(E_2)}.
\label{eq:PurityOverlapComplementarity}
\end{align}
This relation quantifies the identification of thermalization with overlap.

\begin{itemize}
    \item \textit{Thermal eigenstates}: If $\lvert E_1\rangle$ and $\lvert E_2\rangle$ are distributed in a similar manner across the same basis vectors, then again has $(\rho(E_1)-\rho(E_2))$ an $\ell\times\ell$ block structure, with random $O(\ell^{-1})$ off-diagonal elements within the block. However, the diagonal elements, being differences of random $O(\ell^{-1})$ non-negative numbers, also have at most $O(\ell^{-1})$ magnitudes with random signs (if large scale non-uniformities match), and largely cancel each other out in a partial trace. After the partial trace, all matrix elements of $(\rho_B(E_1)-\rho_B(E_2))$ in an $\sim (\ell/D_A)\times (\ell/D_A)$ non-vanishing block are therefore only $O(\sqrt{D_A}/\ell)$ in magnitude, and we have
    \begin{equation}
        \tr[B]{(\rho_B(E_1)-\rho_B(E_2))^2} = O(D_A^{-1}),
        \label{eq:thermalfluctuations}
    \end{equation}
    consistent with diagonal ETH Eq.~\eqref{eq:diagSubETH} in subsystem $B$. For the overlap $\tr[B]{\rho_B(E_1)\rho_B(E_2)}$, the positivity and normalization of the diagonal matrix elements of each reduced density matrix ensures that their contribution is of the form $[(D_A/\ell) + O(D_A/\ell)]$. The products of the off-diagonal matrix elements add up with random phases, leading to a negligible $O(\ell^{-1})$ contribution. We therefore have
    \begin{equation}
        \tr[B]{\rho_B(E_1)\rho_B(E_2)} = (D_A/\ell) + O(D_A/\ell),
        \label{eq:thermaloverlaps}
    \end{equation}
    which is the analog of off-diagonal ETH, Eq.~\eqref{eq:offdiagSubETH}, for subsystem $A$.
    \item \textit{Non-thermal eigenstates}: In the non-thermal case, $\lvert E_1\rangle$ and $\lvert E_2\rangle$ are distributed in completely different ways, and the diagonal elements of $(\rho(E_1)-\rho(E_2))$ do not have completely random signs among elements with comparable magnitudes. Consequently, there is no longer a significant cancellation of the diagonal elements in a partial trace for a general choice of $A$. The fluctuating part $\tr[B]{(\rho_B(E_1)-\rho_B(E_2))^2}$ is then typically much larger than $O(D_A^{-1})$ with some $O(D_A/\ell)$ contribution, and the overlap is correspondingly smaller. In the extreme case of the two eigenstates being distributed across completely different basis vectors, $(\rho(E_1)-\rho(E_2))$ has two different $\ell\times\ell$ blocks, and the reduced difference in subsystem $B$ also has the structure of two independent blocks. We then obtain behavior analogous to the subsystem purities,
    \begin{align}
        &\tr[B]{(\rho_B(E_1)- \rho_B(E_2))^2} \nonumber \\
        &\sim 2\left[(D_A/\ell)+O(D_A/\ell)+O(D_A^{-1})\right],
        \label{eq:nonthermalfluctuations}
    \end{align}
    while the overlap for this case vanishes entirely,
    \begin{equation}
        \tr[B]{\rho_B(E_1)\rho_B(E_2)} = 0.
    \end{equation}
\end{itemize}
We note that these trends hold only for $D_A < \ell$, due to the assumption on subsystem $A$. The reduced energy eigenstates in subsystem $B$ are pure basis states when $D_A = \ell$, and behave accordingly on a further partial trace.

The fluctuations in reduced energy eigenstates and their overlaps therefore contain information about eigenstate thermalization that is not visible to the purity alone, which is merely an indicator of eigenstate delocalization. We also see that, at least for `typical' eigenstates, diagonal subsystem ETH should be understood (in a coarse, order of magnitude sense) as a lower bound relation, while off-diagonal subsystem ETH is a complementary upper bound relation, related through Eq.~\eqref{eq:PurityOverlapComplementarity} to each other and the subsystem purities. In place of Eqs.~\eqref{eq:diagSubETH}, and \eqref{eq:offdiagSubETH}, we can therefore write the more general relations for partially delocalized eigenstates,
\begin{align}
    \tr[B]{\left(\rho_B(E)-P_B(E)\right)^2} &\gtrsim O(D_A^{-1}), \label{eq:diagSubETHgeneral}\\
    \tr[A]{q_A^2(E_1,E_2)} &\lesssim (D_A/\ell)+O(D_A/\ell), \label{eq:offdiagSubETHgeneral}
\end{align}
when $A$ is an unbiased subsystem, with $D_A \leq \ell$. Both bounds are saturated by thermal eigenstates.

For greater completeness of the present discussion, we should account for a more typical choice of subsystem $A$ - one that would mix the zero and non-zero elements of these eigenstate reduced density matrices on performing the partial trace over $A$. We will consider such a typical subsystem to have an effective dimension $D_A^{\rm{eff}} \leq D_A$, corresponding to the typical number of non-zero density matrix elements added together in the partial trace. This can be thought of as a generalization of the notion of effective dimension, discussed for the case of infinite dimensional Hilbert spaces in Ref.~\cite{subETH}. We ignore the more complicated case where the number of matrix elements added together is not approximately uniform for all nonzero matrix elements (and therefore, no effective subsystem dimension exists), with the belief that it would not significantly alter our qualitative conclusions. When the effective dimension does exist, all the above conclusions hold for any system but with $D_A$ replaced by the smaller quantity $D_A^{\rm{eff}}$. As an aside, Eqs.~\eqref{eq:diagSubETHgeneral} and ~\eqref{eq:offdiagSubETHgeneral} continue to hold even without this replacement, but are then not necessarily saturated by thermal eigenstates unless $A$ is an unbiased subsystem.

As a simple example, if subsystem $A$ is unbiased with respect to a set of eigenstates of interest, then its complementary subsystem $B$ has effective dimension $D_B^{\rm{eff}} = \ell/D_A$ (note that $B$ is not unbiased). Using this, we can finally write off-diagonal ETH for subsystem $B$ and diagonal ETH for subsystem $A$ as follows,
\begin{align}
    \tr[B]{q_B^2(E_1,E_2)} &\lesssim D_A^{-1}+O(D_A^{-1}), \label{eq:offdiagSubETHgeneral'} \\
    \tr[A]{\left(\rho_A(E)-P_A(E)\right)^2} &\gtrsim O(D_A/\ell). \label{eq:diagSubETHgeneral'}
\end{align}
More generally, expressing Eqs.~\eqref{eq:diagSubETHgeneral},~\eqref{eq:offdiagSubETHgeneral} in terms of $D_A^{\rm{eff}}$ gives the constraints discussed in App.~\ref{app:subETHconstraints}.

\section{Time-reversal symmetric Floquet thermalization}
\label{app:coe-u2}
In this section, we consider another Floquet model of periodically kicked spin-$1/2$ system. We consider one period of duration $\tau$ to be, \begin{equation}
V_2= e^{-i H^{(x)}\tau/2} e^{-i H^{(y)}\tau/2} ~.
\label{eq:Floquet-illustration-V2}
\end{equation}
At multiples $t=n\tau$ ($n\in \mathbb N$) the  time evolution of this model is governed by the Floquet time evolution operator $T(t=n\tau)=V_2^n$. The Hamiltonians $H^{(x, y)}$ are  $H^{(x,y)}= J\sum_{i=1}^{N-1} \sigma^{(x,y)}_i \sigma^{(x,y)}_{i+1} + \sum_{i=1}^N h_i^{(y,z)} \sigma^{(y,z)}_i$ where the local disorder potentials $h_i^{(y,z)}$  are uniformly and independently sampled from $[-J, J]$. We fix the driving frequency to $\tau^{-1}= J/2$.  With these parameters, the time evolution operator $V_2^n$ is known to have COE eigenvalue statistics after a few initial kicks \cite{Regnault2016}. 
\begin{figure}[!ht]
\includegraphics[width=\linewidth]{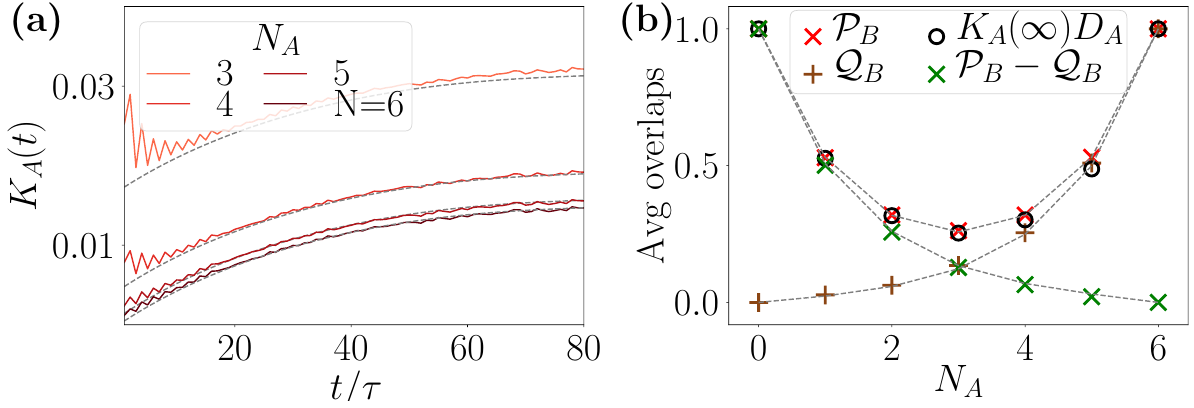}
\caption{\textit{Results for the Floquet $V_2$ model:} (a) The SFF and PSFF are presented for $N=6$, $N_A=3,~4$ and $5$ in red colors. In gray, we plot the same quantities in a COE model. (b) The plateau value $K(\infty)$ times the subsystem dimension $D_A$ is plotted in black circles and matches with the averaged purity $\mathcal{P}_B$ plotted with red crosses. The average overlap $\mathcal{Q}_B$ and the difference $\mathcal{P}_B-\mathcal{Q}_B$ are presented in brown and green respectively. We observe a  match with the respective quantities in COE plotted in gray, signaling the same averaged eigenvalue and eigenvector statistics in COE and $V_2$.}
\label{fig:ManyBodyV2}
\end{figure}

We present numerically obtained SFF and PSFF for a total system size of $N=6$ and subsystem sizes $N_A=3, 4$ and $5$ in Fig.~\ref{fig:ManyBodyV2}(a) and observe a shift in the PSFF in addition to the characteristic chaotic features; the ramp and the plateau. We plot with gray lines the corresponding $K_A(t)$ in a COE model where the analytic forms have been exactly calculated [see Eq.\ \eqref{eq:app_pSFF_O}], and observe a good match between $V_2$ and COE. The  match between the statistics of COE and $V_2$ can further be explored using the second-order moments of the reduced density of eigenstates.

In Fig.~\ref{fig:ManyBodyV2}(b) we present the overlaps  $\mathcal{P}_B$ and $\mathcal{Q}_B$ as functions of subsystem size $N_A$. We plot numerically obtained $K_A(\infty)\DA$ in black circles, and the average purity $\mathcal{P}_B$ with red crosses and note a good match between the two. Note that unlike the SFF in the unitary class where the transition to plateau at the Heisenberg time is sharp, the transition to a constant plateau takes a long time in an orthogonal model. This is why we observe slight differences in the numerically calculated plateau value and the purity in Fig.~\ref{fig:ManyBodyV2}(b). The average overlap $\mathcal{Q}_B$ and the difference $\mathcal{P}_B-\mathcal{Q}_B$ are plotted in brown and green circles respectively and match with those in COE. Therefore, similar to RMT models and the model $V_3$ (in Sec.~\ref{sec:psffNumerical}), the Floquet dynamics $V_2$ also has $\Delta\mathcal{P}_B=0$. Thus, numerically we confirm that the reduced densities in the Floquet system $V_2$ also thermalize to infinite temperature and the ramp is governed entirely by maximally mixed part of $\rho_B(E)$. From the plots of $\mathcal{P}_B-\mathcal{Q}_B$ (in green) we conclude that the constant term added to the SFF is  $\sim 1/D_A^2$, as in the RMT models. 

\section{Additional numerical results for Ising Hamiltonian dynamics}
\label{app:Ising}
In the main text, we considered the Ising Hamiltonian in Eq.~\eqref{eq:Hamiltonian} as an example of local many-body models. In this Appendix, we provide some supporting data which were used in the main section. We first begin with analyzing the interesting set of parameters for which we observe chaotic and localized phases in the Hamiltonian model. In Sec.~\ref{sec:psffAnalytical}, we  derived the orders for the purity and overlaps of the reduced density matrices on the basis of ETH, and presented them  numerically in Sec.~\ref{sec:psffNumerical}. Here, we provide  some additional information on the  numerics used to extract the orders for the Hamiltonian model. In the last subsection, we numerically cross-check the shift in the PSFF data with the shift $\delta \mathcal{P}_B$ calculated using Eq.\ \eqref{eq:PQ-numerics}, where in the latter we directly use the reduced densities of eigenstates.
\subsection{Chaotic and MBL regimes in Ising Hamiltonian}
\label{app:IsingGapRatio}
We explain our choice of parameters in the Ising Hamiltonian Eq.\ \eqref{eq:Hamiltonian}. The Hamiltonian contains $ZZ$ interactions with strength $J$ and the range of interactions is given by $\alpha$. It has a transverse field with strength $J$ and a longitudinal local random disordered field with strength $W$. Our interests lie in the parameters such that the Hamiltonian dynamics is either in the chaotic phase or in the localized phase. For this purpose, we analyze the energy level statistics, using the \textit{adjacent energy gap ratio}. 
\begin{figure}[!ht]
\label{density-plot}
\includegraphics[scale=0.16]{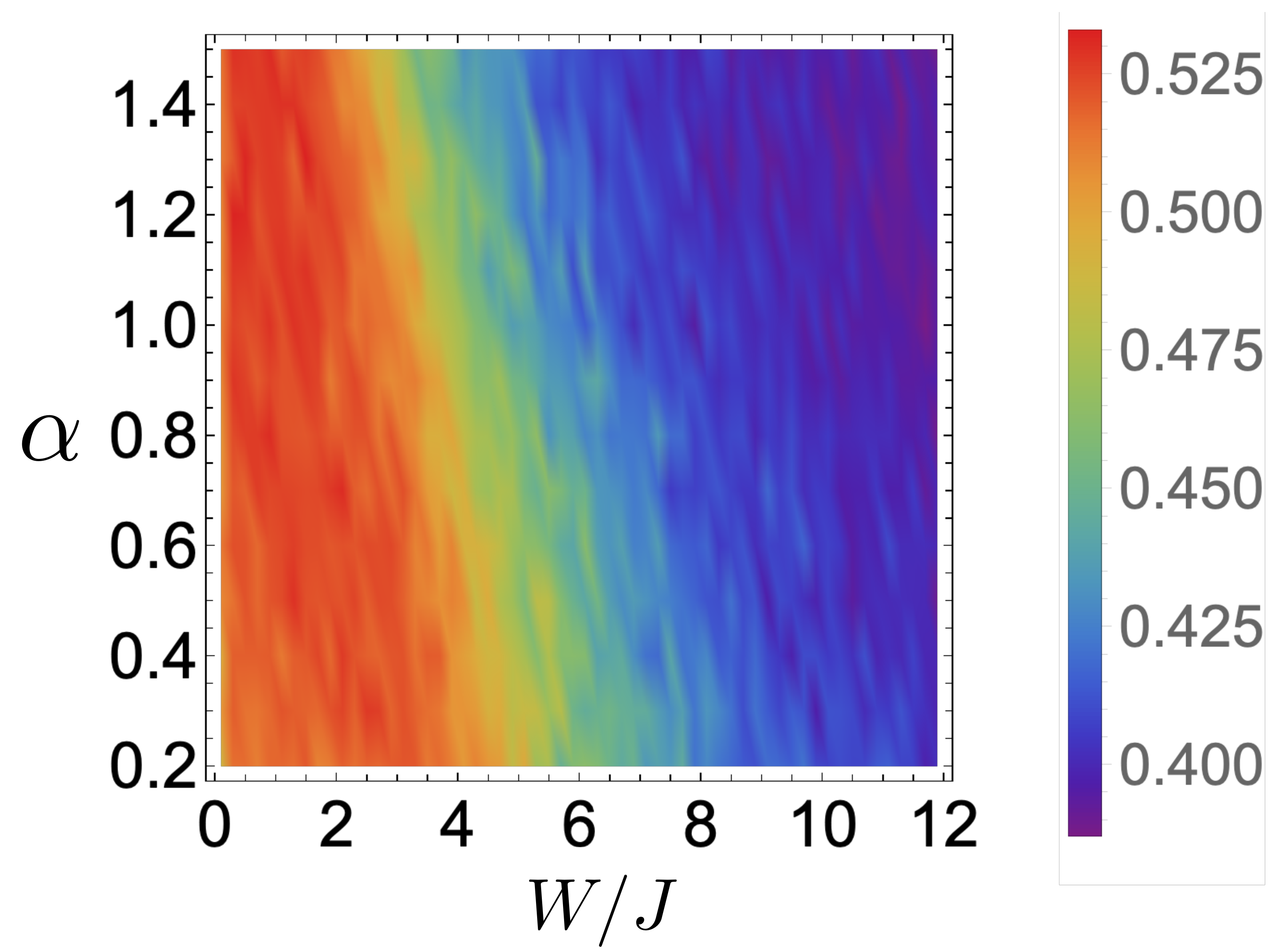}
\caption{\textit{Density plot for the mean adjacent gap ratio $\langle r_m\rangle$.} The behavior of the mean adjacent gap ratio $\langle r_m\rangle$ as a function of disorder strength $W/J$ and range of interactions  $\alpha$ is presented. We notice the presence of chaotic ($\langle r_m\rangle \sim 0.53$) and localized  ($\langle r_m\rangle \sim 0.39$) phases for a wide range of $\alpha$. We have worked with $\alpha=1.2$, $W=J$ (chaotic), and $W=10J$ (MBL).} 
\label{fig:densityplot} 
\end{figure}
From the sorted energy eigenvalues $E_1< E_2<\dots <E_\D$ we compute the energy gaps $\Delta E_m=E_{m+1}-E_{m}$. Then we find,  the adjacent energy gap ratio 
\be
r_m = \frac{\text{min}(\Delta E_m, \Delta E_{m+1})}{\text{max}(\Delta E_m, \Delta E_{m+1})}~.
\ee
Integrable systems are characterized by a mean ratio of $\langle r_m \rangle \approx 0.39$ whereas the chaotic systems with time-reversal symmetry, obeying GOE Wigner-Dyson energy level statistics have a mean $\langle r_m\rangle\approx 0.53$. We use this mean value of $r_m$ to choose the parameters for the chaotic and localized phase in our Hamiltonian model. In a density plot of the mean $\langle r_m\rangle$ as a function of $W/J$ and $\alpha$ in Fig.~\ref{fig:densityplot}, we notice that the chaotic and localized phases exist for both the short ($\alpha>1$) and the long ($\alpha<1$) range of interactions. In this work, to discuss the two phases, we have chosen the parameters to be $\alpha=1.2$, $W/J=1$ (chaotic) and $W/J=10$ (localized). 

\subsection{Orders of magnitude of $\Delta \mathcal{P}_B$ and $\delta \mathcal{P}_B$}
\label{app:IsingOrders}
The subsystem ETH specifies the orders of magnitude for $\Delta\mathcal{P}_B$ and $\delta \mathcal{P}_B$ to be $O(1/\DB)$ and $O(1/\DA)$ respectively for the chaotic models. For the localized models, which are known to not satisfy ETH, we concluded that the shift coefficient, $\delta \mathcal{P}_B\gg O(1/\DA)$. 
%Using Eq.\ \eqref{eq:PQ-numerics}, 
We note that these orders for the chaotic phase, expressed in terms of $\mathcal{P}_B$ and $\mathcal{Q}_B$ as deviations from the RMT prediction, amount to, 
\begin{align}
 D_B\Delta \mathcal P_B &= \DB\mathcal{Q}_B  -1\approx O(1) \nonumber \\
\text {and}\qquad D_A\delta \mathcal P_B-1 &= \DA({\mathcal{P}_B} - {\mathcal{Q}_B} )-1\approx O(1) .
\label{eq:orders}
\end{align}
\begin{figure}[ht]
\includegraphics[width=\linewidth]{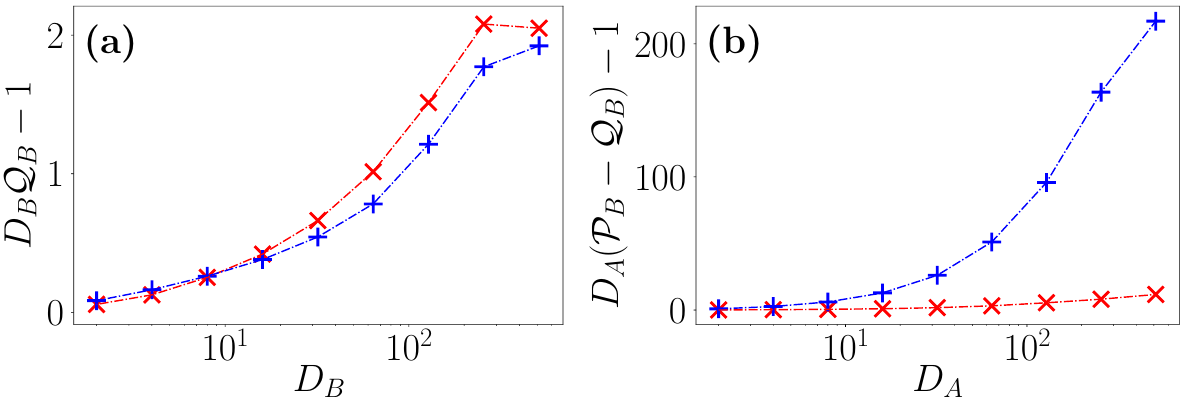}
\caption{\textit{Validating the orders.} We test the validity of Eq.~\eqref{eq:orders} for the chaotic and MBL phases of the Hamiltonian. As expected, the chaotic phase, in red, satisfies the predicted order and the MBL phase, in blue, violates it (in the right plot for $\mathcal{P}_B-\mathcal{Q}_B$).}
    \label{fig:orders}
\end{figure}

In Fig.~\ref{fig:orders}, we plot these quantities for the Hamiltonian model,
Eq.\ \eqref{eq:Hamiltonian}, for a total of $N=10$ qubits. We find that the chaotic phase $W=J$ (in red) satisfies the ETH results whereas for the localized phase $W=10J$ (in blue), the shift coefficient $\mathcal{P}_B-\mathcal{Q}_B \gg 1/\DA$, as predicted in the Sec.~\ref{sec:psffAnalytical}.

\subsection{Comparison of the PSFF shift and $\delta\mathcal{P}_B$}
\label{app:IsingShift}

Here, we numerically verify the prediction of Eq.~\eqref{eq:pSFF_ETH} for the constant late time shift of the PSFF in the chaotic phase - namely, that the shift is given by $\delta\mathcal{P}_B/D_A$ in the ramp region. For this purpose, we subtract the full SFF from the PSFF at some time $t_0$ in the linear ramp region, satisfying $t_{\rm Th}, t_\rho \ll t_0 \ll t_H$, which gives
\begin{equation}
    K_A(t_0) - K(t_0) = \frac{\delta\mathcal{P}_B}{D_A} + t_0\left(\frac{\gamma}{\beta\pi D^2}D_B\widetilde{\Delta\mathcal{P}}_B\right).
    \label{eq:kat0}
\end{equation}
This difference has two contributions - the first term is the additive shift which we are presently interested in, but the second term is due to the differing slopes of the linear ramp, from the excess purity of the smooth part of the reduced density matrix. We will now argue that it is reasonable to take
\begin{equation}
 K_A(t_0) - K(t_0) \approx \frac{\delta\mathcal{P}_B}{D_A},
 \label{eq:shiftapprox}
\end{equation}
for our purposes. In Eq.~\eqref{eq:kat0}, by subsystem ETH, the former is $O(D_A^{-2})$ while the latter is $\sim O(D^{-1}(t_0/t_H))$ (taking $\gamma \sim O(1)$, consistent with $t_H \sim O(D)$). The second term is therefore negligible if  $t_0/t_H \ll O(D_B/D_A)$. This is immediately satisfied for any $t_0$ in the linear ramp region if $D_A < D_B$; conversely, for a given choice of $t_0$, $D_A$ can be as large as $\sim \sqrt{Dt_H/t_0}$ while maintaining the validity of Eq.~\eqref{eq:shiftapprox}. As $t_0 \ll t_H$ in general, we expect Eq.~\eqref{eq:shiftapprox} to be a reasonable approximation for a range of values of $D_A > \sqrt{D}$ as well. A minor additional effect that improves this approximation is that for large $D_A$, the coefficient $D_B\widetilde{\Delta\mathcal{P}}_B$ of the second term would be small, though still $O(1)$, from Fig.~\ref{fig:orders} (as $D_B\widetilde{\Delta\mathcal{P}}_B \sim [D_B Q_B-1]$).

On the basis of Eq.~\eqref{eq:shiftapprox} and the relation $\delta\mathcal{P}_B = \mathcal{P}_B-\mathcal{Q}_B$ from Eq.~\eqref{eq:PQ-numerics}, we compare $D_A(K_A(t_0) - K(t_0))$ for some suitably chosen $t_0$ to $\mathcal{P}_B-\mathcal{Q}_B$ in Fig.~\ref{fig:pSFF_shift_verification} and observe good agreement, especially for smaller $D_A$ as expected. We note that this agreement is much closer than, for instance, the difference between $\mathcal{P}_B-\mathcal{Q}_B$ for the Hamiltonian system and the corresponding RMT prediction in Fig.~\ref{fig:manybodyHamil}(d), which is considerable evidence that the origin of the shift is indeed the randomization of the fluctuating part of the reduced energy eigenstates, as discussed in Sec.~\ref{sec:pSFFchaotic}.

\begin{figure}[ht]
\includegraphics[width=\linewidth]{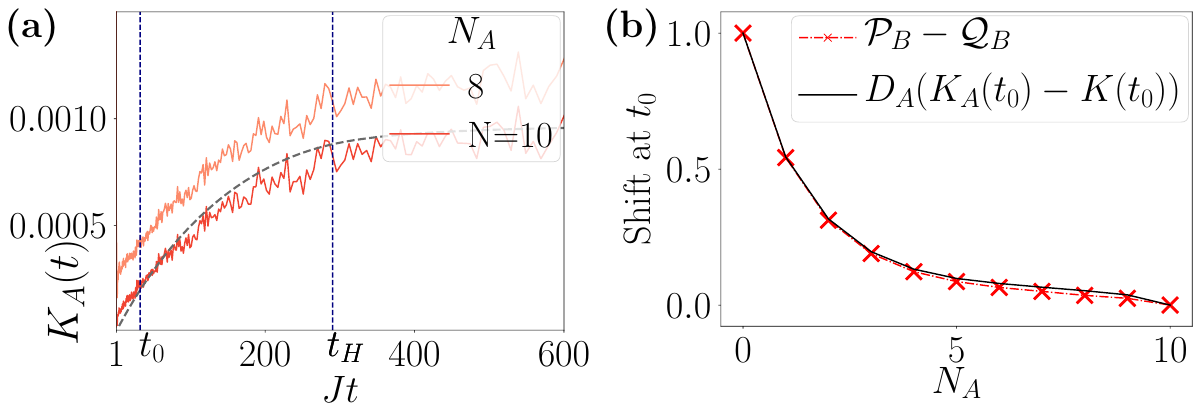}
\caption{\textit{Results on the shift.} (a) Linear-linear plot of the PSFF for $N_A = 8$ and SFF, reflecting the choice of comparison time $t_0$ in the ramp region, with $Jt_0 = 25$; the Heisenberg time $t_H$ (marking the onset of the plateau) and the corresponding GOE SFF (dashed curve) are also shown. (b) Comparison of the scaled PSFF shift $D_A(K_A(t_0) - K(t_0))$ with the predicted value $\delta\mathcal{P}_B = \mathcal{P}_B-\mathcal{Q}_B$ for different subsystem sizes. The total system size is $N=10$ in both plots.}
\label{fig:pSFF_shift_verification}
\end{figure}

\section{Derivation of the measurement protocol}
\label{app:MeasProt}
In this appendix, we show that our measurement protocol indeed allows us to measure the PSFF and SFF. We generalize the proof for the SFF in the main text (Sec.~\ref{sec:proof_mt}) to the PSFF $K_A(t)$ and provide additional mathematical details. Our aim is to prove that $\widehat{K_A(t)}$, as defined in Eq.~\eqref{eq:psffmeas}, is an unbiased estimator of $K_A(t)$, i.e.\
$\mathbb E \left[ \widehat{K_A(t)}\right]= K_A(t)$ where $\mathbb E$ comprises the expectation value taken over the ensemble of  time evolution operators (the disorder average) $
\mathbb{E}_T$, the local random unitaries  $
\mathbb{E}_U$  and projective measurements  $
\mathbb{E}_\text{QM}$.  Note that we  use in the appendix  $
\mathbb{E}_T$ in place of $\overline{\cdots}$ to denote the expectation over an ensemble of time evolution operators. 

We consider a quantum system $\mathcal{S}$ consisting of $N$ qubits  with Hilbert space $\mathcal{H}=(\mathbb{C}^2)^{\otimes N}$ of dimension $\D=2^N$, and $A\subseteq \mathcal{S}$ of $N_A$ qubits with dimension $\DA=2^{N_A}$. From the $r=1,\dots, M$ (single-shot) repetitions of  our protocol with outcome bitstrings $\{ \mathbf{s}^{(r)}\}_{r=1,\dots,M}$ we obtain the  estimator $\widehat{K_A(t)}$ as defined in Eq.~\eqref{eq:psffmeas}.  For simplicity of notation we drop the time argument in the following in this appendix. 

As a first step, it is most convenient to reformulate Eq.~\eqref{eq:psffmeas}  as an average over $r=1,\dots,M$ single shot estimates $\hat{o}^{(r)}$  of an observable $O=\bigotimes_i O_i$  with $O_i= \ket{0}\bra{0}-1/2 \ket{1}\bra{1}$ for $i\in A$ and $O_i= \mathbb{1}_i$ for $i \notin A$, 
\be
\widehat{K_A}=   \frac{1}{M} \sum_{r=1}^{M} \;  (-2)^{- |\mathbf{s}_A^{(r)}|} \equiv  \frac{1}{M} \sum_{r=1}^{M} \hat{o}^{(r)} \; .
\label{eq:sffmeasSM}
\ee
Secondly, we note that the outcome bitstrings $\{ \mathbf{s}^{(r)}\}_{r=1,\dots,M}$ of the $M$ repetitions of our measurement protocol are identically and independently distributed by construction: For each experimental run, a set of local unitaries $\{u_i^{r}\}_{i=1,\dots, N}$ and time-evolution operator $T$ is independently sampled and applied according to the experimental sequence shown in Fig.~\ref{fig:figure2}. Lastly, a single-shot computational basis measurement is taken.
We thus have 
\begin{align}
    \mathbb{E}\left[\widehat{K_A(t)}\right]= \mathbb{E}\left[\hat{o}^{(r)}\right] \; .
\end{align}
for an arbitrary $r\in \{1,\dots,M\}$. We drop the superscript $(r)$ in the following.

\begin{figure}[t]
    \centering
    \includegraphics[width=0.48\textwidth]{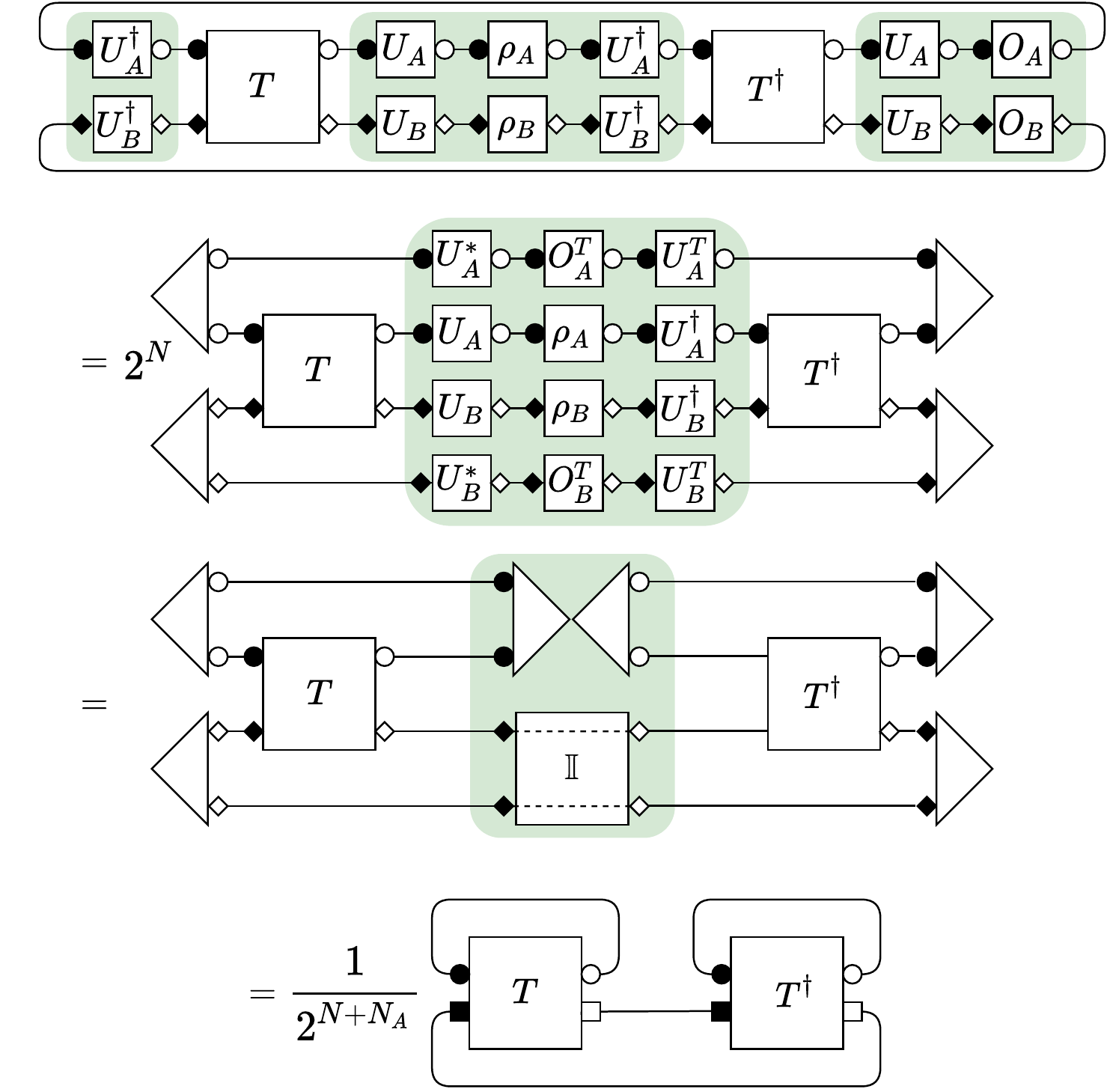}
    \caption{\textit{Diagrammatic proof of the measurement protocol}. We use the diagramatic notation and calculus developed in Ref.~\cite{Collins2010} (see also Ref.~\cite{Elben2019}). With the definitions of the text, we have $U_A=\bigotimes_{i \in A} u_i$, $\rho_A =  \bigotimes_{i \in A} =\rho_i$, $O_A=\bigotimes_{i \in A} O_i$, and accordingly for subsystem $B$. From second to third line, we use the 2-design identities of the local random unitaries $u_i$ (see also Eq.~\eqref{eq:twirl}  and Refs.~\cite{Collins2010,Elben2019}). }
    \label{fig:diagrams_prot}
\end{figure} 

To make progress, we evaluate the expectation  $\mathbb{E}\left[\hat{o} \right]$ over the ensemble of  time evolution operators (the disorder average) $
\mathbb{E}_T$, the local random unitaries $
\mathbb{E}_U$  and projective measurements (the quantum mechanical expectation value) $\mathbb{E}_{\text{QM}}$ step-by-step using the law of total expectation
\begin{align}
    \mathbb{E}\left[\hat{o} \right] = \mathbb{E}_{T} \left[\; \mathbb{E}_{U} \left[\; \mathbb{E}_{\text{QM}} \left[ \hat{o} | U, T\right]\, | \, T   \, \right]\; \right] .
\end{align}
 Here, $\mathbb{E}_{\text{QM}} \left[\hat{o} | U, T\right]$ denotes the quantum mechanical expectation value of the single-shot estimator $\hat{o}$ for a fixed unitary $U$ and a fixed time-evolution operator $T$. By definition, this is just the quantum expectation value of the observable $O$ in the output state $\rho_f= U^\dagger T U \rho_0 U^\dagger T^\dagger U$ of our protocol,
\begin{align}
\mathbb{E}_{\text{QM}} \left[\hat{o} | U, T\right] = \langle O \rangle_{\rho_f} = \tr[]{O U^\dagger T U \rho_0 U^\dagger T^\dagger U} .
\end{align}
The key part of the proof is the evaluation of the expectation value  over the local random unitaries $U=\bigotimes_i u_i$ for a fixed time evolution operator $T$,
\begin{align}
 & \mathbb{E}_{U} \left[\; \mathbb{E}_{\text{QM}} \left[ \hat{o} | U, T\right]\, | \, T   \, \right]\nonumber \\  & \qquad=  \mathbb{E}_{U} \left[ \;\tr[]{O U^\dagger T U \rho_0 U^\dagger T^\dagger U} \, | \,T \, \right].
 \end{align} As also visualized in Fig.~\ref{fig:diagrams_prot}, this requires several steps:
We first  rewrite
\begin{align}
&\mathbb{E}_{U} \left[\; \tr[]{O U^\dagger T U \rho_0 U^\dagger T^\dagger U} \, | \, T\, \right]= \label{eq:euu} \\
& 2^{N} \bra{\Phi_N^+} ( 1\otimes T)  \,  \mathbb{E}_{U} \left[U^*  O^T U^T \otimes U\rho_0 U^\dagger  \right] (1\otimes T^\dagger)  \nonumber \ket{\Phi_N^+} 
\end{align}
as an expectation value of two `virtual copies' of qubit $i$, using the identity $ \tr{AB} = 2^N \braket{\Phi_N^+| A^T \otimes B |\Phi_N^+} $ for any two operators $A$ and $B$. Here, we have defined   $\ket{\Phi_N^+}=\bigotimes_i\ket{\Phi_i^+}$ as the tensor product of Bell states $\ket{\Phi_i^+}=2^{-1/2}(\ket{00}+\ket{11})$  on the doubled Hilbert space $\mathbb{C}^{2} \otimes \mathbb{C}^{2}$.
We  now use the independence of the local random unitaries $u_i$ to completely factorize the expectation value $\mathbb{E}_{U} $ over the local random unitaries $U=\bigotimes_i u_i$
\begin{align}
&\mathbb{E}_{U} \left[U^*  O^T U^T \otimes U\rho_0 U^\dagger  \right]\nonumber \\
&\qquad =\bigotimes_{i=1}^N \int\! \text{d} u_i\,   \left( u_i^* \otimes u_i \right) \left(  O_i^T  \otimes \rho_i \right) \left(  u_i^T \otimes u_i^\dagger  \right)  ~,
\end{align}
where $\int\! \text{d} u_i$ denotes the Haar integral over the unitary group $U(2)$.  As shown in Refs.~\cite{Collins2006, Watrous2018, Elben2020_SPT} [and also follows directly from  Eq.~\eqref{eq:haar_unitary}], we can use the $2$-design identities of the applied local random unitaries $u_i$ to evaluate the Haar integral. We find,
\begin{widetext}
\begin{align}
& \int\! \text{d} u_i\,    \left( u_i^* \otimes u_i \right) \left(  O_i^T  \otimes \rho_i \right) \left(  u_i^T \otimes u_i^\dagger  \right)   \nonumber\\
&= \frac{1}{3} \left( \vphantom{\frac{1}{2}}4 \ketbra{\Phi^+_i}{\Phi^+_i}  \tr{\ketbra{\Phi^+_i}{\Phi^+_i} O^T_i \otimes \rho_i} + \mathbb{1}_i \tr{ O^T_i \otimes \rho_i}  - \mathbb{1}_i \tr{\ketbra{\Phi^+_i}{\Phi^+_i}  O^T_i \otimes \rho_i} -\ketbra{\Phi^+_i}{\Phi^+_i}  \tr{ O^T_i \otimes \rho_i} \right)\nonumber\\
&= \frac{1}{3} \left(\vphantom{\frac{1}{d}} 2 \ketbra{\Phi^+_i}{\Phi^+_i}  \tr{O_i \rho_i} +   \mathbb{1}_i \tr{O_i  }  -\frac{1}{2} \mathbb{1}_i \tr{ O_i \rho_i} - \ketbra{\Phi^+_i}{\Phi^+_i}  \tr{  O_i } \right) \nonumber\\
&= \frac{1}{2}
\begin{cases}
\ketbra{\Phi^+_i}{\Phi^+_i} & i \in A \\
 \mathbb{1}_i & i \notin A 
\end{cases}~.
\label{eq:twirl}
\end{align}
\end{widetext}
To arrive at the last line,  we  used $\rho_i=\ket 0\bra 0$ and that  $O_i= \ket{0}\bra{0}-1/2 \ket{1}\bra{1}$ for $i\in A$ and $O_i= \mathbb{1}_i$ for $i \notin A$.
Inserting this into Eq.~\eqref{eq:euu}, we  find 
\begin{align}
    &\mathbb{E}_{U} \left[\; \mathbb{E}_{\text{QM}} \left[ \hat{o} | U, T\right]\, | \, T   \, \right]  \nonumber \\
    & \quad =    \bra{\Phi_N^+}( 1\otimes T)    \left(\bigotimes_{i \in A} \ketbra{\Phi^+_i}{\Phi^+_i} \right) (1\otimes T^\dagger) \ket{\Phi_N^+} \nonumber\\
    &\quad = 2^{-(N+N_A)} \tr{\tr[A]{T  } \tr[A]{T^\dagger}}~.
\end{align}
Taking finally the ensemble (disorder) average over time evolution operators, we find
\begin{align}
      \mathbb{E}\left[\hat{o} \right] &= \mathbb{E}_{T} \left[\; \mathbb{E}_{U} \left[\; \mathbb{E}_{\text{QM}} \left[ \hat{o} | U, T\right]\, | \, T   \, \right]\; \right] \nonumber \\
      &=\mathbb{E}_{T} \left[\; 2^{-(N+N_A)} \tr{\tr[A]{T  } \tr[A]{T^\dagger}}\right]\nonumber  \\
      &= K_A(t)~.
\end{align}
Thus, we see that, $\widehat{K_A}$ is an unbiased estimator of $K_A(t)$.

\section{Statistical error analysis and required measurement budget}
\label{app:staterrors}
As described in the main text [Eq.\ \eqref{eq:psffmeas}], we obtain an estimate of the PSFF $K_A$ (for notational simplicity we drop the time argument in this appendix) from $r=1,\dots, M$ (single-shot) repetitions of  our protocol with outcome bitstrings $\mathbf{s}^{(r)}$ via
\be
\widehat{K_A}= \frac{1}{M}  \sum_{r=1}^{M} \;  (-2)^{ - |\mathbf{s}_A^{(r)}|}  =   \frac{1}{M} \sum_{r=1}^{M} \hat{o}^{(r)} \;. 
\label{eq:sffmeas_app}
\ee
Here,  $\hat{o}^{(r)}$ is a single shot estimate of an observable $O=\bigotimes_i O_i$  with $O_i= \ket{0}\bra{0}-1/2 \ket{1}\bra{1}$ for $i\in A$ and $O_i= \mathbb{1}_i$ for $i \notin A$, as defined in App.~\ref{app:MeasProt}.  We have shown in App.~\ref{app:MeasProt} that $\widehat{K_A}$ is an unbiased estimator of the PSFF $K_A$, i.e.~ that \ $\mathbb E \left[ \widehat{K_A}\right]= K_A$ with the expectation value  taken over the ensemble of time evolution operators (the disorder average) $\mathbb{E}_T$, the local random unitaries $\mathbb E_U$ and projective measurements $\mathbb{E}_{\text{QM}}$. The  statistical error of $\widehat{K_A}$, and its convergence to $K_A$ is controlled by its variance
\begin{align}
\text{Var}\left[\widehat{K_A}\right] &= \frac{ 1}{M} \text{Var}\left[{\hat{o}^{(r)}}\right] \label{eq:varka} 
\end{align}
for any $r=1,\dots,M$. Here, we used that the individual single shot estimates $\hat{o}^{(r)}$ are statistically independent and identically distributed by construction. We drop the superscript $(r)$ in the following. We can evaluate Eq.~\eqref{eq:varka} using the law of total variance \cite{bowsher2012} 
\begin{align}
   \var{\hat{o}} \nonumber=\,& \expect[T]{\,\expect[U]{\,\var[\text{QM}]{\hat{o}|T,U}\,|\,T\,}\,} \nonumber\\
   &+\expect[T]{\,\var[U]{\,\expect[\text{QM}]{\hat{o}|T,U}\,|\,T\,}\,}\nonumber\\
   &+\var[T]{\,\expect[U]{\expect[\text{QM}]{\,\hat{o}|T,U}\,|\,T\,}\,} \nonumber\\
      =\,& \expect[T]{\,\expect[U]{\,\expect[\text{QM}]{\hat{o}^2|T,U}\,|\,T\,}\,} \nonumber\\
   &-\expect[T]{\,\expect[U]{\expect[\text{QM}]{\,\hat{o}|T,U}\,|\,T\,}\,}^2.
   \label{eq:var1}
\end{align}
To arrive at the second expression, we employed the definition of the (conditional) variance $\var[]{X|Y}=\expect[]{X^2|Y}-\expect[]{X|Y}^2$ for any two random variables $X,Y$ and used then that various terms cancel out.
As shown in App.~\ref{app:MeasProt}, the last term in Eq.~\eqref{eq:var1} simply yields
\begin{align}
    \expect[T]{\,\expect[U]{\expect[\text{QM}]{\,\hat{o}|T,U}\,|\,T\,}\,}^2=K_A^2 .
\end{align}
We thus concentrate on the first term in Eq.~\eqref{eq:var1}. The quantum mechanical expectation value  $\expect[\text{QM}]{\,\hat{o}^2 |T,U}$ of the squared single shot estimate $\hat{o}$ evaluates, for fixed $T$ and $U$, to  
\begin{align}
    \expect[\text{QM}]{\,\hat{o}^2 |T,U} = \langle O^2 \rangle_{\rho_f}= \tr[]{O^2 U^\dagger T U \rho_0 U^\dagger T^\dagger U} \label{eq:var1b}
\end{align}
Next, we evaluate the average over local random unitaries. 
With $O$ replaced by $O^2$, we follow the calculation presented in App. \ref{app:MeasProt}: we first rewrite Eq.~\eqref{eq:var1b} as an expectation value on two copies
\begin{align}
&2^{-N}\expect[U]{\; \tr[]{O^2 U^\dagger T U \rho U^\dagger T^\dagger U} \, | \, T\, } \label{eq:var2} \\
&\! = \! \bra{\Phi_N^+} ( 1\otimes T)    \expect[U]{U^*  (O^T)^2 U^T \otimes U\rho_0 U^\dagger } (1\otimes T^\dagger)  \nonumber \ket{\Phi_N^+} .
\end{align}
Factorizing the average  over local random unitaries, we find 
\begin{align}
&\expect[U]{U^*  (O^T)^2 U^T \otimes U\rho_0 U^\dagger  }\nonumber \\
&\qquad =\bigotimes_{i=1}^N \int\! \text{d} u_i\,   \left( u_i^*  (O_i^T)^2 u_i^T \otimes u_i \rho_i u_i^\dagger  \right) 
\end{align}
with $\int\! \text{d} u_i\, $ the Haar integral over the unitary group $U(2)$.
Using Eq.~\eqref{eq:twirl}, for $O_i\rightarrow O_i^2$, we find
\begin{align}
&\int\! \text{d} u_i\,   \left( u_i^*  (O_i^T)^2 u_i^T \otimes u_i \rho_i u_i^\dagger  \right) \nonumber \\ &\qquad \qquad = 
\frac{1}{2} \begin{cases} 
 \ketbra{\Phi^+_i}{\Phi^+_i}/2 + \mathbb{1}_i/2 & i \in A \\
 \mathbb{1}_i & i \notin A
\end{cases} .
\end{align}
Inserting this into Eq.\ \eqref{eq:var2}, we obtain
\begin{align}
&\expect[U]{\; \tr[]{O^2 U^\dagger T U \rho U^\dagger T^\dagger U} \, | \, T\, } \nonumber  \\ &= 2^{-(N+N_A)} \sum_{B\subseteq A} 2^{-N_B} \tr{\tr[B]{T  } \tr[B]{T^\dagger}}.
\end{align}
Taking  the  ensemble (disorder) average over time evolution operators $T$, we get
\begin{align}
\expect[T]{\,\expect[U]{\,\expect[\text{QM}]{\hat{o}^2|T,U}\,|\,T\,}\,}\! = \!
2^{-N_A} \sum_{B\subseteq A} \SFF[B] .
\end{align}
This finally yields
\begin{align}
    \text{Var}\left[\widehat{K_A}\right] &= \frac{ \text{Var}\left[\hat{o}\right] }{M} =2^{-N_A} \sum_{B\subseteq A} \SFF[B] -K_A^2.
\end{align}

Given the variance $\text{Var}\left[\widehat{K_A}\right]=\text{Var}[\hat{o}]/M $ of our estimator, Chebyshev's inequality asserts that 
\begin{align}
    \text{Prob}\left[| \widehat{K_A}-K_A| \geq \epsilon \right] \leq \frac{\text{Var}\left[\widehat{K_A}\right]}{\epsilon^2} =\frac{\text{Var}\left[\hat{o}\right]}{M \epsilon^2}
\end{align}
for any $\epsilon>0$.  This allows to rigorously obtain an estimate for the required number of measurements $M$ to achieve a certain relative error (for a similar treatment see e.g.\ Refs. \cite{Huang2020,Elben2020_Mixed}).
\newtheorem{prop}{Proposition}
\begin{prop}
Consider a subsystem $A\subseteq \mathcal{S}$ with $N_A\leq N$ qubits. Our aim is to estimate the PSFF $K_A$ using the estimator $\widehat{K_A}$ defined in Eq.\ \eqref{eq:psff}.  Then, for any $\epsilon,\delta>0$, a total of  
\begin{align}
    M\geq  \frac{\tilde V_A}{\delta \epsilon^2}
\end{align}
experimental runs (single shot estimates) suffice to ensure that the relative error of the estimator $\widehat{K_A}$ obeys  $|\widehat{K_A}/K_A -1|\leq \epsilon $ with probability $1-\delta$. Here, we defined the rescaled variance
\begin{align}
    \tilde V_A = \frac{1}{\SFF[A]^2} \left(2^{-N_A} \sum_{B\subseteq A} \SFF[B]  - \SFF[A]^2 \right)
\end{align}
where the sum extends over all subsystems $B\subseteq A$ containing $N_B\leq N_A$ qubits. 
\end{prop}
For the random matrix ensembles considered in App.~\ref{app:pSFF-RMT}, we can determine $\tilde V_A$ explicitly. In particular, for CUE dynamics $T(t=n\tau)$ with $V$ from CUE, we find at the point of weakest signal, i.e.\ the dip time $t=\tau$, $\tilde{V}_A=10^{N_A}-1$.

\bibliography{SFF-Randomized-Measurements,RM_bibliography}

\end{document}